\documentclass[journal,onecolumn]{IEEEtran}
\ifCLASSINFOpdf
\else
\fi
\ifCLASSOPTIONcompsoc
  \usepackage[caption=false,font=normalsize,labelfont=sf,textfont=sf]{subfig}
\else
  \usepackage[caption=false,font=footnotesize]{subfig}
\fi

\usepackage{amsthm}

\usepackage{amssymb}
\usepackage{mathtools}
\usepackage{cases}
\usepackage{autobreak}

\usepackage{booktabs}
\usepackage{multirow}

\usepackage{color}
\usepackage{hyperref}
\usepackage{cite}

\usepackage{amsmath}
\usepackage{extarrows}

\usepackage{comment}

\usepackage{algorithmic}
\usepackage{algorithm}

\newcommand{\E}{\mathbb{E}}
\newcommand{\Prob}{\mathbb{P}}
\newcommand{\Iunit}{\mathrm{i}}

\newcommand{\cov}{\mathrm{Cov}}
\newcommand{\Var}{\mathrm{Var}}

\newcommand{\diag}{\mathrm{diag}}

\newcommand{\LR}{{\bold{R}^{\frac{1}{2}}}}
\newcommand{\RT}{{\bold{T}^{\frac{1}{2}}}}
\newcommand{\MS}{{\bold{S}^{\frac{1}{2}}}}

\newcommand{\FR}{{\bold{R}}}
\newcommand{\FT}{{\bold{T}}}
\newcommand{\FS}{{\bold{S}}}
\newcommand{\BQ}{{\bold{Q}}}
\newcommand{\BH}{{\bold{H}}}
\newcommand{\BI}{{\bold{I}}}
\newcommand{\BX}{{\bold{X}}}
\newcommand{\BE}{{\bold{E}}}
\newcommand{\BZ}{{\bold{Z}}}
\newcommand{\BY}{{\bold{Y}}}

\newcommand{\BG}{{\bold{G}}}

\newcommand{\BA}{{\bold{A}}}
\newcommand{\BB}{{\bold{B}}}
\newcommand{\BC}{{\bold{C}}}
\newcommand{\BW}{{\bold{W}}}
\newcommand{\BM}{{\bold{M}}}
\newcommand{\BF}{{\bold{F}}}

\newcommand{\BO}{{\mathcal{O}}}

\DeclareMathOperator{\Tr}{Tr}

\newcommand{\RNum}[1]{\uppercase\expandafter{\romannumeral #1\relax}}

\newtheorem{remark}{Remark}
\newtheorem{theorem}{Theorem}
\newtheorem{lemma}{Lemma}
\newtheorem{proposition}{Proposition}

\newcommand\numberthis{\addtocounter{equation}{1}\tag{\theequation}}

\allowdisplaybreaks[4]
\begin{document}
%
\title{Asymptotic Mutual Information Analysis for Double-scattering MIMO Channels: A New Approach by Gaussian Tools}
%
%
%

\author{Xin~Zhang,~{\textit{Graduate Student Member,~IEEE}} and 
        Shenghui~Song,~\IEEEmembership{Senior Member,~IEEE}
\thanks{The authors are with the Department of Electronic and Computer Engineering, The Hong Kong University of Science and Technology, Hong Kong
(e-mail: xzhangfe@connect.ust.hk; eeshsong@ust.hk).}}
\maketitle

\begin{abstract}
The asymptotic mutual information (MI) analysis for multiple-input multiple-output (MIMO) systems over double-scattering channels has achieved engaging results, but the convergence rates of the mean, variance, and the distribution of the MI are not yet available in the literature. In this paper, by utilizing the large random matrix theory (RMT), we give a central limit theory (CLT) for the MI and derive the closed-form approximation for the mean and the variance by a new approach---Gaussian tools. The convergence rates of the mean, variance, and the characteristic function are proved to be $\BO(\frac{1}{N})$ for the first time, where $N$ is the number of receive antennas. Furthermore, the impact of the number of effective scatterers on the mean and variance was investigated in the moderate-to-high SNR regime with some interesting physical insights. The proposed evaluation framework can be utilized for the asymptotic performance analysis of other systems over double-scattering channels. 
\end{abstract}


%
\IEEEpeerreviewmaketitle

\section{Introduction}
%
%
%
%
Multiple-input multiple-output (MIMO) has been widely considered as a promising technology to  achieve high spectral efficiency and wide coverage in wireless communications. The full rank Rayleigh/Rician channel model is widely assumed by most works, which is unfortunately not applicable to systems affected significantly by the spatial correlation and rank-deciency, e.g., the \textit{keyhole/pinhole effect}. To overcome this issue, Gesbert~\textit{et al.} proposed the \textit{double-scattering channel}~\cite{gesbert2002outdoor}, which considers both the spatial correlation and the geometry of the propagation environment. The double-scattering model, represented as the product of two correlated Rayleigh channel matrices, was shown to be able to properly characterize the rank deficiency caused by the keyhole/pinhole effect~\cite{almers2006keyhole}. The double-scattering model not only takes the conventional models as special cases but also reflects the nature of the cascaded channel in the intelligent reflecting surface (IRS) aided MIMO systems~\cite{basar2021indoor}. However, the non-Gaussianity of the double-scattering channel, due to the product of two correlated complex Gaussian matrices, makes the information-theoretic analysis very challenging.

The performance of double-scattering channel has been investigated with innovative results. In~\cite{shin2008mimo}, Shin~\textit{et al.} gave the diversity order with respect to the numbers of antennas and scatterers. In~\cite{shin2003capacity}, with the finite-dimension random matrix theory (RMT), an upper bound of the ergodic capacity and an exact expression of the capacity over single keyhole channel were given. The outage probability of double-scattering channels was investigated in~\cite{shi2021outage}. Although finite RMT could provide an exact expression, the computational complexity is very high due to the involvement of special functions, determinant, and integrals. To overcome this issue, the large system approach has been utilized to avoid heavy Monte Carlo simulations and obtain more physical insights. For example, the large RMT method has been used in~\cite{hachem2008new} to prove the central limit theory (CLT) of the mutual information (MI) for single-hop correlated Rayleigh channels  and to obtain a strikingly simple approximation for the mean and variance. Similarly, the CLT for the MI of other single-hop channels has been investigated in~\cite{hachem2008clt,hachem2012clt,bao2015asymptotic,hu2019central} by large RMT. Furthermore, it has been shown that the large system approximation works well even for small dimensions~\cite{hoydis2011asymptotic}. As a result, large RMT has also been utilized for the analysis of double-scattering channels with appealing  results for both first-order and second-order analysis.

\textbf{The first-order analysis}: In~\cite{muller2002random}, the uncorrelated double-scattering channel was investigated by the Stieljies transform in RMT and the ergodic mutual information (EMI) per antenna was given by numerical simulations. The explicit deterministic approximation of the EMI for general, correlated, double-scattering channels was first presented in~\cite{hoydis2011iterative,hoydis2011asymptotic} but only the convergence of the normalized EMI (EMI per antenna) was guaranteed. The approximation can degenerate to the special case of IRS-aided MIMO systems~\cite{zhang2021large,moustakas2021capacity} obtained by the replica method. The sum-rate analysis for multi-user systems in double-scattering channels was given in~\cite{ye2022sum,ye2020asymptotic,kammoun2019asymptotic}. The above methods, though capable of evaluating the quantities with respect to large random matrices, failed to provide a convergence rate. Furthermore, a rigorous proof regarding the convergence of the EMI without normalization is not available in the literature. In this paper, we will use a new approach, i.e., Gaussian tools, to show the convergence of the EMI and determine the convergence rate.

\textbf{The second-order analysis}: In~\cite{zheng2016asymptotic}, Zheng~\textit{et al.} investigated the CLT of the MI over Rayleigh product channel (independent and identically distributed (i.i.d.) double-scattering channel) by a free probability approach and gave the closed-form approximations for the mean and variance when the numbers of transmit and receiver antennas are equal. The CLT of the MI for general, correlated, double-scattering channels was first presented in~\cite{zhang2022outage} by the resolvent method in large RMT. Specifically, a CLT was set up for the MI by splitting the randomness into two parts and evaluating the martingale differences. The closed-form deterministic approximation of the variance was also given. However, the following questions remain unanswered: How fast does the distribution of the MI converge to a Gaussian distribution? What are the convergence rates of the mean and the variance when the number of antennas increases? These questions are related to the performance of the Gaussian approximation (CLT) for the MI and the approximation error for the mean and variance of the MI. In other words, although the resolvent method is powerful in proving the Gaussianity, it does not explicitly indicate the convergence speed.

\textbf{Motivation}: There have been some interesting works on the analysis of the convergence rate. The most related work is~\cite{hachem2008new}, in which RMT is utilized to set up a CLT for the MI of correlated Rayleigh MIMO channels and the convergence rates of the mean and characteristic function are shown to be $\BO(\frac{1}{N})$ by assuming that the number of antennas at the transceivers go to infinity with the same pace. The convergence rate indicates the accuracy of the deterministic approximation for the performance metric and has been investigated for various Gaussian fading channels in~\cite{dumont2010capacity,dupuy2011capacity,zhang2013capacity,asgharimoghaddam2020capacity} by Gaussian tools. It also provides the guidance on deriving the deterministic approximation, e.g., how many terms do we need to keep to construct an asymptotically tight approximation. In~\cite{artigue2011precoder}, the convergence rate obtained in~\cite{hachem2008new} was used to derive an asymptotically tight approximation for the ergodic rate of a MIMO system with the minimum mean square error (MMSE) receiver. Furthermore, the convergence rate was also utilized in setting up a CLT for the mutual information density (MID), which is important for the finite blocklength analysis in MIMO Rayleigh channels~\cite{hoydis2015second}. Despite the importance of the convergence rate, a related analysis for double-scattering MIMO channels is not available. In this paper, we will use the Gaussian tools to recover the result in~\cite{zhang2022outage} and determine the convergence rate, which can be used for the finite blocklength analysis of Rayleigh-product MIMO channels~\cite{zhang2022second}.

\textbf{Gaussian tools}: In this paper, we would like to recommend the Gaussian tools as a powerful method for the analysis of double-scattering channels. The Gaussian tools consists of two parts, namely,~\textit{Nash-Poincar{\'e}} inequality and the~\textit{integration by parts formula}, which have been widely used in RMT~\cite{lytova2009central,gotze2017distribution}. With~{Nash-Poincar{\'e}} inequality, we can evaluate the upper bound of the approximation errors, while by the integration by parts formula, we can convert the expectation of a function of Gaussian random variables to the expectation for the derivative of the function. The Gaussian tools has also been used in the asymptotic analysis of the single-hop MIMO channels. In~\cite{hachem2008new}, a CLT for the MI of general correlated Rayleigh channels was given where the convergence rates of the mean and the characteristic function of the MI were also determined. In~\cite{dumont2010capacity}, the deterministic approximation of the EMI over Rician channels was given with a convergence rate. Other applications of the Gaussian tools for single-hop channels can be found in~\cite{dupuy2011capacity,artigue2011precoder}. A remarkable advantage of the Gaussian tools is its ability to provide a convergence rate. To the best of the authors' knowledge, there has not been any results regarding the convergence rate of the MI in double-scattering channels. In this paper, we will set up a framework for the asymptotic analysis of double-scattering channels based on the Gaussian tools so that the asymptotic analysis is guaranteed with the convergence rate.

\subsection{Contributions}
The main contributions of this paper are summarized as follows.


1) We prove that the convergence rates for the mean and variance of the MI are both $\BO(\frac{1}{N})$, with $N$ being the number of receive antennas, which are not available in the literature. We also show that the convergence rate of the characteristic function is $\BO(\frac{1}{N})$. Based on the CLT, we give an approximation of the outage probability and validate its accuracy by numerical results. The framework involved in the proof can be used for the finite blocklength analysis~\cite{zhang2022second}.


2) Based on the CLT, we give the moderate-to-high SNR approximation for the mean and variance of the MI over Rayleigh product channel and analyze the impact of the rank deficiency. Note that only the case with equal number of antennas at the transceivers is considered in~\cite{zhang2022outage}. In this paper, the cases with general number of antennas at the transceivers are derived, which is a complement for~\cite{zhang2022outage}. The dominating term of the mean and variance for the MI in double-scattering channels is identified, which reveals interesting physical insights regarding the relation between double-scattering and Rayleigh channels.

\subsection{Paper Outline and Notations}
The rest of this paper is organized as follows. In Section~\ref{sec_model}, the system model and problem formulation are introduced. In Section~\ref{sec_res}, the main results regarding the first-order (EMI) and second-order (CLT) analysis with the convergence rates are presented, together with a moderate-to-high SNR approximation, which shows the impact of the number of effective scatterers. Section~\ref{sec_method} introduces the Gaussian tools and some preliminary results. The rigorous proof of the first-order and second-order analysis is given in Section~\ref{sec_pro_fir} and Section~\ref{sec_pro_clt}, respectively. In Section~\ref{sec_simu}, numerical experiments are performed to validate the Gaussianity of the MI and illustrate the approximation performance. Section~\ref{sec_con} concludes the paper.

\textit{Notations:} Matrices and vectors are denoted by the bold, upper case letters and bold, lower case
letters, respectively. $\E x$ represents the expectation of $x$ and $\mathbb{P}(\cdot)$ denotes the
probability measure. $(\cdot)^{*}$ represents the conjugate of a complex number. $\mathbb{C}^{N}$ and $\mathbb{C}^{M\times N}$ denote the space of
$N$-dimensional vectors and the space of $M$-by-$N$ matrices, respectively. $\bold{A}^{H}$ represents the conjugate transpose of $\bold{A}$, and the $(i,j)$-th entry of $\bold{A}$ is denoted by $A_{ij}$. $\|\BA \|$ represents the spectral norm of $\bold{A}$. $\Tr\BA$ refers to the trace of $\BA$ if it is square. $\bold{I}_{N}$ denotes the identity matrix of size $N$. $\phi(x)$ denotes the cumulative distribution function (CDF) of the standard Gaussian distribution. The complex partial derivative operators for a complex number $z=x+\jmath y$ with $\jmath=\sqrt{-1}$ are given by $\frac{\partial}{\partial z}=\frac{1}{2}(\frac{\partial}{\partial x}-\jmath \frac{\partial}{\partial y})$ and $\frac{\partial}{\partial {z}^{*}}=\frac{1}{2}(\frac{\partial}{\partial x}+\jmath \frac{\partial}{\partial y})$, where $\frac{\partial}{\partial x}$ and $\frac{\partial}{\partial y}$ are standard partial derivatives with respect to $x$ and $y$, respectively. $\delta(\cdot)$ denotes the Kronecker function, i.e., $\delta(0)=1$ and $\delta(m)=0,~m\neq 0$.
$\underline{x}=x-\E x$ denotes the centered form of a random variable $x$ and $\cov(x,y)=\E\underline{x}\underline{y} $ represents the covariance of $x$ and $y$. $\BO(\cdot)$ and $\Theta(\cdot)$ denote the big-O and big-theta notations, respectively. Finally, $\xrightarrow[N \rightarrow \infty]{\mathcal{D}}$ denotes the convergence in distribution.

\section{System Model and Problem Formulation}
\label{sec_model}
\subsection{System Model}
Consider a point-to-point MIMO system with $M$ transmit antennas and $N$ receive antennas. The received signal $\bold{y}\in \mathbb{C}^{N}$ can be given by
\begin{equation}
\bold{y}=\BH \bold{s}+\bold{n},
\end{equation}
where $\bold{s}\in \mathbb{C}^{M}$ denotes the transmitted signal with covariance matrix $\E\bold{s}\bold{s}^{H}=\bold{W}$, $\BH\in \mathbb{C}^{N\times M}$ represents the channel matrix, 
and $\bold{n}\in \mathbb{C}^{N}$ is the additive white Gaussian noise. The elements of $\bold{n}$ are i.i.d. and circularly Gaussian random variables with variance $\sigma^2$. As a figure of merit for the performance over coherent MIMO fading channels, the MI is given by~\cite{telatar_capacity}
\begin{equation}
\label{mut_inf}
C(\sigma^2)=\log\det(\bold{I}_{N}+\frac{1}{\sigma^2}\BH\BW\BH^{H}).
\end{equation}
Given $\BH$ is a random matrix, $C(\sigma^2)$ is a random variable whose statistical properties deserve investigation. Specifically, the first-order and second-order analysis for $C(\sigma^2)$ can be used to evaluate the average throughput and outage probability. In this paper, we assume that only the statistical CSI is available at the transmitter.

\subsection{Channel Model}
With statistical CSI, we consider the channel model given by
\begin{equation}
\label{cha_mod}
\BH=\bold{\Phi}_{R}^{\frac{1}{2}}\BX\bold{\Phi}_{S}^{\frac{1}{2}}\BY\bold{\Phi}_{T}^{\frac{1}{2}},
\end{equation}
where $\bold{\Phi}_{R}\in \mathbb{C}^{N\times N}$, $\bold{\Phi}_{S}\in \mathbb{C}^{L\times L}$, and $\bold{\Phi}_{T}\in \mathbb{C}^{M\times M}$ represent the correlation matrices at the receiver, the scatterers, and the transmitter, respectively. $\BX\in  \mathbb{C}^{N\times L}$ and $\BY\in  \mathbb{C}^{L\times M}$ are two random matrices, whose entries are i.i.d. and circularly Gaussian random variables with variances $\frac{1}{L}$ and $\frac{1}{M}$, respectively. The model in~(\ref{cha_mod}) is a special Kronecker model with either a random receive correlation matrix or a random transmit correlation matrix. In practice, this model has the following applications.

\textit{1. Double-scattering channel:}~(\ref{cha_mod}) was proposed in~\cite{gesbert2002outdoor} to model the correlation at the transceiver and scatterers, and the rank of the channel. The rank-deficiency situation is reflected by the case with $L<\min(M,N)$. Without correlation at the transceivers and scatterers, the double-scattering model degenerates to the Rayleigh product model, also referred to as~\textit{multi-keyhole channel}\cite{levin2006multi}, which is the product of two i.i.d. complex Gaussian matrices. Note that $\bold{\Phi}_{S}^{\frac{1}{2}}$ is a Hermitian matrix for a double-scattering channel. 

\textit{2. Cascaded Channel in IRS-aided MIMO systems:}~(\ref{cha_mod}) can also be used to model the cascaded channel in an IRS-aided MIMO system. Under such circumstances, we have $\bold{\Phi}_{S}^{\frac{1}{2}}=\RT_{IRS}\bold{\Psi}\LR_{IRS}$, where $\FT_{IRS}$ and $\FR_{IRS}$ denote the transmit and receive correlation matrix of the IRS, respectively, and $\bold{\Psi}$ represents the phase shifts matrix~\cite{zhang2021large,zhang2022outage}. In this case, $\bold{\Phi}_{S}^{\frac{1}{2}}$ may not be a Hermitian matrix.\footnote{Note that this model and the related analysis for IRS-aided systems is only applicable when the statistical CSI is available. The statistical CSI based analysis and design has been widely used in IRS-aided systems~\cite{zhang2021large,zhang2022outage,zhi2022two,zhao2020intelligent}.}


Given only the statistical CSI at the transmitter, we can assume the covariance matrix of the transmitted signal to be $\BQ=\bold{I}_{M}$. This is because, for any given transmit covariance $\BW$, we can do the substitution $\bold{\Phi}_{T}'=\bold{\Phi}_{T}^{\frac{1}{2}}\BW\bold{\Phi}_{T}^{\frac{H}{2}}$ to construct a new correlation matrix $\bold{\Phi}_{T}'$ at the transmitter such that the evaluation can be converted to the case with an identity transmit covariance matrix, i.e., $C(\sigma^2)=\log\det(\bold{I}_{N}+\frac{1}{\sigma^2}\BH\BW\BH^{H})=\log\det(\bold{I}_{N}+\frac{1}{\sigma^2}\BH'\BH'^{H})$, where $\BH'=\bold{\Phi}_{R}^{\frac{1}{2}}\BX\bold{\Phi}_{S}^{\frac{1}{2}}\BY\bold{\Phi}_{T}'^{\frac{1}{2}}$.

\subsection{Problem Formulation}
In the following, we will derive an equivalent channel model to simplify the evaluation of~(\ref{mut_inf}). Define the eigenvalue decompositions $\bold{\Phi}_{R}=\bold{U}_{R}\FR\bold{U}_{R}^{H}$, $\bold{\Phi}_{T}=\bold{U}_{T}\bold{T}\bold{U}_{T}^{H}$ and the singular value decomposition (SVD) $\bold{\Phi}_{S}^{\frac{1}{2}}=\bold{U}_{S}\MS\bold{V}_{S}^{H}$, where $\bold{R}$, $\bold{S}$, and $\bold{T}$ are diagonal matrices. Then, we can obtain 
\begin{equation}
\begin{aligned}
C(\sigma^2)
&=\log\det(\bold{I}_{N}+\frac{1}{\sigma^2}\bold{\Phi}_{R}^{\frac{1}{2}}\BX\bold{\Phi}_{S}^{\frac{1}{2}}\BY\bold{\Phi}_{T}\BY^{H}\bold{\Phi}_{S}^{\frac{H}{2}}\BX^{H}\bold{\Phi}_{R}^{\frac{1}{2}})
\\
&
=\log\det({\bold{I}}_{N}+\frac{1}{\sigma^2}\bold{U}_{R} \bold{R}^{\frac{1}{2}}\bold{U}_{R}^{H}\bold{X} \bold{U}_{S}\bold{S}^{\frac{1}{2}} \bold{V}_{S}^{H} \bold{Y}\bold{U}_{T} \bold{T}\bold{U}_{T}^{H}
\bold{Y}^{H}\bold{V}_{S}\bold{S}^{\frac{1}{2}} \bold{U}_{S}^{H} \bold{X}^{H}\bold{U}_{R} \bold{R}^{\frac{1}{2}}\bold{U}_{R}^{H})  
\\
&\overset{(a)}{=}\log\det({\bold{I}}_{N}+\frac{1}{\sigma^2}\bold{R}^{\frac{1}{2}}\bold{X}'\bold{S}^{\frac{1}{2}}\bold{Y}'\bold{T}\bold{Y}'^{H}\bold{S}^{\frac{1}{2}}\bold{X}'^{H}\bold{R}^{\frac{1}{2}}),
\end{aligned}
\end{equation}
where $\bold{X}'=\bold{U}_{R}^{H}\bold{X}\bold{U}_{S}$, $\bold{Y}'=\bold{V}_{S}^{H}\bold{Y}\bold{U}_{T}$ and step $(a)$ follows from the identity $\det(\bold{I}+\bold{A}\bold{B})=\det(\bold{I}+\bold{B}\bold{A})$. Given the unitary-invariant attribute of Gaussian random matrices, a Gaussian matrix $\bold{G}$ is equivalent to $\bold{G}'=\bold{U}\bold{G}\bold{V}$ statistically, where $\bold{U}$ and $\bold{V}$ are both unitary matrices. 
Therefore, $\bold{X}$ ($\bold{Y}$) and $\bold{X}'$ ($\bold{Y}'$) are statistically equivalent. Thus, it is equivalent to consider the following channel
\begin{equation}
\label{h_rst}
{\bold{H}}=\bold{R}^{\frac{1}{2}}\bold{X}\bold{S}^{\frac{1}{2}}\bold{Y}\bold{T}^{\frac{1}{2}},
\end{equation}
where $\FR=\diag\left(r_1,r_2,...,r_{N} \right)$, $\FS=\diag\left(s_1,s_2,...,s_{L} \right)$, $\FT=\diag\left(t_1,t_2,...,t_{M} \right)$. In this paper, we will investigate the asymptotic distribution of the statistic $C(\sigma^2)=\log\det(\bold{I}_{N}+\frac{1}{\sigma^2}\BH\BH^{H})$.

\subsection{Assumptions}
The results of this paper are developed based on the following assumptions:

\textbf{A.1.} $0<\lim\inf\limits_{M \ge 1}  \frac{M}{L} \le \frac{M}{L}  \le \lim \sup\limits_{M \ge 1} \frac{M}{L} <\infty$ and $0<\lim \inf\limits_{M \ge 1}  \frac{M}{N} \le \frac{M}{N}  \le \lim \sup\limits_{M \ge 1}  \frac{M}{N} <\infty$,

\textbf{A.2.} $\lim \sup\limits_{M\ge 1} \| \FR \| <\infty$, $\lim \sup\limits_{M\ge 1} \| \FT\| <\infty$, and
$\lim \sup\limits_{M\ge 1} \| \FS\| <\infty$,

\textbf{A.3.} $\inf\limits_{M\ge 1} \frac{1}{M}\Tr\FR>0  $, $\inf\limits_{M\ge 1} \frac{1}{M}\Tr\FT>0 $, and $\inf\limits_{M\ge 1} \frac{1}{M}\Tr\FS>0$.

 \textbf{A.1} defines the asymptotic regime considered for the large-scale system, where the dimensions of the system ($M$, $N$, and $L$) grow to infinity at the same pace. \textbf{A.2} and~\textbf{A.3} restrict the rank of the correlation matrices so that the extremely low-rank case, i.e., the rank of the correlation matrix does not increase with the number of antennas, will not occur~\cite{hachem2008new,zhang2021bias}.

The product of two random matrices makes it challenging to give an explicit expression for the distribution of the MI. In this paper, we will investigate the asymptotic distribution of the MI for double-scattering channels by setting up a CLT utilizing the Gaussian tools.

\section{Main Results}
\label{sec_res}
\subsection{Useful notations}
Before presenting the main results, we first introduce some notations that will be utilized in this paper.

\textit{The Fundamental System of Equations)} The asymptotic mean and variance in the CLT will be induced by the solution of the following system of equations with respect to $(\delta(z),\omega(z),\overline{\omega}(z))$,
 \begin{equation}
 \label{basic_eq1}
 \begin{aligned}
 \begin{cases}
  &\delta(z)=\frac{1}{L}\Tr\bold{R}\left(z \bold{I}_{N}+\frac{M \omega(z)\overline{\omega}(z)}{L\delta(z)}\bold{R}   \right)^{-1},
 \\
 &\omega(z)=\frac{1}{M}\Tr \bold{S}\left(\frac{1}{\delta(z)}\bold{I}_{L}+\overline{\omega}(z)\bold{S}\right)^{-1}, 
 \\
 &\overline{\omega}(z)=\frac{1}{M}\Tr\bold{T}\left(\bold{I}_{M}+\omega(z)\bold{T}\right)^{-1}.
\end{cases}
  \end{aligned}
 \end{equation}
For the simplicity of notations, we will omit $z$ in $(\delta(z),\omega(z),\overline{\omega}(z))$ and use $(\delta,\omega,\overline{\omega})$ in the following derivation. The existence and uniqueness of the positive solution of~(\ref{basic_eq1}) was proved in~\cite{hoydis2011iterative} by the~\textit{standard interference function theory}, which is widely utilized to show the existence and uniqueness of the solution for such fixed-point equations, e.g.,~\cite{couillet2011deterministic,zx2021large}. in the following derivation. The fixed-point algorithm for obtaining the solution was given in~\cite{zhang2021outageext}. It was shown that the complexity of the algorithm is $\BO(\log(\frac{1}{\varepsilon}))$, where $\varepsilon$ is the accuracy of the solution.

\textit{Resolvent)} 
The resolvent matrix of $\BH\BH^{H}$ is given by
  \begin{equation}
  \label{def_res}
  \BQ(z)=\left(z\bold{I}_{N}+\BH\BH^{H} \right)^{-1},
\end{equation}
which will be shown to be highly related to the MI in Section~\ref{sec_pro_fir}. The following resolvent identity holds true
  \begin{equation}
  \label{res_ide}
 z\BQ(z)=\bold{I}_{N}-\BH\BH^{H}\BQ(z).
 \end{equation}
 For simplicity, we will denote $\BQ(z)$ as $\BQ$ in the following.

\textit{Deterministic Approximations of the Resolvent Matrix)} We define the following deterministic matrices
  \begin{equation}
  \label{G_RST}
  \begin{aligned}
 \BG_{R}&=\left(z \bold{I}_{N}+\frac{M \omega\overline{\omega}}{L\delta}\bold{R}   \right)^{-1},
 \BG_{S}=\left(\frac{1}{\delta}\bold{I}_{L}+\overline{\omega}\bold{S}\right)^{-1},
\BG_{T}=\left(\bold{I}_{M}+\omega\bold{T}\right)^{-1},~\BF_{S}=\left(\bold{I}_{L}+\delta\overline{\omega}\bold{S}\right)^{-1}=\frac{1}{\delta}\BG_{S}.
\end{aligned}
\end{equation}
$\BG_{R}$ will be shown to be an approximation of $\E\BQ(z)$. Related notations induced by~$\BG_{R},\BG_{S},\BG_{T}$ are presented in Table~\ref{var_list}, which will be used in the following derivation.
 \begin{table*}[!htbp]
\centering
\caption{Table of Notations.}
\label{var_list}
\begin{tabular}{|cc|cc|cc|cc|}
\toprule
Notations& Expression &  Notations & Expression &Notations& Expression &  Notations& Expression \\
\midrule
$\nu_{R}$ & $\frac{1}{L}\Tr\bold{R}^{2}\BG_{R}^{2}$
&
$\nu_{R,I}$ & $\frac{1}{L}\Tr\bold{R}\BG_{R}^{2}$
&
$\nu_{S}$ &  $\frac{1}{M}\Tr\bold{S}^2\BG_{S}^2$
&
$\nu_{S,I}$ &  $\frac{1}{M}\Tr\bold{S}\BG_{S}^2$
\\
$\nu_{T}$& $\frac{1}{M}\Tr\bold{T}^2\BG_{T}^2$
&
$\nu_{T,I}$& $\frac{1}{M}\Tr\bold{T}\BG_{T}^2$
&
$\eta_{R}$ & $\frac{1}{L}\Tr\bold{R}^3\BG_{R}^3$
&
$\eta_{R,I}$ & $\frac{1}{L}\Tr\bold{R}^2\BG_{R}^3$
\\
$\eta_{S}$  & $\frac{1}{M}\Tr\bold{S}^3\BG_{S}^{3}$
&
 $\eta_{S,I}$ & $\frac{1}{M}\Tr\bold{S}^2\BG_{S}^{3}$
 &
  $\eta_{T}$ & $\frac{1}{M}\Tr\bold{T}^3\BG_{T}^{3}$
  &
    $\eta_{T,I}$ & $\frac{1}{M}\Tr\bold{T}^{2}\BG_{T}^{3}$
  \\
  $\Delta_{S}$& $1-\nu_{S}\nu_{T}$
&
$\Delta$ & $1- \frac{M \omega\overline{\omega}\nu_{R} }{L\delta^2}+\frac{M\nu_{R}\nu_{S,I}\nu_{T,I} }{L\delta^3 \Delta_{S}}$
&
$\theta$
&
$\frac{M \omega\overline{\omega}}{L\delta^2}-\frac{M \nu_{S,I}\nu_{T,I} }{L\delta^3\Delta_{S} }$
  &
  &
 \\
\bottomrule
\end{tabular}
\end{table*}

\subsection{First-order Result}
\begin{theorem} 
\label{fir_the}
Given that assumptions~\textbf{A.1}-\textbf{A.3} hold true and $(\delta,\omega,\overline{\omega})$ is the positive solution of~(\ref{basic_eq1}) when $z=\sigma^2$, the EMI can be approximated by
\begin{equation}
\E C(\sigma^2)= \overline{C}(\sigma^2)+\mathcal{O}(\frac{1}{N}),
\end{equation}
where
\begin{equation}
\begin{aligned}
\label{the_mean}
\overline{C}(\sigma^2)=\log\det(\bold{I}_{N}+\frac{ M \omega\overline{\omega}}{\sigma^2 L\delta}\bold{R} )+\log\det(\bold{I}_{L}+\delta\overline{\omega}\bold{S})
+\log\det(\bold{I}_{M}+\omega\bold{T})-2M\omega\overline{\omega}.
\end{aligned}
\end{equation}
\end{theorem}
\begin{proof}
The proof of Theorem~\ref{fir_the} will be given in Section~\ref{sec_pro_fir}.
\end{proof}
\begin{remark} (Comparison with existing results) A similar approximation was first proposed in~\cite[Theorem 2]{hoydis2011asymptotic} and~\cite[Theorem 5]{hoydis2011iterative} but only the convergence of the normalized MI (MI per receive antenna) was guaranteed. The result about the normalized MI is also consistent with that shown in~\cite[Corollary 1]{zhang2021large} and~\cite[Proposition 1]{moustakas2021capacity}, which were derived by the replica method. The replica method is efficient in determining the deterministic approximation, but can not provide an exact convergence rate. In fact, the convergence rate of the EMI without normalization is not available in the literature. In this paper, by utilizing the Gaussian tools, we derive the approximation of the EMI and obtain a convergence rate. From Theorem~\ref{fir_the}, we can observe that the error of the deterministic approximation in~(\ref{the_mean}) can be bounded by $\BO(\frac{1}{N})$. This convergence rate was also validated for single hop Gaussian channels~\cite{hachem2008new}, frequency selective channels~\cite{dupuy2011capacity}, and Rician channels~\cite{dumont2010capacity}. 
\end{remark}

\subsection{Second-order Result}
\begin{theorem}
\label{sec_the}
(CLT of the MI) Given the same setting as that in Theorem~\ref{fir_the}, the asymptotic distribution of the MI over double-scattering channel is Gaussian and the following convergence holds true
\begin{equation}
\frac{C(\sigma^2)-\overline{C}(\sigma^2)}{\sqrt{V(\sigma^2)}} \xrightarrow[N \rightarrow \infty]{\mathcal{D}}  \mathcal{N}(0,1),
\end{equation}
where $\overline{C}(\sigma^2)$ is given in~(\ref{the_mean}). The asymptotic variance $V(\sigma^2)$ can be expressed as 
\begin{equation}
\label{low_var}
\begin{aligned}
V(\sigma^2)&=-\log(\Delta)-\log(\Delta_{S}),
\end{aligned}
\end{equation}
where $\Delta$ and $\Delta_{S}$ are given in Table~\ref{var_list}. Furthermore, the convergence of the variance is given by
\begin{equation}
\Var(C(\sigma^2))=V(\sigma^2)+\BO(\frac{1}{N}),
\end{equation}
and the convergence rate of the characteristic function of $C(\sigma^2)$ is $\BO(\frac{1}{N})$.
\end{theorem}
\begin{proof}
The proof of Theorem~\ref{sec_the} will be given in Section~\ref{sec_pro_clt}.
\end{proof}
\begin{remark} (Non-Gaussian fading)
The $\BO(\frac{1}{N})$ convergence rate of the asymptotic mean and variance is preserved by the Gaussianity of the two random matrices in~(\ref{h_rst}), and will not hold for the product of non-Gaussian matrices. In~\cite{gotze2017distribution}, the linear spectral statistics (LSS) (mutual information is a special case of LSS by taking $f(x)=\log(1+\frac{x}{\sigma^2})$) of two i.i.d., real, non-Gaussian square matrices was investigated to show the asymptotic Gaussianity of the LSS but the asymptotic mean and variance are related to the fourth-order cumulant of entries of the two matrices, which is zero for Gaussian entries. This phenomenon has also been discovered in single-hop channels. The centered and non-centered non-Gaussian channels were investigated in \cite{bao2015asymptotic}, and~\cite{hachem2012clt,zhang2021bias}, respectively. These results indicate that compared with the mean and variance of Gaussian channels, the asymptotic mean and variance of non-Gaussian channels should be amended by terms related to the pseudo-variance and fourth-order cumulant of the complex entries, which are referred to as the~\textit{bias}. Therefore, it is reasonable to make a conjecture that the asymptotic Gaussianity of the MI still holds true for the two-hop non-Gaussian channels where the asymptotic mean and variance of the MI should be amended by adding a bias to the mean and variance derived in Theorem~\ref{sec_the}. The investigation of the bias for the product of two general non-Gaussian random matrices will be left as a future work.
\end{remark}

\begin{remark} (Comparison with existing results)
~\cite{zheng2016asymptotic} presented a CLT for the MI over Rayleigh product channel (i.i.d. double-scattering channel) by a free probability approach. The explicit expression for the variance was given for the case with $N=M$, which is a special case of Theorem~\ref{sec_the}. Motivated by the thought of splitting the randomness into two parts with respect to two random matrices,~\cite[Theorem 2]{zhang2022outage} derived the same result by using the~\textit{resolvent method}. However, how fast the distribution of $C(\sigma^2)$ converges to a Gaussian distribution and how fast the mean and the variance converge to their deterministic approximations are both unknown. In this paper, we recover the result in~\cite[Theorem 2]{zhang2022outage} and show, for the first time, that the convergence rates of the variance and the characteristic function are $\BO(\frac{1}{N})$.
\end{remark}
\begin{proposition} 
\label{bound_var}
Given assumptions~\textbf{A.1},~\textbf{A.2},~\textbf{A.3}, and a $\sigma^2>0$, there exist two positive numbers $b_{\sigma^2}$ and $B_{\sigma^2}$ satisfying
\begin{equation}
0<b_{\sigma^2} \le \inf_{N} V(\sigma^2) \le  \sup_{N} V(\sigma^2) \le B_{\sigma^2} \le \infty.
\end{equation}
\end{proposition}
\begin{proof}
The proof of Proposition~\ref{bound_var} is given in Appendix~\ref{proof_bound_var}.
\end{proof}
\begin{remark} Proposition~\ref{bound_var} indicates that the asymptotic variance $V(\sigma^2)$ is well-defined for a given $\sigma^2$.
\end{remark}
By Theorem~\ref{sec_the}, we can obtain a closed form approximation for the outage probability with a given rate $R$ as
\begin{equation}
\label{appro_outage}
p_{out}(R)\approx \Phi(\frac{R-\overline{C}(\sigma^2)}{\sqrt{V(\sigma^2)}}).
\end{equation}
Conversely, given the outage probability $p_{out}$, the outage rate, defined by
\begin{equation}
C_{p_{out}}=\sup\{ R\ge 0 : \mathbb{P}(C({\sigma^2})<R) \le p_{out}  \},
\end{equation}
can be approximated by $C_{p_{out}}\approx \overline{C}(\rho)+  \sqrt{V(\rho)}\Phi^{-1}(p_{out})$. Please note that~(\ref{appro_outage}) is obtained in a large system regime defined by~\textbf{A.1} and not applicable for the SNR-asymptotic regime~\cite{zheng2003diversity}. However, the outage probability approximation in~(\ref{appro_outage}) can be utilized for the finite-SNR diversity-multiplexing trade-off (DMT)~\cite{loyka2010finite} analysis for two-hop channels, which was given in~\cite{zhang2022outage}.

\subsection{i.i.d Case}
To investigate the impact of the rank deficiency, we ignore the correlations and focus on the i.i.d. channel, i.e., Rayleigh product channel. In this case, the results will only be related to the ratios between the numbers of antennas and scatterers, denoted by $\eta=\frac{N}{M}$ and $\kappa=\frac{M}{L}$.

\begin{proposition}
\label{iid_the}
 (Rayleigh-product channel)
Given~\textbf{Assumption~{A.1}}, the asymptotic distribution of the MI over i.i.d. double-scattering channel is Gaussian and the following convergence holds true
\begin{equation}
\frac{C_{iid}(\sigma^2)-\overline{C}_{iid}(\sigma^2)}{\sqrt{V_{iid}(\sigma^2)}} \xrightarrow[{N  \xrightarrow[]{}\infty}]{\mathcal{D}}  \mathcal{N}(0,1),
\end{equation}
where $\overline{C}_{iid}$ and $V_{iid}$ are given by
\begin{equation}
\label{cap_app}
\begin{aligned}
\overline{C}_{iid}(\sigma^2)&=M[\log(1+\omega)-(\eta-\frac{1}{\kappa})\log(1-\frac{\omega}{\eta(1+\omega)})
-\frac{\log(\sigma^2)}{\kappa}-\frac{\log(\omega)}{\kappa}
-\frac{2\omega}{1+\omega}],
\\
V_{iid}(\sigma^2)&=-\log(\Delta_{iid}),
\end{aligned}
\end{equation}
with
\begin{equation}
\label{del_iid}
\Delta_{iid}=\frac{(1+\delta\overline{\omega}^2)\delta(\sigma^2+ \frac{\kappa\omega\overline{\omega}^2}{\delta(1+\delta \overline{\omega}^2)})}{\eta\kappa(1+\delta\overline{\omega})}.
\end{equation}
Here $\omega$ is the root of the following cubic equation
\begin{equation}
\label{new_cubic_eq}
\begin{aligned}
P(\sigma^2)=\omega^3+(2\sigma^2+\eta\kappa-\kappa-\eta+1 )\frac{\omega^2}{\sigma^2}+(1+\frac{\eta\kappa}{\sigma^2}-\frac{2\eta}{ \sigma^2}+\frac{1}{\sigma^2})\omega-\frac{\eta}{\sigma^2}=0,
\end{aligned}
\end{equation}
which satisfies $\omega>0$ and $\eta+(\eta-1)\omega>0$, and 
\begin{equation}
\label{delta_iid_exp}
\begin{aligned}
\delta&=\frac{1}{\sigma^2}(\frac{N}{L}-\frac{M\omega}{L(1+\omega)} ), ~~~
\overline{\omega}&=\frac{1}{1+\omega}.
\end{aligned}
\end{equation}
\end{proposition}
\begin{proof} Proposition~\ref{iid_the} can be obtained by letting $\FR=\BI_{N}$, $\FS=\BI_{L}$, and $\FT=\BI_{M}$ in Theorem~\ref{sec_the} and the proof is omitted here.
\end{proof}

\begin{remark}
\label{degenerate_to_iid_Ray}
(From Rayleigh-product channel to single-Rayleigh channel) In the following, we will show that
when $\kappa \xrightarrow[]{} 0 $ and $\eta=\BO(1)$, i.e., the number of scatterers is far larger than the number of antennas, the results in Proposition~\ref{iid_the} degenerate to those for the of i.i.d. Rayleigh channel. When $\kappa  \xrightarrow[]{} 0$, equation~(\ref{new_cubic_eq}) becomes 
\begin{equation}
\label{qua_eq}
\omega^2+(1-\frac{\eta}{\sigma^2}+\frac{1}{\sigma^2})\omega -\frac{\eta}{\sigma^2}=0,
\end{equation}
whose positive solution is given by
\begin{equation}
\omega_{0}=\frac{\rho(\eta-1)-1+\sqrt{ [\rho(\eta-1)-1]^2+4\eta\rho   }}{2}
=\frac{\rho(\eta-1)-1+\sqrt{\rho^2(\eta+1)^2 +2\rho(\eta+1)+1 -4\rho^2\eta }}{2}.
\end{equation}
Denote
\begin{equation}
\label{v_eta_rho}
v(\eta,\rho)=\frac{\omega_{0}}{1+\omega_{0}}=\frac{\eta+1+\sigma^2  -   \sqrt{(1+\sigma^2+\eta)^2 -4\eta   }}{2},
\end{equation}
which is identical to~\cite[Eq. (11)]{kamath2005asymptotic}. We can also obtain from~(\ref{delta_iid_exp}) that
\begin{equation}
\frac{L\delta_{0}}{M}=\frac{1}{\sigma^2}(\eta-\frac{\omega_{0}}{1+\omega_{0}})\overset{(a)}{=}\omega_{0},
\end{equation}
where $(a)$ follows from~(\ref{qua_eq}). We can rewrite $\overline{C}_{iid}(\sigma^2)$ as
\begin{equation}
\label{C_iid_rayleigh_ref}
\begin{aligned}
&\overline{C}_{iid}(\sigma^2)=M[\eta\log(1+\frac{\kappa \omega\overline{\omega}}{\sigma^2\delta})
+\frac{1}{\kappa}\log(1+\delta\overline{\omega})
+\log(1+\omega)-2\omega\overline{\omega}]
\xrightarrow[]{\kappa\rightarrow 0} 
\\
&
M[\eta\log(1+\frac{1}{\sigma^2(1+\omega_{0} )})
+\log(1+\omega_{0} )-\frac{\omega_{0} }{1+\omega_{0} }]
=M[\eta\log(1+ \rho -\rho v(\eta,\rho) ) +\log(1+\eta\rho-\rho v(\eta,\rho))-v(\eta,\rho)    ]
,
\end{aligned}
\end{equation}
which coincides with~\cite[Eq. (11)]{kamath2005asymptotic}. Now we turn to evaluate the variance. By~(\ref{del_iid}), we have
\begin{equation}
\Delta_{iid} \xrightarrow[]{\kappa\rightarrow 0} \frac{\omega_{0}(\sigma^2+\frac{1}{(1+\omega_{0})^2})}{\eta}=\frac{\omega_{0}\sigma^2+\frac{\omega_0}{1+\omega_{0}}-\frac{\omega_0^2}{(1+\omega_{0})^2}}{\eta}=1-\frac{\omega_{0}^2}{\eta(1+\omega_{0})^2}
=1-\frac{v(\eta,\rho)^2}{\eta},
\end{equation}
and $V_{iid}(\sigma^2)\xrightarrow[]{\kappa\rightarrow 0}  -\log\left(1-\frac{v(\eta,\rho)^2}{\eta}\right)$, which is the same as~\cite[Eq. (13)]{kamath2005asymptotic}. Therefore, when $\kappa \xrightarrow[]{} 0$, the results for the Rayleigh-product channel degenerate to those for the i.i.d. single Rayleigh channel.
\end{remark}
\begin{remark} (Extreme rank-deficient case) Different from Remark~\ref{degenerate_to_iid_Ray}, here we consider the case when $L$ is finite but $M$ and $N$ go to infinity with the same pace, i.e., $\kappa \xrightarrow[]{}  \infty$. Under such circumstances, by analyzing the coefficient of the dominating term $\kappa$ at the LHS of~(\ref{new_cubic_eq}), we can obtain $\omega\xrightarrow[]{\kappa\xrightarrow[]{}\infty } 0$ and 
\begin{equation}
\begin{aligned}
 \rho\kappa\eta\omega-\eta\rho  \xrightarrow[]{\kappa\xrightarrow[]{}\infty }   0,
\end{aligned}
\end{equation}
so that we can obtain the solution of~(\ref{new_cubic_eq}) as $\omega_{\infty}=\kappa^{-1}+\frac{(1-\eta\rho) \kappa^{-2}}{\eta\rho}+\BO(\kappa^{-3})$ and $\delta_{\infty}=\frac{\eta\kappa}{\sigma^2}-\frac{1}{\sigma^2}+\BO(\kappa^{-1})$. In this case, we have
\begin{equation}
\label{cap_app}
\begin{aligned}
\overline{C}_{iid}(\sigma^2)&=M[\log(1+\omega)-(\eta-\frac{1}{\kappa})\log(1-\frac{\omega}{\eta(1+\omega)})
-\frac{\log(\sigma^2)}{\kappa}-\frac{\log(\omega)}{\kappa}
+\frac{\log(\eta)}{\kappa}-\frac{2\omega}{1+\omega}]
\\
&\approx 2L-L\log(\sigma^2)+L\log(\frac{M}{L})  +L\log(\eta)-2L
= L\log(\frac{N}{L\sigma^2}).
\end{aligned}
\end{equation}
This indicates that, if $L$ is fixed, increasing $M$ and $N$ will result in an $\BO(\log(N))$ increase in the ergodic rate. However, we can obtain from~(\ref{cap_app}) that increasing $L$, $M$, and $N$ with the same pace will result in an $N$ increase in the ergodic rate. By plugging $\omega_{\infty}$ and $\delta_{\infty}$ into~(\ref{del_iid}), we can evaluate the variance as
\begin{equation}
\label{V_iid_exp}
\begin{aligned}
V_{iid}(\sigma^2)&=-\log(\Delta_{iid})=-\log(1-\frac{1+\eta}{\eta\kappa})+\BO(\kappa^{-2})=\BO(\kappa^{-1}),
\end{aligned}
\end{equation}
where
\begin{equation}
\begin{aligned}
\Delta_{iid}&=\frac{(1+\delta\overline{\omega}^2)\delta(\sigma^2+ \frac{\kappa\omega\overline{\omega}^2}{\delta(1+\delta \overline{\omega}^2)})}{\eta\kappa(1+\delta\overline{\omega})}
= (1- \kappa^{-1}+\BO(\kappa^{-2}) )(\sigma^2+ \BO(\kappa^{-2}) )\frac{1-\eta^{-1}\kappa^{-1}+\BO(\kappa^{-2}) }{\sigma^2} 
\\
&=1-\frac{1+\eta}{\eta\kappa}+\BO(\kappa^{-2}).
\end{aligned}
\end{equation}
The variance is of order $\BO(\kappa^{-1})=\BO(\frac{1}{M})$, which vanishes as $M$ goes to infinity. By Markov's inequality, we can obtain that for any given positive $\varepsilon$, there holds true that 
\begin{equation}
\Prob(|C_{iid}(\sigma^2)-\E C_{iid}(\sigma^2)  |>\varepsilon ) = \Prob(|C_{iid}(\sigma^2)-\E C_{iid}(\sigma^2)  |^2>\varepsilon^2 ) \le \frac{\Var(C_{iid}(\sigma^2))}{\varepsilon^2}
 \xrightarrow[]{\kappa\xrightarrow[]{}\infty } 0,
\end{equation}
which indicates that $C_{iid}(\sigma^2)$ converges to its expectation in probability. Therefore, if $L$ is fixed and $M$, $N$ go to infinity with the same pace, $C_{iid}(\sigma^2)$ tends to be deterministic.
\end{remark}
\subsection{Moderate-to-High SNR Approximation}
To investigate the impact of the system dimensions on the mean and variance of the MI, we give the moderate-to-high SNR approximation of the mean and variance for the Rayleigh-product channel by the following proposition.

\begin{proposition}
\label{pro_rnk}
 (The impact of the number of scatterers) Denote $(S_{N},S_{L},S_{M})$ as the rearranged version of $(N, L, M)$ in the descending order, i.e., $S_N\ge S_L \ge S_M$ and $\{S_N,S_L,S_M \}=\{N,L,M \}$. Given $\rho=\frac{1}{\sigma^2}$, the mean $\overline{C}_{iid}(\rho^{-1})$ and variance ${V}_{iid}(\rho^{-1})$ of the MI can be approximated by
 \begin{subnumcases}
 { \overline{C}_{iid}(\rho^{-1})
 \!\!=\!\!\label{appros_Ciid}}
S_{M}\log(\frac{\rho N}{e^2 S_{M} }) \!-\! (S_{N}-S_{M})\log(1-\frac{S_M}{S_N}) \!-\! (S_L-S_M)\log(1-\frac{S_{M}}{S_{L}}) +\BO(\rho^{-1}),~\text{when}~S_{L}\neq S_{M}
,\label{c_appro_case_1}
  \\
S_M\log(\frac{\rho N }{ e^{2}S_{M}}) \!-\!  (S_N-S_M)\log(1-\frac{S_{M}}{S_{N}})
+\frac{2S_{M}  (\frac{N \rho}{S_M})^{-\frac{1}{2}}}{(1-\frac{S_{M}}{S_{N}})^{\frac{1}{2}}}+\BO(\rho^{-1}), ~\text{when}~S_{L}= S_{M}\neq S_{N},
  \label{c_appro_case_2}
  \\
 S_M\log(\frac{\rho}{e^2})+3S_M\rho^{-\frac{1}{3}}+\BO(\rho^{-\frac{2}{3}}),~\text{when}~S_{L}= S_{M}=S_{N}.\label{c_appro_case_3}
\end{subnumcases}

  \begin{subnumcases}
 { V_{iid}(\rho^{-1})
 =\label{appros_Viid}}
-\log((1-\frac{S_{M}}{S_{L}})(1-\frac{S_{M}}{S_{N}}))+\BO(\rho^{-1}),~\text{when}~S_{L}\neq S_{M}
,\label{v_appro_case_1}
  \\
\frac{1}{2}\log(\frac{\rho N }{4(1-\frac{S_M}{S_N})S_{M}})
+\frac{(1-\frac{2S_{M}}{S_{N}})  (\frac{N \rho}{S_M})^{-\frac{1}{2}}}{(1-\frac{S_{M}}{S_{N}})^{\frac{1}{2}}}+\BO(\rho^{-1}), ~\text{when}~S_{L}= S_{M}\neq S_{N},
  \label{v_appro_case_2}
  \\
 \frac{2\log(\rho)}{3}-\log(3)+\frac{4\rho^{-\frac{1}{3}}}{3}+\BO(\rho^{-\frac{2}{3}}),~\text{when}~S_{L}= S_{M}=S_{N}.\label{v_appro_case_3}
\end{subnumcases}
\end{proposition}
\begin{proof} The proof of Proposition~\ref{pro_rnk} is given in Appendix~\ref{proof_pro}.
\end{proof}

\begin{remark}\label{remark_multiplex}
From~(\ref{appros_Ciid}) and~(\ref{appros_Viid}), we can obtain that the approximations of $\overline{C}_{iid}(\rho^{-1})$ and ${V}_{iid}(\rho^{-1})$ are only determined by $N$ and the ordered $(N,L,M)$. A similar phenomenon was also noticed when investigating the DMT in~\cite{yang2011diversity}. We can observe that $\overline{C}_{iid}(\rho^{-1})$ grows in $\BO(\log(\rho))$ when $\rho$ increases and the increasing speed, i.e., the coefficient of $\log(\rho)$, is determined by $S_{M}$, which is the minimum of $(N,L,M)$. The impact of the number of the scatterers can be observed from the dominating term in the mean and variance. When $S_M=L$, the multiplexing gain is limited by L, and increasing N and M will not help. Meanwhile, when $L>\max\{M,N \} $, large $L$ does not contribute to the multiplexing gain although it helps for the diversity gain, which agrees with the result in~\cite[Example 6]{shin2003capacity} and that for the IRS-aided MIMO channel in~\cite[Theorem 5]{zhang2022outage}.
\end{remark}
\begin{remark}
\label{scatter_ana} (Comparison between Rayleigh-product and single Rayleigh channel) The moderate-to-high SNR approximation for the mean and variance of the i.i.d. Rayleigh channels is given by~\cite[Eq.(13),(14)]{loyka2010finite},
\begin{equation}
\label{single_app}
\overline{C}_{Rayleigh}(\rho^{-1})\approx  
  U_M\log(\frac{\rho N}{e U_{M}})-(U_N-U_M)\log(1-\frac{U_M}{U_N}),
\end{equation}
  \begin{subnumcases}
 { V_{Rayleigh}(\rho^{-1})
 \approx\label{v_appro_ray}}
-\log(1-\frac{U_M}{U_N}),~~~U_M\neq U_N,
  \label{v_appro_ray_1}
  \\
\frac{1}{2}[\log(\frac{\rho}{4})+2\rho^{-\frac{1}{2}} ],~~~U_M=U_N,\label{v_appro_ray_2}
\end{subnumcases}
where $U_M=\min\{M,N \}$ and $U_N=\max\{M,N \}$. From~(\ref{single_app}), we can observe that the multiplexing gain of the single Rayleigh channel is also limited by the minimum dimension. For the Rayleigh-product channel, when $L<U_M$, i.e., $L$ is the smallest among $N,L$, and $M$, the loss of the ergodic rate due to rank deficiency is reflected by the coefficient $L$ of the dominating term in~(\ref{appros_Ciid}). When $L>U_M$, by comparing~(\ref{c_appro_case_1}) and~(\ref{single_app}), the loss of the ergodic rate is reflected by the $\BO(1)$ term, which is $S_M(-1-(\frac{L}{S_M}-1)\log(1-\frac{S_M}{L}))<0$ ($f(x)=(\frac{1}{x}-1)\log(1-x)$ is an increasing function so that $f(x)>-1$ when $0<x<1$). When $L \rightarrow \infty$, the loss vanishes and~(\ref{c_appro_case_1}) becomes~(\ref{single_app}). From the high-order terms, i.e., $\BO(\rho^{-\frac{1}{2}})$ in~(\ref{c_appro_case_2}) and $\BO(\rho^{-\frac{1}{3}})$ in~(\ref{c_appro_case_3}), we can observe that the condition $S_L=S_M$ will affect the high order behavior, which also happens for the single Rayleigh case. In fact, according to~(\ref{v_eta_rho}) and~(\ref{C_iid_rayleigh_ref}),~(\ref{single_app}) can be further written as
\begin{equation}
\overline{C}_{Rayleigh}(\rho^{-1})=
\begin{cases}  
  U_M\log(\frac{\rho N}{e U_{M}})-(U_N-U_M)\log(1-\frac{U_M}{U_N})+\BO(\rho^{-1}),~~~U_M\neq U_N
  \\
  U_M\log(\frac{\rho}{e})+2U_M\rho^{-\frac{1}{2}}+\BO(\rho^{-1}),~~~U_M=U_N,
\end{cases}
\end{equation}
where the case with $U_M=U_N$ coincides with~(\ref{c_appro_case_2}) when $L\rightarrow \infty$. 

For the variance,~(\ref{v_appro_ray_1}) and~(\ref{v_appro_ray_2}) can be obtained by letting $\kappa \rightarrow 0$ in~(\ref{v_appro_case_1}) and~(\ref{v_appro_case_2}), respectively. When $L<U_M$, the impact of the rank deficiency is reflected by the dominating term, $-\log((1-\frac{L}{M})((1-\frac{L}{N})))$. When $L<L_{0}=\frac{1}{2}\{U_M+U_N-[(U_N-U_M)(3U_M+U_N)]^{\frac{1}{2}}\}$, $-\log((1-\frac{L}{M})((1-\frac{L}{N})))<-\log(1-\frac{U_M}{U_N})$. When $L_{0}<L<U_M$, $-\log((1-\frac{L}{M})((1-\frac{L}{N})))>-\log(1-\frac{U_M}{U_N})$. When $L>U_M\neq U_N$, the Rayleigh-product channel has a larger variance and the increment is $-\log(1-\frac{U_M}{L})$. When $M=N$ and $L>M$, the dominating terms in~(\ref{v_appro_case_2}) and~(\ref{v_appro_ray_2}) are same and the increment of the variance is $-\frac{1}{2}\log[4(1-\frac{M}{L})]$. From~(\ref{v_appro_case_3}), we can observe that the maximum variance occurs when $M=N=L$.
\end{remark}

\begin{remark} Proposition~\ref{pro_rnk} provides an approximation for the asymptotic mean and variance given by  Proposition~\ref{iid_the} when the SNR $\rho$ is high. One may observe from the proof of Theorem~\ref{fir_the} and the bound of the error in~(\ref{trace_QG_err}) that the approximation error for the asymptotic mean is $\BO(\frac{\mathcal{P}_{8}(\frac{1}{z})}{Mz} )$, which is guaranteed to be $o(1)$ when $\frac{1}{z}$ is $\BO(M^{\frac{1-\varepsilon}{9}})$, where $0<\varepsilon<1$. This means that $\E C_{iid}(\sigma^2)$ converges to the high SNR approximations in~(\ref{appros_Ciid}) when $M$, $N$, $L$ go to infinity with the same pace and $\rho$ grows to infinity with the order $\BO(M^{\frac{1-\varepsilon}{9}})$. The approximation works well in the moderate-to-high SNR region, which is validated by simulations in Section~\ref{sec_simu}. Note that there should be bounds better than $\BO(M^{\frac{1-\varepsilon}{9}})$ and a larger $\rho$ may also be shown to be tight. 
\end{remark}

\section{Methodology and Preliminary Results}
\label{sec_method}
In this section, we will first present the main technique used in this paper, i.e., the Gaussian tools. Then, we will introduce some preliminary results obtained by this approach that will be useful for the following derivation.

\subsection{Gaussian Tools}
The Gaussian tools consist of two parts, i.e., \textit{Nash-Poincar{\'e} Inequality} and the \textit{Integration by Parts Formula}. 

\textit{1. Nash-Poincar{\'e} Inequality}
Denote $\bold{x}=[x_1,...,x_N]^{T}$ as a complex Gaussian random vector satisfying $\E \bold{x}=\bold{0}$, $\E\bold{x}\bold{x}^{T}=\bold{0}$, and $\E\bold{x}\bold{x}^{H}=\bold{\Omega}$. $f=f(\bold{x},\bold{x}^{*})$ is a $\mathbb{C}^{1}$ complex function such that both itself and its derivatives are polynomially bounded. Then, the variance of $f$ satisfies the following inequality,~\cite[Eq. (18)]{hachem2008new},~\cite[Proposition 2.5]{pastur2005simple},
\begin{equation}
\label{nash_p}
\Var[f(\bold{x},\bold{x}^{*})] \le \E[\nabla_{\bold{x}}f(\bold{x},\bold{x}^{*})^{T}\bold{\Omega}\nabla_{\bold{x}}f(\bold{x},\bold{x}^{*})^{*}]
+ \E[\nabla_{\bold{x}^{*}}f(\bold{x},\bold{x}^{*})^{H}\bold{\Omega}\nabla_{\bold{x}^{*}}f(\bold{x},\bold{x}^{*})],
\end{equation}
where $\nabla_{\bold{x}}f(\bold{x})=[\frac{\partial f}{\partial x_1},...,\frac{\partial f}{\partial x_N}]^{T}$ and $\nabla_{\bold{x}^{*}}f(\bold{x})=[\frac{\partial f}{\partial x_1^{*}},...,\frac{\partial f}{\partial x_N^{*}}]^{T}$.
(\ref{nash_p}) is referred to as the~\textit{Nash-Poincar{\'e} inequality}, which gives an upper bound for functional Gaussian random variables and is widely involved in the error estimation for the expectation of Gaussian matrices. By utilizing this inequality, the authors of~\cite{hachem2008new,dumont2010capacity} showed that the approximation error of the deterministic approximation for the MI of single-hop MIMO channels is $\BO(\frac{1}{N})$. In this work, we will use this inequality to bound the error between the expectations and their deterministic approximations to obtain the convergence rate of the EMI and characteristic function of the MI for double-scattering channels. 

\textit{2. Integration by Parts Formula}~\cite[Eq. (17)]{hachem2008new}
\begin{equation}
\label{int_part}
\E[x_{i}f(\bold{x},\bold{x}^{*})]=\sum_{m=1}^{N}[\bold{\Omega}]_{i,m} \E [\frac{\partial f(\bold{x},\bold{x}^{*})}{\partial x_{m}^{*}}].
\end{equation}
If $\bold{\Omega}=\bold{I}_{N}$,~(\ref{int_part}) can be simplified as
\begin{equation}
\E[x_{i}f(\bold{x},\bold{x}^{*})]=\E [\frac{\partial f(\bold{x},\bold{x}^{*})}{\partial x_{i}^{*}}].
\end{equation}
By this formula, the expectation for the product of a Gaussian random variable and a functional Gaussian random variable is converted to the expectation for the derivative of the functional Gaussian random variable. 

\subsection{Preliminary Results}
We first briefly introduce the boundness of the empirical moments for spectral norm. 
\begin{lemma} (The boundness of the expectation for spectral norm~\cite[Eq.(8)]{rubio2012clt}) Given a fixed integer $P>0$ and a sequence of $P$ $N\times N$ diagonal matrices $\{\bold{D}^{(p)} \},~1\le p \le P $ with uniformly bounded norm by $N$, there holds true that
\begin{equation}
\label{X_bound}
\E \|\frac{\BX\bold{D}^{(1)} \BX^{H} }{N}\frac{\BX\bold{D}^{(2)} \BX^{H} }{N}...\frac{\BX\bold{D}^{(P)} \BX^{H} }{N}    \| < K_{P},
\end{equation}
where $\BX$ is a $M \times N$ random matrix with i.i.d. entries and $0<\inf_{N>0}\frac{M}{N}\le \frac{M}{N}\le \sup_{N>0}\frac{M}{N}<\infty$. $K_P$ only depends on $P$. Furthermore, the following bound also holds true,
\begin{equation}
\label{Xp_bound}
\sup_{N\ge 1}\E \| \frac{\BX\BX^{H}}{N}\|^{p} <\infty.
\end{equation}
\end{lemma}

\textit{Some Matrix Results}

1. (Trace inequality) If $\BA$ is a non-negative matrix, we have
\begin{equation}
\label{trace_inq}
|\Tr\BA\BB|\le \|\BB \|\Tr\BA.
\end{equation}

2. (Derivatives of the Resolvent Matrix) Given $\BQ$ is the resolvent matrix defined in~(\ref{def_res}), the related derivatives can be computed by,
\begin{equation}
\label{t_Q_der}
\begin{aligned}
\frac{\partial [\BQ]_{a,b}}{\partial X_{ij}}&=-[\BQ\frac{\partial \BH\BH^{H}}{\partial X_{ij}} \BQ]_{a,b}
=-[\BQ\LR \bold{e}_{i}\bold{e}_{j}^{T}  \MS\BY\RT\BH^{H}  \BQ]_{j,i}
=-[ \MS\BY\RT\BH^{H}  \BQ]_{j,b}[\BQ \LR]_{a,i},
\\
\frac{\partial [\BQ]_{a,b}}{\partial Y_{ij}}&=-[\BQ\frac{\partial \BH\BH^{H}}{\partial Y_{ij}} \BQ]_{a,b}
=-[\BQ\LR\BX\MS \bold{e}_{i}\bold{e}_{j}^{T} \RT\BH^{H}  \BQ]_{a,b}
=-[ \RT\BH^{H}  \BQ]_{j,b} [\BQ\LR\BX\MS]_{a,i},
\end{aligned}
\end{equation}
\begin{equation}
\label{log_Q_der}
\frac{\partial \log\det\BQ^{-1}}{\partial z}=\Tr\BQ,~~~
\frac{\partial \log\det\BQ^{-1}}{\partial Y_{ij}}=\Tr\left(\BQ\frac{\partial \BH\BH^{H}}{\partial Y_{ij}} \right)
=[\RT\BH^{H}\BQ\LR\BX\MS]_{j,i},
\end{equation}
where $\bold{e}_{i}$ represents the unit vector whose $i$-th entry is one and all others are zero. (\ref{t_Q_der}) and~(\ref{log_Q_der}) can be obtained by the derivative rule of matrix~\cite{petersen2008matrix}.

By utilizing the Gaussian tools, the following variance control can be set up.
\begin{proposition}
\label{var_con}
(Variance Control) Given deterministic matrices $\BA$, $\BB$, $\BC$, $\bold{D}$, $\bold{E}$ such that $\sup_{N}\max\{\|\BA \|,\|\BB \|,\|\BC \|, \|\bold{D}\|, \|\bold{E}\|  \} < U<\infty$ and $\BZ=\LR\BX\MS$ as defined in~(\ref{h_rst}), the following holds true
\begin{equation}
\label{V_ABCQ}
\Var(\Tr \bold{A}\BX \bold{B}\BY\bold{C}\BY^{H}\bold{D}\BX^{H}\bold{E}\BQ)=\BO(\frac{\mathcal{P}_{2}(\frac{1}{z})}{z^2}),
\end{equation}
\begin{equation}
\label{VABCQQ}
\Var(\Tr \bold{A}\BX \bold{B}\BY\bold{C}\BY^{H}\bold{D}\BX^{H}\bold{E}\BQ\bold{F}\BQ)=\BO(\frac{\mathcal{P}_{2}(\frac{1}{z})}{z^4}),
\end{equation}
\begin{equation}
\label{VABCQZQ}
\Var(\Tr \bold{A}\BX \bold{B}\BY\bold{C}\BY^{H}\bold{D}\BX^{H}\bold{E}\BQ\BZ\BZ^{H}\BQ)=\BO(\frac{\mathcal{P}_{2}(\frac{1}{z})}{z^4}),
\end{equation}
where $\mathcal{P}_{i}(\cdot)$ denotes an $i$-degree polynomial with positive coefficients, which are independent of $N$, $L$, $M$, and $z$.
\end{proposition}
\begin{proof}
The proof of Proposition~\ref{var_con} is given in Appendix~\ref{pro_var_con}.
\end{proof}
\begin{remark}
$\mathcal{P}(\frac{1}{z})$ is introduced to handle the integral over the approximation error in the following proof and analyze the order of $z$.
\end{remark}

\section{First-order Approximation: Expectation for the MI}
\label{sec_pro_fir}
In this section, we will give a detailed proof of the main result regarding the EMI over double-scattering channels using the Gaussian tools. The proof includes three steps, 

A. Convert the evaluation of the EMI to that of the expectation for the trace of the resolvent. 

B. Give an approximation for the expectation of the trace of the resolvent with the convergence rate $\BO(\frac{1}{N})$.

C. Return to evaluate the EMI.

\subsection{From MI to the Trace of the Resolvent}
The expectation of the MI can be written as~\cite{dumont2010capacity,zhang2021bias}

\begin{equation}
\E {C}(\sigma^2)=\int_{\sigma^2}^{\infty}\frac{N}{z}-\E\Tr\left(z\bold{I}_{N}+\BH\BH^{H} \right)^{-1}\mathrm{d}z=\int_{\sigma^2}^{\infty}\frac{N}{z}-\E\Tr\BQ(z)\mathrm{d}z.
\end{equation}
Thus, the problem turns to be the evaluation of the expectation for the trace of the resolvent $\E \Tr\BQ(z)$.

\subsection{The evaluation of $\E\Tr\BQ$}
In this step, we will show that $\E\Tr\BQ$ can be approximated by $\Tr\BG_{R}$, where $\BG_{R}$ is given in~(\ref{G_RST}).
For that purpose, we first introduce the auxiliary quantities as follows:

\begin{equation}
\label{int_qua}
\begin{aligned}
e_{\omega}&=\frac{\E\Tr\BZ\BZ^{H}\BQ}{M},~
\alpha_{\delta}=\frac{\E\Tr\FR\BQ}{L},~\alpha_{\omega}=\frac{\Tr\FS(\frac{1}{\alpha_{\delta}}\bold{I}_{L}+\alpha_{\overline{\omega}}\FS)^{-1}}{M},
~\alpha_{\overline{\omega}}=\frac{\Tr\FT\BG_{T,\alpha}}{M},
\\
\BG_{R,\alpha}&=\left(z\BI_{N}+\frac{\alpha_{\omega}\alpha_{\overline{\omega}}}{\alpha_{\delta}}\FR  \right)^{-1},~
\BG_{S,\alpha}=\left(\frac{1}{\alpha_{\delta}}\BI_{L}+\alpha_{\overline{\omega}}\FS \right)^{-1},~
\BG_{T,\alpha}=\left(\BI_{M}+e_{{\omega}}\FT \right)^{-1},~
\BF_{S,\alpha}=\frac{\BG_{S,\alpha}}{\alpha_{\delta}}.
\end{aligned}
\end{equation}
The idea here is to utilize $\Tr\BG_{R,\alpha}$ as the intermediate approximation between $\E\Tr\BQ$ and $\Tr\BG_{R}$.

\textit{Sketch of the proof}: The proof involves two steps:

1) In the first step, by utilizing the resolvent identity~(\ref{res_ide}), the integration by parts formula~(\ref{int_part}), and the variance control given in Proposition~\ref{var_con}, we show that for any matrix $\BA$ with bounded spectral norm, there holds true that
\begin{equation}
\nonumber
\E\Tr\BA\BQ=\Tr\BA\BG_{R,\alpha}+\BO(\frac{1}{N}).
\end{equation}

2) In the second step, we will use the Montel's theorem to show that $\Tr\BA\BG_{R,\alpha}$ can be further approximated by
\begin{equation}
\nonumber
\Tr\BA\BG_{R,\alpha}=\Tr\BA\BG_{R}+\BO(\frac{1}{N}).
\end{equation}

By steps 1) and 2), we can conclude that 
\begin{equation}
\nonumber
\E\Tr\BA\BQ=\Tr\BA\BG_{R}+\BO(\frac{1}{N}).
\end{equation}

\subsubsection{Step 1, A Weak Approximation of $\E\Tr\BA\BQ$} The approximation of $\E\Tr\BA\BQ$ can be obtained by $\BA\BG_{R,\alpha}$ according to the following lemma. 

\begin{lemma} 
\label{Q_TO_AL}
Given that~\textbf{A.1}-\textbf{A.3} hold true and $\bold{A}$ is a deterministic matrix whose spectral norm is uniformly bounded by $N$, the following approximation holds true
\begin{equation}
\begin{aligned}
\E\Tr\bold{A}\BQ=\E\Tr\bold{A} \BG_{R,\alpha}+\BO(\frac{1}{N}),
\end{aligned}
\end{equation}
where $\BG_{R,\alpha}$ is given in~(\ref{int_qua}). 
\end{lemma}
\begin{proof}
The proof of Lemma~\ref{Q_TO_AL} is given in Appendix~\ref{pro_Q_TO_AL}.
\end{proof}

The approximation $\Tr\BA\BG_{R,\alpha}$ involves the expectations of $\alpha_{\delta}$ and $e_{\omega}$, which have not been given explicitly. In next step, $\Tr\BA\BG_{R,\alpha}$ will be further approximated by a deterministic expression.

\subsubsection{Step 2, A Deterministic Approximation of $\Tr\BA\BG_{R,\alpha}$} In this step, 
$\Tr\BA\BG_{R,\alpha}$ will be further evaluated so that the ultimate approximation can be expressed with respect to $\FR$, $\FS$, $\FT$ and $z$. To this end, we have the following Lemma. 
\begin{lemma} 
\label{AL_TO_QR}
Given the same settings as Lemma~\ref{Q_TO_AL}, the following approximation for $\Tr\BA\BG_{R,\alpha}$ holds true
\begin{equation}
\begin{aligned}
\Tr\bold{A} \BG_{R,\alpha}=\Tr\bold{A} \BG_{R}+\BO(\frac{1}{N}),
\end{aligned}
\end{equation}
where $\BG_{R}$ is given in~(\ref{G_RST}).
\end{lemma}
\begin{proof}
The proof of Lemma~\ref{AL_TO_QR} is given in Appendix~\ref{pro_AL_TO_QR}.
\end{proof}

Therefore, by Lemma~\ref{Q_TO_AL} and Lemma~\ref{AL_TO_QR}, we can approximate $\E\Tr\BA\BQ$ as follows. Given $\E\Tr\BA\BQ\approx\Tr\BA\BG_{R,\alpha}$ and $\Tr\BA\BG_{R,\alpha}\approx \Tr\BA\BG_{R}$, we have $\E\Tr\BA\BQ\approx \Tr\BA\BG_{R}$ by the following theorem.
\begin{theorem} 
\label{QA_conv}
Given assumptions~\textbf{A.1}-\textbf{A.3} and $(\delta, \omega, \overline{\omega})$ as the unique positive solution of the system of equations~(\ref{basic_eq1}), there holds true that for any $\bold{A}$ with bounded spectral norm,
\begin{equation}
\begin{aligned}
\label{trace_AQ_conv}
\E \Tr\bold{A}\BQ=\Tr\bold{A}\BG_{R}+\BO(\frac{1}{N}).
\end{aligned}
\end{equation}
\end{theorem}

\begin{remark}
Theorem~\ref{QA_conv} indicates a $\BO(\frac{1}{N})$ convergence rate of $\E\Tr\bold{A}\BQ$, which follows from Nash-Poincar{\'e} inequality and the variance control. Such a convergence rate coincides with that of the single Rayleigh model~\cite{hachem2008new}. The convergence rate $\BO(\frac{1}{M})$ in~(\ref{trace_AQ_conv}) is essential in setting up the CLT for the MID to perform the finite-block length analysis of large MIMO systems~\cite{hoydis2015second,zhang2022second}.
\end{remark}

From Theorem~\ref{QA_conv}, we can obtain that
\begin{equation}
\E \Tr\BQ = \Tr\BG_{R}+\mathcal{O}(\frac{1}{N}),
\end{equation}
\begin{remark}
\cite[Corollay 2]{zhang2021large} showed the convergence of the expectation for Stieltjes transform by the replica method but the convergence rate was missing. Also, whether the convergence of the expectation for the trace of the resolvent holds true remains unclear. In this paper, the approximation of the trace of the resolvent is recovered by the Gaussian tools and the convergence rate is determined with rigorous justification. In particular, Theorem~\ref{QA_conv} shows that the deterministic approximation is a good approximation for the trace of the resolvent. Such a  convergence rate for the resolvent also holds for the MI. In fact, from the proof of above conclusions, the convergence rate $\BO(\frac{1}{N})$can be further written as $\BO(\frac{\mathcal{P}(\frac{1}{z})}{N^2z^2})$ when $z$ is bounded away from zero.
\end{remark}

\subsection{Back to Evaluate $C(z)$: Proof of Theorem~\ref{fir_the}}
\begin{proof} 
We will prove $\overline{C}(\sigma^2)$ shown in~(\ref{the_mean}) is the deterministic approximation for the EMI. For that purpose, we first rewrite~(\ref{the_mean}) as 
\begin{equation}
\begin{aligned}
\overline{C}(z)&=-N\log(z)+\log\det(z\bold{I}_{N}+\frac{ M \omega\overline{\omega}}{ L\delta}\bold{R} )+\log\det(\bold{I}_{L}+\delta\overline{\omega}\bold{S})
+\log\det(\bold{I}_{M}+\omega\bold{T})-2M\omega\overline{\omega}
\\
&=
-N\log(z)+f(z,\delta,\omega,\overline{\omega}).
\end{aligned}
\end{equation}
By the definition of $\delta$, $\omega$ and $\overline{\omega}$, we can observe that 
\begin{equation}
\label{fp_zero}
\frac{\partial f }{\partial \delta}=\frac{\partial f }{\partial \omega}=\frac{\partial f }{\partial \overline{\omega}}=0, 
\end{equation}
which indicates that 
\begin{equation}
\frac{\mathrm{d} \overline{C}(z) }{\mathrm{d} z }= -\frac{N}{z} +\frac{\partial f}{\partial z} +
\frac{\partial f}{\partial \delta}\frac{\partial \delta}{\partial z}+\frac{\partial f}{\partial \omega}\frac{\partial \omega}{\partial z}+\frac{\partial f}{\partial \overline{\omega}}\frac{\partial \overline{\omega}}{\partial z}
{=}-\frac{N}{z}+\Tr\left(z\bold{I}_{N}+\frac{M\omega\overline{\omega}}{L\delta}\FR \right)^{-1}.
\end{equation}
As a result, 
\begin{equation}
\overline{C}(\sigma^2)=\int_{\sigma^2}^{\infty}\left(\frac{N}{z}-\Tr\BG_{R}(z) \right)\mathrm{d}z.
\end{equation}
According to Theorem~\ref{QA_conv}, there holds true that $\E\Tr\left(z\bold{I}_{N}+\BH\BH^{H} \right)^{-1}=\Tr\BG_{R}(z)+\varepsilon({z})$, where $|\varepsilon(z)|\le \frac{\mathcal{P}(\frac{1}{z})}{z^2N}$. Since $\frac{\mathcal{P}(\frac{1}{z})}{z^2}$ is integrable on $(\sigma^2,\infty)$, we can obtain
\begin{equation}
\E{C}(\sigma^2)=\int_{\sigma^2}^{\infty}\frac{N}{z}-\E\Tr\left(z\bold{I}_{N}+\BH\BH^{H} \right)^{-1}\mathrm{d}z
=\overline{C}(\sigma^2)-\int_{\sigma^2}^{\infty} \varepsilon(z)\mathrm{d}z=\overline{C}(\sigma^2)+\BO(\frac{1}{N}).
\end{equation}
\end{proof}

\section{Second-order Analysis: Proof of the CLT}
\label{sec_pro_clt}
In this section, we present a detailed proof of the CLT in Theorem~\ref{sec_the} by investigating the characteristic function of the MI. As there are many computations involved in the proof, we first present some approximation rules.

\subsection{Necessary Quantities and Their Approximations}
\label{qua_eva}
We define the following quantities about the resolvents
\begin{equation}
\begin{aligned}
&
\kappa(\BA,\BB,\BC)=\frac{1}{M}\E\Tr\BA\BX\BB\BY\BC\BH^{H}\BQ,~~\mu(\BA,\BB)=\frac{1}{M}\E\Tr\BA\BX\BB\BZ^{H}\BQ,
\\
&\Gamma(\BA,\BB,\BC)=\frac{1}{M}\E\Tr\BA\BX\BB\BY\BC\BH^{H}\BQ\BZ\BZ^{H}\BQ,
~~\chi(\BA,\BB)=\frac{1}{M}\E\Tr\BA\BX\BB\BZ^{H}\BQ\BZ\BZ^{H}\BQ,
\\
&\Upsilon(\BA,\BB,\BC)=\frac{1}{M}\E\Tr\BA\BX\BB\BY\BC\BH^{H}\BQ\FR\BQ,
~~\zeta(\BA,\BB)=\frac{1}{M}\E\Tr\BA\BX\BB\BZ^{H}\BQ\FR\BQ,
\\
&
\chi=\chi(\LR,\MS),
~~\zeta=\zeta(\LR,\MS).
\end{aligned}
\end{equation}
$\kappa$ and $\mu$ will be referred to as the first-order trace-of-resolvents-related quantities because there is only one $\BQ$ in the expression. $\Gamma$, $\chi$, $\Upsilon$, and $\zeta$ will be named as the second-order trace-of-resolvents-related quantities due to the order of the $\BQ$. For brevity, we introduce the notation $\BO_{z}(A)=\BO(\frac{A\mathcal{P}(\frac{1}{z})}{z^2})$. All the quantities can be approximated by their deterministic approximations with error $\BO_{z}(\frac{1}{N^2})$ by the following two lemmas.

\begin{lemma} (Approximations of the first-order trace-of-resolvents-related quantities) 
\label{ka_mu}
Given~\textbf{A.1}-\textbf{A.3} hold true and $\BA$, $\BB$, $\BC$ are deterministic matrices with bounded norm, the following approximations can be obtained
\begin{equation}
\label{exp_ka}
\kappa(\BA,\BB,\BC)=\frac{1}{L}\Tr\BA\LR\BG_{R}\frac{1}{M}\Tr\MS\BB\BF_{S}\frac{1}{M}\Tr\RT\bold{C}\bold{G}_{T}+\BO_{z}(\frac{1}{N^2}),
\end{equation}
\begin{equation}
\label{exp_mu}
\mu(\BA,\BB)=\frac{1}{L}\Tr\BA\LR\BG_{R}\frac{1}{M}\Tr\MS\BB\BF_{S}+\BO_{z}(\frac{1}{N^2}).
\end{equation}
\end{lemma}
\begin{proof} The proof of Lemma~\ref{ka_mu} is similar to the derivation of~$\E\Tr\BA\BQ$, and is given in Appendix~\ref{pro_ka_mu}.
\end{proof}

\begin{lemma} 
\label{sec_app}
(Approximations of the second-order trace-of-resolvents-related quantities) Given the same settings as Lemma~\ref{ka_mu}, the following approximations can be obtained
    \begin{equation}
    \label{CHAB}
    \begin{aligned}
&\chi(\BA,\BB)
=\frac{\Tr\BA\LR\BG_{R}}{L}\frac{\Tr\FS^{\frac{3}{2}}\BB\BG_{S}^2}{M\delta}\frac{1-\frac{M \omega\overline{\omega}\nu_{R}}{L\delta^2}}{\Delta_{S}\Delta}
+\frac{M\nu_{S,I}\Tr\BG_{R}^2\FR^{\frac{3}{2}}\BA}{L^2\delta^3\Delta_{S}\Delta}\frac{\Tr\MS\bold{B}\BG_{S}}{M}
+\BO_{z}(\frac{1}{N^2}),
    \end{aligned}
    \end{equation}
    \begin{equation}
    \begin{aligned}
    \label{ZEAB}
&\zeta(\BA,\BB)
=\frac{\Tr\BA\FR^{\frac{3}{2}}\BG_{R}^2}{L\Delta} \frac{\Tr\MS\BB\BG_{S}}{\delta M}
-\frac{\nu_{R}\nu_{T,I}\Tr\BA\LR\BG_{R}}{L\delta^2\Delta_{S}\Delta}\frac{\Tr\BB\FS^{\frac{3}{2}}\BG_{S}^2}{M}
+\BO_{z}(\frac{1}{N^2}),
    \end{aligned}
    \end{equation}
        \begin{equation}
    \begin{aligned}
    \label{GAABC}
&\Gamma(\BA,\BB,\BC)
=\frac{ \Tr\RT\bold{C}\bold{G}_{T}}{M} \chi(\BA,\BB)- \frac{\Tr\BA\LR\BG_{R}}{L}\frac{\Tr\BB\MS\BF_{S}}{M}\frac{\Tr\BC\FT^{\frac{3}{2}}\BG^2_{T}}{M}\chi+\BO_{z}(\frac{1}{N^2}),
       \end{aligned}
    \end{equation}
      \begin{equation}
      \label{UPSABC}
    \begin{aligned}
&\Upsilon(\BA,\BB,\BC)=\frac{\Tr\bold{C}\RT\BG_{T}}{M}\zeta(\BA,\BB)-\frac{\Tr\BA\LR\BG_{R}}{L}\frac{\Tr\BB\MS\BF_{S}}{M}\frac{\Tr\BC\FT^{\frac{3}{2}}\BG^2_{T}}{M}\zeta+\BO_{z}(\frac{1}{N^2}),
    \end{aligned}
    \end{equation}
where
        \begin{equation}
        \label{zeta_chi_sim}
        \chi=\frac{1}{\Delta_{S}}(\nu_{S}+\frac{M\nu_{S,I}\nu_{R} }{L\Delta_{S}\Delta\delta^4} )+\BO_{z}(\frac{1}{N^2}),
        ~~~\zeta=\frac{\nu_{S,I}\nu_{R}}{\delta^2\Delta_{S}\Delta}+\BO_{z}(\frac{1}{N^2}).
      \end{equation}
\end{lemma}
\begin{proof}
The proof of Lemma~\ref{sec_app} is given in Appendix~\ref{pro_qua}.
\end{proof}
\begin{remark} (Potential applications) The evaluations in Lemma~\ref{ka_mu} and Lemma~\ref{sec_app}, which are shown to have an $\BO(\frac{1}{N^2})$ convergence rate, are useful in the analysis of other systems over double-scattering channel. In fact, many performance evaluation problems are ultimately converted to the resolvents related computation, e.g., channels matrix inversion based techniques including the MMSE receiver and RZF precoding. For example, the challenge for evaluating the EMI of a MIMO system with the MMSE receiver~\cite{artigue2011precoder} was the computation of the sum of variances for the diagonal entries of the resolvent matrix. In~\cite{kammoun2019asymptoticit}, the evaluation of the ergodic sum rate of a single-cell large-scale multiuser MIMO system with RZF precoding over Rician fading was converted to the computation of the second-order trace of the resolvents. Similar evaluations for double-scattering channels can be resolved by Lemma~\ref{sec_app}, which can also be used for the evaluations of other systems, e.g., the CLT for the SNR of minimum variance distortionless response (MVDR) filters~\cite{rubio2012clt} and the CLT for the SINR of the MMSE receiver~\cite{kammoun2009central}.
\end{remark}

\subsection{Proof of Theorem~\ref{sec_the}}
Let $z=\sigma^2$ and define $I(z)=C(z)+Nz= \log\det(z\bold{I}_{N}+\BH\BH^{H})$. There holds true that $\underline{I}(z)=\underline{C}(z)$. In the following, we will show the asymptotic Gaussianity of the MI over double-scattering channel by investigating the characteristic function of $C(\sigma^2)$, which is given by
\begin{equation}
\label{cha_fun}
\Psi(u,z)=\E e^{\jmath u [C(\sigma^2)-\E C(\sigma^2) ]}=\E e^{\jmath u [I(z)-\E I(z) ]}=\E e^{\jmath u \underline{I(z)}}=\E \Phi(u,z),
\end{equation}
where $\Phi(t,z)=e^{\jmath u \underline{I(z)}}$. To show the asymptotic Gaussianity, we need to show that the characteristic function~(\ref{cha_fun}) converges to the characteristic function of Gaussian distribution, i.e., 
\begin{equation}
\label{cha_con}
\Psi(u,z)\xrightarrow{N\rightarrow \infty} e^{-\frac{ V(z)u^2}{2} },
\end{equation}
where $V(z)$ represents the asymptotic variance. Due to the difficulty in handling the logarithm of a determinant in $C(\sigma^2)$, we will first turn to handle the derivative of~(\ref{cha_fun}) with respect to $z$,
\begin{equation}
\frac{\partial \Psi(u,z)}{\partial z}=\jmath u \E \underline{ \Tr(\BH\BH^{H}+z\BI_{N})^{-1}} \Phi(u, z),
\end{equation}
which is related to the trace of the resolvent. Now if we can show the convergence of the derivative,
\begin{equation}
\label{de_conv}
\frac{\partial \Psi(u,z)}{\partial z} =\jmath u \E \underline{\Tr\BQ}\Phi(u,z)
=-\frac{\partial V(z)}{2\partial z} u^2  \Psi(u,z) +\BO(\frac{\mathcal{P}(\frac{1}{z})}{z^2 N}),
\end{equation}
we will be able to prove the convergence of the characteristic function by solving the differential equation in~(\ref{de_conv}) and obtain
\begin{equation}
 \Psi(u,z)e^{\frac{V(z)u^2}{2}} =1- \int_{z}^{\infty} \BO(\frac{\mathcal{P}(\frac{1}{x})e^{\frac{V(x)u^2}{2}}}{x^2 N}) \mathrm{d}x
\end{equation}
so that $ \Psi(u,z)\xrightarrow[]{N\rightarrow \infty} e^{-\frac{V(z)u^2}{2}}$. Finally, the CLT will be concluded from the convergence of the characteristic function.


The proof includes four main steps, which are

1) We first show that $\frac{\partial \Psi(u,z)}{\partial z}$ can be represented as a linear combination of $\E\underline{\Tr\FR\BQ}\Phi(u,z)$, $\E\underline{\Tr\BZ\BZ^{H}\BQ}\Phi(u,z)$ and the quantities we introduced in Section~\ref{qua_eva}. 

2) By using the approximations for the necessary quantities in Section~\ref{qua_eva}, we give the approximations for $\E\underline{\Tr\FR\BQ}\Phi(u,z)$ and $\E\underline{\Tr\BZ\BZ^{H}\BQ}\Phi(u,z)$. Finally, the approximation of $\frac{\partial \Psi(u,z)}{\partial z}$ with $\BO(\frac{1}{N})$ error will be given.

3) From the convergence of the derivative $\frac{\partial \Psi(u,z)}{\partial z}$, we will show that~(\ref{de_conv}) holds true and further prove~(\ref{cha_con}). Based on the evaluation of $ \frac{\partial V(z) }{\partial z}$, we will determine the asymptotic variance $V(z)$. 

4) We will conclude the CLT by the convergence of the characteristic function.

5) Finally, we prove that the convergence rate of the variance is $\BO(\frac{1}{N})$.


\subsubsection{Step 1, Decomposition of $\frac{\partial \Psi(u,z)}{\partial z}$}
According to the derivative rules in~(\ref{log_Q_der}), the derivative of the characteristic function with respect to $z$ is given by
\begin{equation}
\label{trQpsi}
\frac{\partial \Psi(u,z)}{\partial z} =\jmath u \E \underline{\Tr\BQ}\Phi(u,z).
\end{equation}
We will first focus on the evaluation of $\E{\Tr\BH\BH^{H}\BQ}\Phi(u,z)$ so that we can obtain $\E\underline{\Tr\BQ}\Phi(u,z) $ by the resolvent identity~(\ref{res_ide}). According to the integration by parts formula with respect to $Y_{p,j}^{*}$ and~(\ref{log_Q_der}), the following identity holds true,
\begin{equation}
\label{QHHP0}
\begin{aligned}
&\E [\RT\BY^{H}\BZ^{H} \BQ]_{j,i} [\BZ]_{i,q}  [\BY\RT]_{q,j} \Phi(u,z)
=\sum_{p}\E [\RT\BY^{H}]_{j,p} [\BZ^{H} \BQ]_{p,i} [\BZ]_{i,q}  [\BY\RT]_{q,j} \Phi(u,z)
\\
&
= \E \{\frac{1}{M} t_{j} [\BZ^{H} \BQ]_{q,i}[\BZ]_{i,q} \Phi(u,z)
-\frac{t_{j}\Tr\BZ\BZ^{H}\BQ }{M}   [\RT\BY^{H}\BZ^{H}\BQ]_{j,i}[\BZ]_{i,q}  [\BY\RT]_{q,j}\Phi(u,z)
\\
&
+\frac{\jmath u}{M}[\RT\BH^{H}\BQ\BZ\BZ^{H} \BQ]_{j,i} [\BZ]_{i,q}  [\BY\FT]_{q,j}  \Phi(u,z)\}.
\end{aligned}
\end{equation}
By adding $ e_{\omega}t_{j}\E [\RT\BY^{H}\BZ^{H}\BQ]_{j,i}[\BZ]_{i,q}  [\BY\RT]_{q,j}\Phi(u,z)$,
multiplying $[(\BI_{M}+e_{{\omega}}\FT)^{-1}]_{j,j}$ on both sides of~(\ref{QHHP0}), and summing over $j$, we can obtain
\begin{equation}
\begin{aligned}
\label{QHHP1}
&\E [\BY \FT\BY^{H}\BZ^{H} \BQ]_{q,i} [\BZ]_{i,q} \Phi(u,z) =\alpha_{\overline{\omega}}\E [\BZ^{H} \BQ]_{q,i}[\BZ]_{i,q} \Phi(u,z)
\\
&-\frac{1}{M} \cov( {\Tr\BZ\BZ^{H}\BQ},[\BY\FT^{\frac{3}{2}}(\bold{I}_{M}+e_{{\omega}}\FT)^{-1}\BY^{H}\BZ^{H}\BQ ]_{q,i}[\BZ]_{i,q} \Phi(u,z))
\\
&+\frac{\jmath u}{M}\E [\BY\FT^{\frac{3}{2}}(\bold{I}_{M}+e_{\omega}\FT)^{-1}\BH^{H}\BQ\BZ\BZ^{H} \BQ]_{q,i} [\BZ]_{i,q}   \Phi(u,z).
\end{aligned}
\end{equation}
By using the integration by parts formula, we perform the same operations over $\E [\BZ^{H} \BQ]_{q,i}[\BZ]_{i,q} \Phi(u,z)$ with respect to $X_{m,q}^{*}$ to obtain
\begin{equation}
\begin{aligned}
\label{QZZP}
&
\E [\BZ^{H} \BQ]_{q,i}[\BZ]_{i,q} \Phi(u,z)=\E \sum_{m}[\LR\BX^{*}\MS]_{m,q} [\BQ]_{m,i}[\BZ]_{i,q} \Phi(u,z) =\E \{\frac{s_{q}}{L}[\FR\BQ]_{i,i} \Phi(u,z)
\\
&
-\frac{s_{q}\Tr\FR\BQ}{L} [\BY\RT\BH^{H}\BQ]_{q,i}[\BZ]_{i,q} \Phi(u,z)
+\frac{\jmath u}{L} [\BZ]_{i,q} [ \FS\BY\RT \BH^{H} \BQ\FR\BQ ]_{q,i}  \Phi(u,z)\}.
\end{aligned}
\end{equation}
By plugging~(\ref{QZZP}) into~(\ref{QHHP1}) to replace $\E [\BZ^{H} \BQ]_{q,i}[\BZ]_{i,q} \Phi(u,z)$ and multiplying $[(\BI_{L}+\alpha_{\delta}\alpha_{\overline{\omega}}\FS)^{-1}]_{q,q}$ on both sides of~(\ref{QHHP1}), we have
\begin{equation}
\begin{aligned}
\label{HHQP}
&\E [\BY\RT\BH^{H}\BQ]_{q,i}[\BZ]_{i,q}\Phi(u,z)=\E \{\frac{s_{q}\alpha_{\overline{\omega}}}{L(1+s_{q}\alpha_{\delta}\alpha_{\overline{\omega}})}[\FR\BQ]_{i,i}\Phi(u,z) 
\\
&
-\frac{\alpha_{\overline{\omega}}}{L} \cov( {\Tr\FR\BQ},[\FS\BF_{S,\alpha}\BY\RT\BH^{H}\BQ]_{q,i}[\BZ]_{i,q}\Phi(u,z) )
+\frac{\jmath u\alpha_{\overline{\omega}}}{L} [\BZ]_{i,q} [ \FS\BF_{S,\alpha} \BY\RT \BH^{H} \BQ\FR\BQ ]_{q,i} \Phi(u,z) 
\\
&
-\frac{1}{M} \cov( {\Tr\BZ\BZ^{H}\BQ},[\BF_{S,\alpha}\BY\FT^{\frac{3}{2}}\BG_{T,\alpha}\BY\BZ^{H}\BQ ]_{q,i}[\BZ]_{i,q}\Phi(u,z))
\\
&+\frac{\jmath u}{M}[\BY\FT^{\frac{3}{2}}\BG_{T,\alpha}\BH^{H}\BQ\BZ\BZ^{H} \BQ]_{q,i} [\BZ\BF_{S,\alpha}]_{i,q}   \Phi(u,z) \}.
\end{aligned}
\end{equation}
Summing over $q$, the following equation holds true by the resolvent identity~(\ref{res_ide}), 
\begin{equation}
\label{QPII}
\begin{aligned}
&\E [\BQ]_{i,i}\Phi(u,z)=\E\{[(z\BI_{N}+\frac{M \alpha_{\overline{\omega}}\alpha_{\omega} }{L \alpha_{\delta}}\FR)^{-1}]_{i,i}\Phi(u,z)
+\frac{\alpha_{\overline{\omega}}}{L}\cov({\Tr\FR\BQ}\Phi(u,z), [\BG_{R,\alpha} \BZ\FS\BF_{S,\alpha}\BY\RT\BH\BQ]_{i,i})
\\
&+\frac{1}{M}  \cov( {\Tr\BZ\BZ^{H}\BQ},[\BG_{R,\alpha}\BZ\BF_{S,\alpha}\BY\FT^{\frac{3}{2}}\BG_{T,\alpha}\BY\BZ^{H}\BQ ]_{i,i}\Phi(u,z))
-\frac{\jmath u \alpha_{\overline{\omega}}}{L}[\BG_{R,\alpha}\BZ\FS\BF_{S,\alpha} \BY\RT \BH^{H} \BQ\FR\BQ]_{i,i}\Phi(u,z))
\\
&-\frac{\jmath u}{M}[\BG_{R,\alpha}\BZ\BF_{S,\alpha}\BY\FT^{\frac{3}{2}}\BG_{T,\alpha}\BH^{H}\BQ\BZ\BZ^{H} \BQ]_{i,i}  \Phi(u,z)\}.
\end{aligned}
\end{equation}
According to the variance control in Proposition~\ref{var_con}, the second covariance term in~(\ref{QPII}) is bounded by
\begin{equation}
\label{for_var_con}
\begin{aligned}
& |\cov( {\Tr\BZ\BZ^{H}\BQ}, \gamma\Phi(u,z))
  -\E \underline{\Tr\BZ\BZ^{H}\BQ}\Phi(u,z)\E \gamma|
=|\E \underline{\Tr\BZ\BZ^{H}\BQ}\Phi(u,z)\underline{\gamma} |
 \le \E^{\frac{1}{2}} |\underline{\Tr\BZ\BZ^{H}\BQ}|^2  \E^{\frac{1}{2}}|\underline{\gamma}|^2=\BO(\frac{\mathcal{P}_{2}(\frac{1}{z})}{z^2 N}),
\end{aligned}
\end{equation}
where $\gamma=\frac{1}{M}\Tr\BG_{R,\alpha}\BZ\BF_{S,\alpha}\BY\FT^{\frac{3}{2}}\BG_{T,\alpha}\BY\BZ^{H}\BQ$. The first covariance term can be handled similarly to obtain
\begin{equation}
\begin{aligned}
\frac{M}{L}\cov( {\Tr\FR\BQ}\Phi(u,z),\beta )
&=\frac{M}{L}[\E\underline{\Tr\FR\BQ}\Phi(u,z)\E {\beta}+\E\Tr\FR\BQ\E\Phi(u,z)\underline{\beta}+\E\underline{\Tr\FR\BQ}\Phi(u,z)\underline{\beta}  ]
\\
&
=\frac{M}{L}\E\underline{\Tr\FR\BQ}\Phi(u,z)\E\beta +\BO(\frac{\mathcal{P}_{2}(\frac{1}{z})}{z^2 N}),
\end{aligned}
\end{equation}
where $\beta=\frac{1}{M}\Tr\BG_{R,\alpha} \BZ\FS\BF_{S,\alpha}\BY\RT\BH\BQ$.

By summing~(\ref{QPII}) over $i$, $\E \Tr\BQ\Phi(u,z) $ can be represented by
\begin{equation}
\label{TRQP}
\begin{aligned}
&
\E \Tr\BQ\Phi(u,z)=\Tr\BG_{R,\alpha}\E\Phi(u,z)+\frac{M\alpha_{\overline{\omega}}}{L}\kappa(\LR\BG_{R,\alpha},\FS^{\frac{3}{2}}\BF_{S,\alpha},\RT)\E\underline{\Tr\FR\BQ}\Phi(u,z)
\\
&
+\kappa(\LR\BG_{R,\alpha},\FS^{\frac{1}{2}}\BF_{S,\alpha},\FT^{\frac{3}{2}}\BG_{T})\E\underline{\Tr\BZ\BZ^{H}\BQ}\Phi(u,z)
-\frac{\jmath u\alpha_{\overline{\omega}} M }{L}\Upsilon(\LR\BG_{R,\alpha},\FS^{\frac{3}{2}}\BF_{S,\alpha},\RT)\E \Phi(u,z)
\\
&-\jmath u \Gamma(\LR\BG_{R,\alpha},\FS^{\frac{1}{2}}\BF_{S,\alpha},\FT^{\frac{3}{2}}\BG_{T,\alpha})\E \Phi(u,z)+\BO_{z}(\frac{1}{N}).
\end{aligned}
\end{equation}
By far, $\E \Tr\BQ\Phi(u,z)$ has been decomposed as a linear combination of $\kappa$, $\Upsilon$, $\Gamma$, $\E \underline{\Tr\FR\BQ}\Phi(u,z)$, and $\E\underline{\Tr\BZ\BZ^{H}\BQ}\Phi(u,z)$, which will be evaluated in next step.

\subsubsection{Step 2, Approximation of $\frac{\partial \Psi(u,z)}{\partial z}$}

In this step, we first show that $\E\underline{\Tr\BZ\BZ^{H}\BQ}\Phi(u,z)$ can be represented by $\E\underline{\Tr\FR\BQ}\Phi(u,z)$ and the necessary quantities. Then we will construct an equation with respect to $\E\underline{\Tr\FR\BQ}\Phi(u,z)$ to obtain its approximation. Finally, $\frac{\partial \Psi(u,z)}{\partial z}$ can be obtained by plugging the above results into~(\ref{TRQP}) and then~(\ref{trQpsi}).
\begin{lemma}
\label{lem_ZZQP}
If assumptions \textbf{A.1} to \textbf{A.3} hold true, the evaluation of~$\E\underline{\Tr\BZ\BZ^{H} \BQ} \Phi(u,z)$ can be given by
\begin{equation}
\label{QZZP_}
\begin{aligned}
&
\E\underline{\Tr\BZ\BZ^{H} \BQ} \Phi(u,z)=
\frac{1}{\Delta_{S}}[(\frac{M \omega}{L\delta}-\frac{M\overline{\omega}\nu_{S} }{L\delta})
\E\underline{\Tr\FR\BQ}  \Phi(u,z) 
+\frac{\jmath u M}{L} \Upsilon(\LR,\FS^{\frac{3}{2}}\BF_{S},\RT ) \E  \Phi(u,z) 
\\
&
-\jmath u \delta  \Gamma(\LR,\FS^{\frac{3}{2}}\BF_{S},\FT^{\frac{3}{2}}\BG_{T} )\E \Phi(u,z) ]
+\BO_{z}(\frac{1}{N}).
\end{aligned}
\end{equation}
\end{lemma}
\begin{proof}
The proof of Lemma~\ref{lem_ZZQP} is given in Appendix~\ref{pro_lem_ZZQP}.
\end{proof}
By substituting~(\ref{QZZP_}) into~(\ref{TRQP}) to replace $\E\underline{\Tr\BZ\BZ^{H} \BQ} \Phi(u,z)$ by $\E\underline{\Tr\FR \BQ} \Phi(u,z)$, and the $\Upsilon$ and $\Gamma$ related terms, we can obtain
\begin{equation}
\begin{aligned}
\label{TRQRQ}
&
\E \underline{\Tr\BQ}\Phi(u,z)=(\Tr\BG_{R,\alpha}\E\Phi(u,z)- \E\Tr\BQ\E\Phi(u,z))  +K_{I} \E\underline{\Tr\FR\BQ}\Phi(u,z)
+[-\frac{\alpha_{\overline{\omega}} M }{L}\Upsilon(\LR\BG_{R,\alpha},\FS^{\frac{3}{2}}\BF_{S,\alpha},\RT)
\\
&
- \Gamma(\LR\BG_{R,\alpha},\FS^{\frac{3}{2}}\BF_{S,\alpha},\FT^{\frac{3}{2}}\BG_{T,\alpha})
-\frac{\kappa(\LR\BG_{R,\alpha},\FS^{\frac{1}{2}}\BF_{S,\alpha},\FT^{\frac{3}{2}}\BG_{T,\alpha})\delta}{\Delta_{S}} \Gamma(\LR,\FS^{\frac{3}{2}}\BF_{S},\FT^{\frac{3}{2}}\BG_{T} )
\\
&
+\frac{M\kappa(\LR\BG_{R,\alpha},\MS\BF_{S,\alpha},\FT^{\frac{3}{2}}\BG_{T,\alpha})}{L\Delta_{S}}\Upsilon(\LR,\FS^{\frac{3}{2}}\BF_{S},\RT )]\jmath u \E \Phi(u,z)
+\BO_{z}(\frac{1}{N})
\\
&\overset{(a)}{=}F_{1}+\theta\nu_{R,I} \E\underline{\Tr\FR\BQ}\Phi(u,z)+(W_{1}+W_{2}+W_{3}+W_{4})\jmath u \E \Phi(u,z)+\BO_{z}(\frac{1}{N}),
\end{aligned}
\end{equation}
where $|F_1|=|\Tr\BG_{R,\alpha}-\E\Tr\BQ||\E\Phi(u,z) |=\BO(\frac{1}{N})$ and step $(a)$ follows from the following evaluation according to Lemma~\ref{ka_mu}
\begin{equation}
\begin{aligned}
K_{I}&=\frac{M\alpha_{\overline{\omega}}}{L}\kappa(\LR\BG_{R,\alpha},\FS^{\frac{3}{2}}\BF_{S,\alpha},\RT)+ \frac{M \nu_{S,I} }{L\delta^2\Delta_{S}}\kappa(\LR\BG_{R,\alpha},\FS^{\frac{1}{2}}\BF_{S,\alpha},\FT^{\frac{3}{2}}\BG_{T})
\\
&
=\frac{M\omega\overline{\omega}\nu_{R,I}}{L\delta^2}-\frac{M\overline{\omega}\nu_{R,I}\nu_{S,I}}{L\delta^3}
+\frac{M\nu_{R,I} \nu_{S,I}^2\nu_{T} }{L\delta^4\Delta_{S}}+\BO_{z}(\frac{1}{N^2})
=\theta\nu_{R,I}+\BO_{z}(\frac{1}{N^2}).
\end{aligned}
\end{equation}

By far, only $\E\underline{\Tr\FR\BQ}\Phi(u,z)$ remains unknown in~(\ref{TRQRQ}). By multiplying $r_{i}$ on both sides of~(\ref{QPII}) and summing over $i$, we can obtain the following equation with respect to $\E\underline{ \Tr\FR\BQ}\Phi(u,z)$,
\begin{equation}
\label{QRP_eq}
\begin{aligned}
&
\E \underline{\Tr\FR\BQ}\Phi(u,z)=(\Tr\FR\BG_{R,\alpha}-\E\Tr\FR\BQ)\E\Phi(u,z)
+K_{R} \E\underline{\Tr\FR\BQ}\Phi(u,z)
+[-\frac{\alpha_{\overline{\omega}} M }{L}\Upsilon(\FR^{\frac{3}{2}}\BG_{R,\alpha},\FS^{\frac{3}{2}}\BF_{S,\alpha},\RT)
\\
&
-\Gamma(\FR^{\frac{3}{2}}\BG_{R,\alpha},\FS^{\frac{1}{2}}\BF_{S,\alpha},\FT^{\frac{3}{2}}\BG_{T,\alpha})
-\frac{\delta}{\Delta_{S}} \kappa(\FR^{\frac{3}{2}}\BG_{R,\alpha},\FS^{\frac{1}{2}}\BF_{S,\alpha},\FT^{\frac{3}{2}}\BG_{T})\Gamma(\LR,\FS^{\frac{3}{2}}\BF_{S},\FT^{\frac{3}{2}}\BG_{T} )
\\
&
+\frac{M}{L\Delta_{S}}\kappa(\FR^{\frac{3}{2}}\BG_{R,\alpha},\FS^{\frac{1}{2}}\BF_{S,\alpha},\FT^{\frac{3}{2}}\BG_{T})\Upsilon(\LR,\FS^{\frac{3}{2}}\BF_{S,\alpha},\RT )]\jmath u \E \Phi(u,z)
+\BO_{z}(\frac{1}{N})
\\
&
=
E_{1}+E_{2}+[X_{1}+X_{2}+X_{3}+X_{4}]\jmath u \E \Phi(u,z)+\BO_{z}(\frac{1}{N}),
\end{aligned}
\end{equation}
where 
\begin{equation}
K_{R}=\frac{M\alpha_{\overline{\omega}}}{L}\kappa(\FR^{\frac{3}{2}}\BG_{R,\alpha},\FS^{\frac{3}{2}}\BF_{S,\alpha},\RT)+ \frac{M \nu_{S,I} }{L\delta^2\Delta_{S}}\kappa(\FR^{\frac{3}{2}}\BG_{R,\alpha},\FS^{\frac{1}{2}}\BF_{S,\alpha},\FT^{\frac{3}{2}}\BG_{T,\alpha}) ]=\theta\nu_{R}+\BO_{z}(\frac{1}{N^2})=1-\Delta+\BO_{z}(\frac{1}{N^2}),
\end{equation}
and $\Delta$ is given in Table~\ref{var_list}. Furthermore, $|E_1|=|\Tr\BG_{R,\alpha}-\E\Tr\FR\BQ)\E\Phi(u,z)|\le |\Tr\BG_{R,\alpha}-\E\Tr\FR\BQ|=\BO_{z}(\frac{1}{N}) $ holds true by Lemma~\ref{AL_TO_QR}. By moving $E_{2}$ to the LHS of~(\ref{QRP_eq}), $\E \underline{\Tr\FR\BQ}\Phi(u,z)$ can be represented by
\begin{equation}
\label{TRRQP}
\E \underline{\Tr\FR\BQ}\Phi(u,z)=\frac{(X_{1}+X_{2}+X_{3}+X_{4})\jmath u \E \Phi(u,z)}{\Delta} 
+\BO_{z}(\frac{1}{N}).
\end{equation}
By~(\ref{TRRQP}), the evaluation of $\frac{\partial \Psi(u,z)}{\partial z}$ is finally converted to the evaluation of the quantities related to $\Gamma$ and $\Upsilon$. Following the definition of $\delta$, $\omega$, and $\overline{\omega}$, we first introduce some useful equations that will be widely used in the derivation as follows
\begin{subequations}
\begin{align}
\omega&=\frac{\nu_{S,I}}{\delta}+\overline{\omega}\nu_{S},\label{nu_S_eqn}
\\
\overline{\omega}&=\nu_{T,I}+\omega\nu_{T}\label{nu_T_eqn},
\\
\nu_{S}&=\frac{\eta_{S,I}}{\delta}+\overline{\omega}\eta_{S}\label{eta_S_eqn},
\\
\nu_{T}&=\eta_{T,I}+\omega\eta_{T}\label{eta_T_eqn},
\\
\frac{\nu_{S,I}\nu_T}{\delta\Delta_{S}}&=\overline{\omega}-\frac{\nu_{T,I}}{\Delta_S}\label{DeltaS_eqn}.
\end{align}
\end{subequations}
By Lemma~\ref{sec_app}, $X_{i}$, $i=1,2,3,4$ can be approximated by
\begin{subequations}
\label{exp_x1}
\begin{align}
X_{1}&=-\frac{\alpha_{\overline{\omega}} M }{L}\Upsilon(\FR^{\frac{3}{2}}\BG_{R,\alpha},\FS^{\frac{3}{2}}\BF_{S,\alpha},\RT)\nonumber
\\
&=-\frac{M\overline{\omega}  }{L}[ \overline{\omega}\zeta(\FR^{\frac{3}{2}}\BG_{R,\alpha},\FS^{\frac{3}{2}}\BF_{S,\alpha})-\frac{\nu_{R}\nu_{S}\nu_{T}\zeta}{\delta^2}    ] +\BO_{z}(\frac{1}{N^2})\label{X_1_step1}
\\
&=-\frac{ M \overline{\omega}}{L}[ \frac{ \overline{\omega}\eta_{R}\nu_{S}}{\delta^2\Delta}-\frac{\nu_{R}^2 \eta_{S}\nu_{T,I}\overline{\omega}}{\delta^2\Delta_{S}\Delta}-\frac{\nu_{R}^2\nu_{S}\nu_{S,I}\nu_{T}}{\delta^4\Delta_{S}\Delta}   ]+ \BO_{z}(\frac{1}{N^2})\label{X_1_step2}
\\
&=-\frac{M\overline{\omega}}{L} (\frac{\overline{\omega}\eta_{R}\nu_{S} }{\delta^2\Delta}  +  \frac{\nu_{R}^2\eta_{S}\nu_{T}{\omega}\overline{\omega}}{\delta^3\Delta_{S}\Delta}
-\frac{\nu_{R}^2\eta_{S}\overline{\omega}^2}{\delta^3\Delta_{S}\Delta}
-\frac{\nu_{R}^2\nu_{S}\nu_{S,I}\nu_{T}}{\delta^4\Delta_{S}\Delta}
 ) +\BO_{z}(\frac{1}{N^2})\label{X_1_step3}
\\
&=-\frac{M\overline{\omega}}{L} (\frac{\overline{\omega}\eta_{R}\nu_{S} }{\delta^2\Delta}  +  \frac{\nu_{R}^2\eta_{S}\nu_{S,I}\nu_{T}\overline{\omega}}{\delta^4\Delta_{S}\Delta}-\frac{\overline{\omega}^2\nu_{R}^2\eta_{S}}{\delta^3\Delta}
-\frac{\nu_{R}^2\nu_{S}\nu_{S,I}\nu_{T}}{\delta^4\Delta_{S}\Delta}
 ) +\BO_{z}(\frac{1}{N^2})\label{X_1_step4}
 \\
&=-\frac{M}{L}[\frac{\overline{\omega}^2\nu_{S}\eta_{R}}{\delta^2\Delta}-\frac{\overline{\omega}^3 \eta_{S} \nu_{R}^2}{\delta^3\Delta}-\frac{\overline{\omega}\eta_{S,I}\nu_{R}^2\nu_{S,I}\nu_{T}}{\delta^5\Delta_{S}\Delta }]+\BO_{z}(\frac{1}{N^2})\label{X_1_step5}
\\
& =-\frac{M}{L}[\frac{\overline{\omega}^2\nu_{S}\eta_{R}}{\delta^2\Delta}-\frac{\overline{\omega}^3 \eta_{S} \nu_{R}^2}{\delta^3\Delta}-\frac{\omega\overline{\omega}\eta_{S,I}\nu_{R}^{2}\nu_{T}}{\delta^4\Delta_{S}\Delta }+\frac{\overline{\omega}^2\eta_{S,I}\nu_{S}\nu_{R}^2\nu_{T}}{\delta^4\Delta_{S}\Delta }]+\BO_{z}(\frac{1}{N^2})\label{X_1_step6}
  \\
&= -\frac{M}{L}[ \frac{\overline{\omega}^2\nu_{S}\eta_{R}}{\delta^2\Delta}-\frac{\overline{\omega}^3 \eta_{S} \nu_{R}^2}{\delta^3\Delta}+\frac{\overline{\omega} \nu_{T,I}\eta_{S,I}\nu_{R}^2}{\delta^4\Delta_{S}\Delta }
-\frac{\overline{\omega}^2\eta_{S,I}\nu_{R}^2}{\delta^4\Delta }]+\BO_{z}(\frac{1}{N^2})\label{X_1_step7}
\\
&=-[\frac{M\overline{\omega}^2\nu_{S}\eta_{R}}{L\delta^2\Delta}-\frac{M\overline{\omega}^2 \nu_{S} \nu_{R}^2}{L\delta^3\Delta}
+\frac{M \nu_{R}^2\nu_{S,I}\nu_{T,I}}{L\delta^4\Delta_{S}\Delta }
-\frac{M \eta_{S,I,I}\nu_{R}^2\nu_{T,I}}{L\delta^5\Delta_{S}\Delta }]+\BO_{z}(\frac{1}{N^2})\label{X_1_step8}.
\end{align}
\end{subequations}
~(\ref{X_1_step1}) follows from~(\ref{UPSABC}). ~(\ref{X_1_step2}) is obtained by plugging~(\ref{ZEAB}) and~(\ref{zeta_chi_sim}) into~(\ref{X_1_step1}). (\ref{X_1_step3}) can be obtained by decomposing the second term of (\ref{X_1_step2}) using~(\ref{nu_T_eqn}).~(\ref{X_1_step4}) is obtained by handling the second and third terms in~(\ref{X_1_step3}) using~(\ref{nu_S_eqn}) as
\begin{equation}
 \frac{\nu_{R}^2\eta_{S}\nu_{T}{\omega}\overline{\omega}}{\delta^3\Delta_{S}\Delta}
-\frac{\nu_{R}^2\eta_{S}\overline{\omega}^2}{\delta^3\Delta_{S}\Delta}
= \frac{\nu_{R}^2\eta_{S}\nu_{S,I}\nu_{T}\overline{\omega}}{\delta^4\Delta_{S}\Delta}+ \frac{\nu_{R}^2\eta_{S}\nu_{S}\nu_{T}\overline{\omega}^2}{\delta^3\Delta_{S}\Delta}
-\frac{\nu_{R}^2\eta_{S}\overline{\omega}^2}{\delta^3\Delta_{S}\Delta}
=\frac{\nu_{R}^2\eta_{S}\nu_{S,I}\nu_{T}\overline{\omega}}{\delta^4\Delta_{S}\Delta}-\frac{\nu_{R}^2\eta_{S}\overline{\omega}^2}{\delta^3\Delta}.
\end{equation}
(\ref{X_1_step5}) follows by combining the second and fourth terms in~(\ref{X_1_step4}) using~(\ref{eta_S_eqn}).~(\ref{X_1_step6}) is obtained by decomposing the third term of~(\ref{X_1_step5}) using~(\ref{nu_S_eqn}).~(\ref{X_1_step7}) results from applying~(\ref{nu_T_eqn}) to the third term of~(\ref{X_1_step6}).~(\ref{X_1_step8}) can be obtained by applying~$\nu_{S,I}=\frac{\eta_{S,I,I}}{\delta}+\overline{\omega}\eta_{S,I}$ to the third term and combining the second and fourth terms of~(\ref{X_1_step7}). By~(\ref{CHAB}) and~(\ref{GAABC}), we can obtain the evaluations of $X_2$ and $X_3$ below.
\begin{equation}
\label{exp_x2}
\begin{aligned}
X_{2}&=-\Gamma(\FR^{\frac{3}{2}}\BG_{R,\alpha},\FS^{\frac{1}{2}}\BF_{S,\alpha},\FT^{\frac{3}{2}}\BG_{T})
\\
&=-\nu_{T} [\frac{M\nu_{S,I}^2\eta_{R}}{L\delta^4
 \Delta_{S} \Delta}+\frac{\delta}{\delta^3 \Delta \Delta_{S}}(1-\frac{M \omega\overline{\omega}}{L\delta^2}\nu_{R})\nu_{R}\eta_{S,I}]
  +\frac{1}{\Delta_{S}\delta^2 }(\nu_{S}+\frac{M\nu_{S,I}^2 \nu_{R}}{L\Delta\Delta_{S} \delta^4 })\nu_{R}\eta_{T}\nu_{S,I}+\BO_{z}(\frac{1}{N^2}),
\end{aligned}
\end{equation}
\begin{equation}
\label{exp_x3}
\begin{aligned}
X_{3}&=-\frac{\delta}{\Delta_{S}} \kappa(\FR^{\frac{3}{2}}\BG_{R,\alpha},\FS^{\frac{1}{2}}\BF_{S,\alpha},\FT^{\frac{3}{2}}\BG_{T})\Gamma(\LR,\FS^{\frac{3}{2}}\BF_{S},\FT^{\frac{3}{2}}\BG_{T} )
\\
&
=-\frac{\nu_{R}\nu_{S,I}\nu_{T} }{\delta\Delta_{S}}[\frac{M \nu_{R}\nu_{S,I}\nu_{S} \nu_{T} }{L\delta^4\Delta_{S}\Delta}+\frac{\eta_{S}\nu_{T}}{\delta\Delta_{S}\Delta}(1-\frac{M\omega\overline{\omega}\nu_{R}}{L\delta^2} )
 -\frac{1}{\Delta_{S}\delta }(\nu_{S}+\frac{M\nu_{S,I}^2 \nu_{R}}{L \delta^4\Delta_{S} \Delta})\eta_{T}\nu_{S}
]+\BO_{z}(\frac{1}{N^2}).
\end{aligned}
\end{equation}
Now we evaluate $X_4$ as
\begin{subequations}
\label{exp_x4}
\begin{align}
X_{4}
&=\frac{M}{L\Delta_{S}}\kappa(\FR^{\frac{3}{2}}\BG_{R,\alpha},\FS^{\frac{1}{2}}\BF_{S,\alpha},\FT^{\frac{3}{2}}\BG_{T})\Upsilon(\LR,\FS^{\frac{3}{2}}\BF_{S,\alpha},\RT )\nonumber
\\
&
=[\frac{\nu_{R}\nu_{S}\overline{\omega}}{\delta^2\Delta}-\frac{\overline{\omega}\nu_{R}\eta_{S}\nu_{T,I}}{\delta\Delta_{S}\Delta}-\frac{\nu_{T}\nu_{S}\nu_{S,I}\nu_{R}}{\delta^3\Delta_{S}\Delta} ]\frac{M}{L}\frac{\nu_{R}\nu_{S,I}\nu_{T}}{\delta^2\Delta_{S}}+\BO_{z}(\frac{1}{N^2})\label{X_4_step0}.
\\
&
=[\frac{\nu_{R}\nu_{S}\overline{\omega}}{\delta^2\Delta}+\frac{\overline{\omega}\nu_{R}\eta_{S}\nu_{S,I}\nu_{T}}{\delta^3\Delta_{S}\Delta}-\frac{\overline{\omega}^2\nu_{R}\eta_{S}}{\delta^2\Delta}-\frac{\nu_{T}\nu_{S}\nu_{S,I}\nu_{R}}{\delta^3\Delta_{S}\Delta} ]\frac{M}{L}\frac{\nu_{R}\nu_{S,I}\nu_{T}}{\delta^2\Delta_{S}}+\BO_{z}(\frac{1}{N^2})\label{X_4_step1}
\\
&
=
[-\frac{\eta_{S,I}\nu_{S,I}\nu_{T}\nu_{R} }{\delta^4\Delta\Delta_{S}}
+\frac{\overline{\omega}\nu_{R}\nu_{S}}{\delta^2\Delta}-\frac{ \overline{\omega}^2\nu_{R}\eta_{S}}{\delta^2\Delta}]\frac{M}{L}\frac{\nu_{R}\nu_{S,I}\nu_{T}}{\delta^2\Delta_{S}}+\BO_{z}(\frac{1}{N^2})\label{X_4_step2}
\\
&
=[ \frac{-\eta_{S,I}(\omega-\overline{\omega}\nu_{S})\nu_{T} }{\delta^3\Delta_{S}\Delta} 
+\frac{\overline{\omega}\nu_{S}}{\delta^2\Delta}-\frac{ \overline{\omega}^2\eta_{S}}{\delta^2\Delta}
]\nu_{R}\frac{M}{L}\frac{\nu_{R}\nu_{S,I}\nu_{T}}{\delta^2\Delta_{S}}+\BO_{z}(\frac{1}{N^2})\label{X_4_step3}
\\
&
=[ \frac{\eta_{S,I}\nu_{T,I} }{\delta^3\Delta_{S}\Delta} 
-\frac{\eta_{S,I}\overline{\omega}}{\delta^3\Delta}
+\frac{\overline{\omega}\nu_{S}}{\delta^2\Delta}-\frac{ \overline{\omega}^2\eta_{S}}{\delta^2\Delta}
]\nu_{R}\frac{M}{L}\frac{\nu_{R}\nu_{S,I}\nu_{T}}{\delta^2\Delta_{S}}+\BO_{z}(\frac{1}{N^2})\label{X_4_step4}
\\
&
= \frac{M}{L}\frac{\eta_{S,I}\nu_{T,I} }{\delta^3\Delta_{S}\Delta} 
\frac{\nu_{R}^2\nu_{S,I}\nu_{T}}{\delta^2\Delta_{S}}+\BO_{z}(\frac{1}{N^2})\label{X_4_step5}.
\end{align}
\end{subequations}
(\ref{X_4_step0}) follows from~(\ref{ZEAB}) and~(\ref{UPSABC}).~(\ref{X_4_step1}) is obtained by applying~(\ref{nu_T_eqn}) to the second term of~(\ref{X_4_step0}).~(\ref{X_4_step2}) is obtained by combining the second and the fourth terms of~(\ref{X_4_step1}) by using~(\ref{eta_S_eqn}).~(\ref{X_4_step3}) results from~(\ref{nu_S_eqn}).~(\ref{X_4_step4}) can be obtained by applying $(\omega-\overline{\omega}\nu_{S})\nu_{T}=\omega\nu_{T}-\overline{\omega}+\overline{\omega}\Delta_{S}=-\nu_{T,I}+\overline{\omega}\Delta_{S}$ to the first term of~(\ref{X_4_step3}).~(\ref{X_4_step5}) is the result of combining the second to fourth terms of~(\ref{X_4_step4}) by using~(\ref{eta_S_eqn}).

By substituting~(\ref{TRRQP}) into~(\ref{TRQRQ}), we can obtain
\begin{equation}
\begin{aligned}
\label{Q_P1}
\E \underline{\Tr\BQ}\Phi(u,z)
&
=[\frac{K_{I}}{\Delta}(X_{1}+X_{2}+X_{3}+X_{4}) +\frac{1-K_{R}}{\Delta}(W_{1}+W_{2}+W_{3}+W_{4})]\jmath u \E\Phi(u,z) +\BO_{z}(\frac{1}{N}).
\end{aligned}
\end{equation}

It is easy to observe that the only difference between $X_{1}$ and $W_{1}$ is the order of $\FR$, which is $0.5$ in $W_1$ and $1.5$ in $X_1$. Therefore, $W_{1}$ is given by
\begin{equation}
\label{exp_w1}
\begin{aligned}
W_{1}&=-\frac{M\overline{\omega}^2\nu_{S}\eta_{R,I}}{L\delta^2\Delta}+\frac{M\overline{\omega}^2\nu_{R}\nu_{R,I} \nu_{S} }{L\delta^3\Delta}
-\frac{M \nu_{R}\nu_{R,I}\nu_{S,I}\nu_{T,I}}{L\delta^4\Delta_{S}\Delta }
+\frac{M \eta_{S,I,I}\nu_{R}\nu_{R,I}\nu_{T,I}}{L\delta^5\Delta_{S}\Delta }+\BO_{z}({\frac{1}{N^2}})
\\
&=W_{1,1}+W_{1,2}+W_{1,3}+W_{1,4}+\BO_{z}({\frac{1}{N^2}}).
\end{aligned}
\end{equation}
Similarly, we can obtain $W_{2}$ as
\begin{equation}
\label{exp_w2}
\begin{aligned}
W_{2}&=-\nu_{T} [\frac{M\nu_{S,I}^2\eta_{R,I}}{L\delta^4
 \Delta_{S} \Delta}+\frac{\delta}{\delta^3 \Delta \Delta_{S}}(1-\frac{M \omega\overline{\omega}}{L\delta^2}\nu_{R})\nu_{R,I}\eta_{S,I}]
  +\frac{1}{\Delta_{S}\delta^2 }(\nu_{S}+\frac{M\nu_{S,I}^2 \nu_{R}}{L\Delta\Delta_{S} \delta^4 })\nu_{R,I}\eta_{T}\nu_{S,I}+\BO_{z}({\frac{1}{N^2}})
  \\
  &=W_{2,1}+W_{2,2}+W_{2,3}+W_{2,4}+\BO({\frac{1}{N^2}}).
\end{aligned}
\end{equation}
The differences between $W_3$, $W_4$ and $X_{3}$, $X_{4}$ come from $\kappa$. Therefore, $W_{3}$ and $W_{4}$ can be written as
\begin{equation}
\label{exp_w3}
\begin{aligned}
W_{3}
&=-\frac{\nu_{R,I}\nu_{S,I}\nu_{T} }{\delta\Delta_{S}}[\frac{M \nu_{R}\nu_{S,I}\nu_{S} \nu_{T} }{L\delta^4\Delta_{S}\Delta}+\frac{\eta_{S}\nu_{T}}{\delta\Delta_{S}\Delta}(1-\frac{M\omega\overline{\omega}\nu_{R}}{L\delta^2} )
 -\frac{1}{\Delta_{S}\delta }(\nu_{S}+\frac{M\nu_{S,I}^2 \nu_{R}}{ L\delta^4\Delta\Delta_{S} })\eta_{T}\nu_{S}
]
\\
&
\overset{(a)}{=}\frac{M\nu_{R,I}\nu_{R}\nu_{S,I}^2\nu_{T}}{L\delta^{5}\Delta_{S}\Delta}-\frac{M\nu_{R,I}\nu_{R}\nu_{S,I}^2\nu_{T}}{L\delta^{5}\Delta_{S}^2\Delta}
-\frac{\eta_{S}\nu_{R,I}\nu_{S,I}\nu_{T}^{2}}{\delta^2\Delta_{S}^{2}\Delta}(1-\frac{M\omega\overline{\omega}\nu_{R}}{L\delta^2} )
+ \frac{\eta_{T}\nu_{R,I}\nu_{S}\nu_{S,I}\nu_{T}}{\Delta_{S}^{2}\delta^2 }(\nu_{S}+\frac{M\nu_{S,I}^2 \nu_{R}}{L\delta^4\Delta_{S}\Delta  })+\BO_{z}({\frac{1}{N^2}})
\\
&=W_{3,1}+W_{3,2}+W_{3,3}+W_{3,4}+\BO_{z}({\frac{1}{N^2}}),
\end{aligned}
\end{equation}
\begin{equation}
\label{exp_w4}
\begin{aligned}
W_{4}= \frac{M}{L}\frac{\eta_{S,I}\nu_{T,I} }{\delta^3\Delta_{S}\Delta} 
\frac{\nu_{R}\nu_{R,I}\nu_{S,I}\nu_{T}}{\delta^2\Delta_{S}}+\BO_{z}({\frac{1}{N^2}}),
\end{aligned}
\end{equation}
where step $(a)$ in~(\ref{exp_w3}) can be obtained by applying $\nu_{S}\nu_{T}=1-\Delta_{S}$ to the first term in the first line of~(\ref{exp_w3}). Based on the expression of $W_i$ and $X_i$, $i=1,2,3,4$, we can observe that all the terms except the $\eta_{R}$ related terms in $K_{I}(\sum_{i=1}^{4}X_{4})$ will be cancelled by those in $-K_{R}(\sum_{i=1}^{4}W_{4})$ so that only the $\eta_{R}$ and $\eta_{R,I}$ related terms remain in $K_{I}(\sum_{i=1}^{4}X_{4})-K_{R}(\sum_{i=1}^{4}W_{4})$. The $\eta_{R}$ related terms, which only appear in $X_{1}+X_{2}$, can be written as
\begin{equation}
-\frac{M\eta_{R}}{L\Delta}(\frac{\nu_{S,I}^2\nu_{T}}{\delta^4\Delta_{S}}+\frac{\overline{\omega}^2\nu_{S}}{\delta^2})
=-\frac{M\eta_{R}}{L\Delta}(\frac{-\nu_{S,I}\nu_{T,I}}{\delta^3\Delta_{S}}+\frac{\overline{\omega}\nu_{S,I}}{\delta^3} +\frac{\overline{\omega}^2\nu_{S}}{\delta^2})
=-\frac{M\eta_{R}}{L\Delta}(\frac{\omega\overline{\omega}}{\delta^2}-\frac{\nu_{S,I}\nu_{T,I}}{\delta^3\Delta_{S}}).
\end{equation}
The $\eta_{R,I}$ related terms which only come from $W_{1}$ and $W_{2}$ are given as
\begin{equation}
-\frac{M\eta_{R,I}}{L\Delta}(\frac{M\omega\overline{\omega}}{L\delta^2}-\frac{M\nu_{S,I}\nu_{T,I}}{L\delta^3\Delta_{S}}).
\end{equation}
Therefore,~(\ref{Q_P1}) can be rewritten as
\begin{equation}
\begin{aligned}
\label{Q_P2}
 \E \underline{\Tr\BQ}\Phi(u,z)= \jmath u \Omega \E \Phi(u,z)+\BO(\frac{1}{N})
= \jmath u \Omega \Psi(u,z) +\BO_{z}(\frac{1}{N}),
\end{aligned}
\end{equation}
where 
\begin{equation}
\label{WXterms}
\Omega= \frac{K_{I}(X_{1,1}+X_{1,2}) }{\Delta}
+\frac{1}{\Delta}(W_{1}+W_{2}+W_{3}+W_{4}-W_{1,1}-W_{2,1})
+W_{1,1}+W_{2,1}.
\end{equation}
Thus, we can obtain
\begin{equation}
\begin{aligned}
\label{Q_P3}
\frac{\partial \Psi(u,z)}{\partial z}
& = \jmath u \E \underline{\Tr\BQ}\Phi(u,z)
=-u^2 \Omega \Psi(u,z) +\BO_{z}(\frac{1}{N}).
\end{aligned}
\end{equation}

\subsubsection{Step 3, Convergence of the characteristic function and identification of the asymptotic variance}
In this step, we prove the convergence of the characteristic function and show that $V(z)=-\log(\Delta)-\log(\Delta_{S})$ is the asymptotic variance. For that purpose, we first compute $\frac{\partial V(z)}{\partial z}$ by the chain rule. The derivatives of $\delta, \omega, \overline{\omega}$ can be obtained by the solution of the following system of equations,
\begin{equation}
\label{der_eq}
\begin{aligned}
\bold{P}
\begin{bmatrix}
\delta'
\\
\omega'
\\
\overline{\omega}'
\end{bmatrix}
=
\begin{bmatrix}
1-\frac{M\omega\overline{\omega} }{L\delta^2}\nu_{R} & \frac{M\overline{\omega}}{L\delta} \nu_{R}  & \frac{M\omega}{L\delta} \nu_{R} 
\\
- \frac{\nu_{S,I}}{\delta^2} & 1 & \nu_{S} 
\\
0 & \nu_{T} & 1
\end{bmatrix}
\begin{bmatrix}
\delta'
\\
\omega'
\\
\overline{\omega}'
\end{bmatrix}
=
\begin{bmatrix}
-\nu_{R,I}
\\
0
\\
0
\end{bmatrix},
\end{aligned}
\end{equation}
which can be derived by taking derivative on both sides of the fundamental equation~(\ref{basic_eq1}). Given $|\bold{P}|=\Delta\Delta_{S}$, $\delta', \omega', \overline{\omega}'$ can be determined as,
\begin{equation}
\begin{aligned}
\label{sys_ders}
~~&\delta'=\frac{-\nu_{R,I}}{\Delta},~~\omega'=\frac{-\nu_{R,I}\nu_{S,I}}{\delta^2\Delta\Delta_{S}},~~\overline{\omega}'=\frac{\nu_{R,I}\nu_{S,I}\nu_{T}}{\delta^2\Delta\Delta_{S}}.
\end{aligned}
\end{equation}
By the chain rule, we can obtain the following derivatives,
\begin{subequations}
\label{sec_sys_ders}
\begin{align}
&(\frac{M\omega\overline{\omega}}{L\delta})' =\frac{\delta'+\nu_{R,I}}{-\nu_{R}}
=\frac{M \omega\overline{\omega}\nu_{R,I}}{L\delta^2\Delta}-\frac{M \nu_{S,I}\nu_{T,I}\nu_{R,I} }{L\delta^3\Delta_{S} \Delta}=\frac{\theta \nu_{R,I}}{\Delta},\label{theta_der}
\\
&\nu_{R}'=-2\eta_{R}(\frac{M\omega\overline{\omega}}{L\delta})'-2\eta_{R,I}
=-\frac{2\eta_{R}\theta\nu_{R,I}}{\Delta} -2\eta_{R,I},\label{nu_R_der}
\\
&\nu_{S}'= -\frac{2\eta_{S,I}\nu_{R,I}}{\delta^2\Delta}+\frac{-2\eta_{S} \nu_{R,I}\nu_{S,I}\nu_{T}}{\Delta\Delta_{S}\delta^2},\label{nu_S_der}
\\
&\nu_{T}'=\frac{2\eta_{T}\nu_{R,I}\nu_{S,I}}{\Delta\Delta_{S}\delta^2},\label{nu_T_der}
\\
&\nu_{S,I}'= -\frac{2\eta_{S,I,I}\nu_{R,I}}{\delta^2\Delta}+\frac{-2\eta_{S,I} \nu_{R,I}\nu_{S,I}\nu_{T}}{\Delta\Delta_{S}\delta^2}\label{nu_SI_der}.
\end{align}
\end{subequations}
Therefore, $\frac{\mathrm{d}{V}(z)}{\mathrm{d} z}$ can be computed as
\begin{equation}
\label{der_Vz}
\begin{aligned}
\frac{\mathrm{d}{V}(z)}{\mathrm{d} z}&=-\frac{\Delta'\Delta_{S}+\Delta_{S}'\Delta}{\Delta\Delta_{S}}
=-\frac{1}{\Delta\Delta_{S}}[ -\Delta_{S}\theta\nu_{R}'+ \frac{M \nu_{R}\nu_{T,I}\nu_{S,I}'}{L\delta^3} 
-(1-\frac{M \omega\overline{\omega}\nu_{R}}{L\delta^2})\nu_{T}\nu_{S}'
+ \frac{M\nu_{R} \nu_{S,I} \nu_{T,I}'}{L\delta^3}
\\
&
-(1-\frac{M \omega\overline{\omega}\nu_{R}}{L\delta^2})\nu_{S}\nu_{T}'
+\frac{M\nu_{R} }{L}( -(\frac{\omega\overline{\omega}}{\delta^2})'\Delta_{S}-\frac{3\nu_{S,I}\nu_{T,I}\delta'}{\delta^{4}}  )
]
=Z_{1}+Z_{2}+Z_{3}+Z_{4}+Z_{5}+Z_{6}.
\end{aligned}
\end{equation}
To show~(\ref{de_conv}), we only need to compare $\sum_{i=1}^{6} Z_{i}$ and $\Omega$ in~(\ref{WXterms}). First, $Z_{1}$ can be further decomposed as
\begin{equation}
\begin{aligned}
Z_{1}=-\frac{1}{\Delta\Delta_{S}}[2\Delta_{S}\theta^2\eta_{R}\nu_{R,I}+2\Delta_{S}\theta\eta_{R,I}]
=Z_{1,1}+Z_{1,2}.
\end{aligned}
\end{equation}
By evaluation in~(\ref{exp_x1}) and~(\ref{exp_x2}), we have
\begin{equation}
\label{Z_11}
\begin{aligned}
&\frac{(X_{1,1}+X_{2,1})K_{I}}{\Delta}=
(-\frac{M\overline{\omega}^2\nu_{S}\eta_{R}}{L\delta^2\Delta^2}
- \frac{M\nu_{S,I}^2\nu_{T}\eta_{R}}{L\delta^4
 \Delta_{S} \Delta^2})\theta\nu_{R,I}
 +\BO_{z}(\frac{1}{N^2})
  \\
 &
 =(-\frac{M\eta_{R}\omega\overline{\omega}}{L\delta^2\Delta^2}
 +\frac{M\eta_{R}\nu_{S,I}\nu_{T,I}}{L\delta^3\Delta_{S}\Delta^2})\theta\nu_{R,I}
\overset{(a)}{=}-\frac{\theta^2\Delta_{S} \eta_{R}\nu_{R,I}}{\Delta_{S}\Delta}+\BO_{z}(\frac{1}{N^2})=0.5Z_{1,1}+\BO_{z}(\frac{1}{N^2}),
\end{aligned}
\end{equation}
where step $(a)$ in~(\ref{Z_11}) follows from~(\ref{DeltaS_eqn}). By~(\ref{exp_w1}) and~(\ref{exp_w2}), we have
\begin{equation}
\label{Z_12}
\begin{aligned}
W_{1,1}+W_{2,1}=-\frac{M\overline{\omega}^2\nu_{S}\eta_{R,I}}{L\delta^2\Delta} 
- \frac{M\nu_{S,I}^2\nu_{T}\eta_{R,I}}{L\delta^4
 \Delta_{S} \Delta}+\BO_{z}(\frac{1}{N^2})
=0.5Z_{1,2}+\BO_{z}(\frac{1}{N^2}).
\end{aligned}
\end{equation}
We can thus determine $Z_1$ from~(\ref{Z_11}) and~(\ref{Z_12}). Next, we will handle $Z_5+Z_4$. First, notice that
\begin{subequations}
\label{notice_eqn}
\begin{align}
\nu_{S}+\frac{M\nu_{S,I}^2 \nu_{R}}{L\Delta\Delta_{S} \delta^4 }
&
=
\frac{\nu_{S}(1-\frac{M\omega\overline{\omega}\nu_{R}}{L\delta^2}) }{\Delta}+\frac{1}{\Delta}(\frac{M\nu_{R}\nu_{S,I}\nu_{T,I}\nu_{S} }{L\delta^3\Delta_{S}}-\frac{M\overline{\omega}\nu_{S,I}\nu_{S}\nu_{R}}{L\delta^3\Delta_{S}}+\frac{M \omega\nu_{R}\nu_{S,I}}{L\delta^3\Delta_{S}} )\label{notice_step1}
\\
&
=\frac{\nu_{S}(1-\frac{M\omega\overline{\omega}\nu_{R}}{L\delta^2}) }{\Delta}+\frac{M\nu_{R}\nu_{S,I}\omega}{L\delta^3\Delta}.~\label{notice_step2}
\end{align}
\end{subequations}
(\ref{notice_step1}) can be obtained by applying the definition of $\Delta$ to $\nu_{S}$ and applying~(\ref{nu_S_eqn}) to one $\nu_{S,I}$ in $\frac{M\nu_{S,I}^2 \nu_{R}}{L\Delta\Delta_{S} \delta^4 }$.~(\ref{notice_step2}) follows from~(\ref{nu_T_eqn}) and the definition of $\Delta_{S}$. By~(\ref{notice_eqn}),~(\ref{exp_w2}), and~(\ref{exp_w3}), we have
\begin{subequations}
\begin{align}
&\frac{1}{\Delta}(W_{2,3}+W_{3,4}+W_{3,2})\nonumber
\\
&
=\frac{1}{\Delta_{S}\Delta}(\frac{\eta_{T}\nu_{R,I}\nu_{S,I}\nu_{S}(1-\frac{M\omega\overline{\omega}\nu_{R}}{L\delta^2})}{\delta^2\Delta_{S}\Delta} +\frac{M\eta_{T}\nu_{R}\nu_{R,I}\nu_{S,I}^2\omega}{L\delta^5\Delta_{S}\Delta}-\frac{M\nu_{R,I}\nu_{R}\nu_{S,I}^2\nu_{T}}{L\delta^{5}\Delta_{S}\Delta} )+\BO_{z}(\frac{1}{N^2})\label{W_gourp1_step0}
\\
&=\frac{1}{\Delta_{S}\Delta}(\frac{\eta_{T}\nu_{R,I}\nu_{S,I}\nu_{S}(1-\frac{M\omega\overline{\omega}\nu_{R}}{L\delta^2})}{\delta^2\Delta_{S}\Delta} -\frac{M\nu_{R,I}\nu_{R}\nu_{S,I}^2\eta_{T,I}}{L\delta^{5}\Delta_{S}\Delta} )+\BO_{z}(\frac{1}{N^2})\label{W_gourp1_step1}
\\
&
=\frac{1}{2\Delta_{S}\Delta}[(1-\frac{M\nu_{R}\omega\overline{\omega}}{L\delta^2})\nu_{S}\nu_{T}'-\frac{M\nu_{R}\nu_{S,I}\nu_{T,I}'}{L\delta^3}   ]+\BO_{z}(\frac{1}{N^2})\label{W_gourp1_step2}
\\
&
=0.5(Z_{5}+Z_{4})+\BO_{z}(\frac{1}{N^2}).\label{W_gourp1_step3}
\end{align}
\end{subequations}
(\ref{W_gourp1_step1}) is obtained by applying~(\ref{eta_T_eqn}) to combine the second and third terms of~(\ref{W_gourp1_step0}). (\ref{W_gourp1_step2}) follows from~(\ref{nu_T_der}).~(\ref{W_gourp1_step3}) can be obtained by comparing~(\ref{W_gourp1_step2}) with~(\ref{der_Vz}). Now we handle $Z_3$. By~(\ref{exp_w2}) and~(\ref{exp_w3}), we can obtain
\begin{equation}
\label{handle_Z3}
\begin{aligned}
&\frac{1}{\Delta}(W_{2,2}+W_{3,3})
\\
&=-\frac{\nu_{T}\delta}{\delta^3 \Delta^2 \Delta_{S}}(1-\frac{M \omega\overline{\omega}}{L\delta^2}\nu_{R})\nu_{R,I}\eta_{S,I}-\frac{\eta_{S}\nu_{R,I}\nu_{S,I}\nu_{T}^{2}}{\delta^2\Delta_{S}^{2}\Delta^2}(1-\frac{M\omega\overline{\omega}\nu_{R}}{L\delta^2} )+\BO_{z}(\frac{1}{N^2})
\\
&
=\frac{1}{2\Delta_{S}\Delta}(1-\frac{M\nu_{R}\omega\overline{\omega}}{L\delta^2})\nu_{T}(-\frac{2\eta_{S,I}\nu_{R,I}}{\delta^2\Delta}-\frac{2\eta_{S}\nu_{R,I}\nu_{S,I}\nu_{T}}{\Delta_{S}\Delta})+\BO_{z}(\frac{1}{N^2})
\\
&
\overset{(a)}{=}\frac{1}{2\Delta_{S}\Delta}(1-\frac{M\nu_{R}\omega\overline{\omega}}{L\delta^2})\nu_{T}\nu_{S}'+\BO_{z}(\frac{1}{N^2})
\\
&
=
0.5Z_{3}+\BO_{z}(\frac{1}{N^2}),
\end{aligned}
\end{equation}
where step $(a)$ in~(\ref{handle_Z3}) follows from~(\ref{nu_S_der}). By~(\ref{exp_w1}),~(\ref{exp_w4}), and~(\ref{nu_SI_der}), $Z_2$ can be handled by
\begin{equation}
\begin{aligned}
\frac{1}{\Delta}(W_{1,4}+W_{4})&=\frac{M \eta_{S,I,I}\nu_{R}\nu_{R,I}\nu_{T,I}}{L\delta^5\Delta_{S}\Delta^2 }
+ 
\frac{M\eta_{S,I}\nu_{T,I}\nu_{R}\nu_{R,I}\nu_{S,I}\nu_{T}}{L\delta^5\Delta_{S}^2\Delta^2}+\BO_{z}(\frac{1}{N^2})
\\
&= -\frac{M\nu_{R}\nu_{T,I}\nu_{S,I}'}{2L\Delta_{S}\Delta\delta^3}+\BO_{z}(\frac{1}{N^2})
=0.5Z_{2}+\BO_{z}(\frac{1}{N^2}).
\end{aligned}
\end{equation}
By the evaluations in~(\ref{exp_w1}) and~(\ref{exp_w3}), we can handle $Z_6$ by
\begin{subequations}
\begin{align}
&
\frac{1}{\Delta}(W_{1,2}+W_{1,3}+W_{3,1})=
\frac{1}{\Delta\Delta_{S}}(\frac{M\Delta_{S}\overline{\omega}^2\nu_{R}\nu_{R,I} \nu_{S} }{L\delta^3\Delta}
-\frac{M \nu_{R}\nu_{R,I}\nu_{S,I}\nu_{T,I}}{L\delta^4\Delta }
+\frac{M\nu_{R,I}\nu_{R}\nu_{S,I}^2\nu_{T}}{L\delta^{5}\Delta})+\BO_{z}(\frac{1}{N^2})\label{Z6_step1}
\\
&
=\frac{1}{\Delta\Delta_{S}}(\frac{M\Delta_{S}\omega\overline{\omega}\nu_{R}\nu_{R,I} }{L\delta^3\Delta}
-\frac{M\Delta_{S}\overline{\omega}\nu_{R}\nu_{R,I} \nu_{S,I} }{L\delta^4\Delta}
+\frac{M\Delta_{S}\overline{\omega}\nu_{R,I}\nu_{R}\nu_{S,I}}{L\delta^{4}\Delta}
-\frac{2M\nu_{R,I}\nu_{R}\nu_{S,I}\nu_{T,I}}{L\delta^{4}\Delta})+\BO_{z}(\frac{1}{N^2})\label{Z6_step2}
\\
&
=\frac{1}{2\Delta\Delta_{S}}( \frac{\Delta_{S}\theta\nu_{R}\nu_{R,I} }{\delta\Delta}+\frac{M\Delta_{S}\omega\overline{\omega}\nu_{R}}{L\delta}(\frac{1}{\delta})'+\frac{3M\nu_{R}\nu_{S,I}\nu_{T,I}\delta'}{L\delta^{4}}    )+\BO_{z}(\frac{1}{N^2})\label{Z6_step3} 
\\
&=\frac{1}{2\Delta\Delta_{S}} \frac{M\nu_{R} }{L}( (\frac{\omega\overline{\omega}}{\delta^2})'\Delta_{S}+\frac{3\nu_{S,I}\nu_{T,I}\delta'}{\delta^{4}}    )+\BO_{z}(\frac{1}{N^2})\label{Z6_step4}
\\
&
=0.5Z_{6}+\BO_{z}(\frac{1}{N^2}).
\end{align}
\end{subequations}
(\ref{Z6_step2}) results from applying~(\ref{nu_T_eqn}) and~(\ref{DeltaS_eqn}) to the first term and third term of~(\ref{Z6_step1}), respectively. (\ref{Z6_step3}) and~(\ref{Z6_step4}) follow from~(\ref{sys_ders}) and~(\ref{theta_der}), respectively. Therefore, by comparing~(\ref{WXterms}) and~(\ref{der_Vz}), we can obtain
\begin{equation}
\label{v_omega}
\frac{\partial V(z)}{\partial z}=2\Omega+\BO(\frac{\mathcal{P}(\frac{1}{z})}{N^2z^2}).
\end{equation}
By observing that $V(\infty)=0$, we have $V(z)=-\int_{z}^{\infty}\frac{\mathrm{d}V(x)}{\mathrm{d} x}\mathrm{d}x$. By far, we can conclude that $\Omega=\frac{-\partial \log(\Delta\Delta_{S})}{2\partial z}+\BO(\frac{\mathcal{P}(\frac{1}{z})}{N^2z^2})$.
It thus follows from~(\ref{Q_P3}) that 
\begin{equation}
\frac{\partial \Psi(u,z)}{\partial z} =\jmath u \E \underline{\Tr\BQ}\Phi(u,z)
=-\frac{\partial V(z)}{2\partial z} u^2  \Psi(u,z) +\BO(\frac{\mathcal{P}(\frac{1}{z})}{z^2 N}),
\end{equation}
which indicates that 
\begin{equation}
\frac{\partial \Psi(u,z) e^{\frac{V(z)u^2}{2}} }{\partial z} =\BO(\frac{\mathcal{P}(\frac{1}{z})e^{\frac{V(z)u^2}{2}}}{z^2 N}).
\end{equation}
After integration, we can obtain 
\begin{equation}
 \Psi(u,z)e^{\frac{V(z)u^2}{2}} =1- \int_{z}^{\infty} \BO(\frac{\mathcal{P}(\frac{1}{x})e^{\frac{V(x)u^2}{2}}}{x^2 N}) \mathrm{d}x.
\end{equation}
Therefore, we have 
\begin{equation}
\label{pde_sol_cha}
 \Psi(u,z) =e^{-\frac{V(z)u^2}{2}}- e^{-\frac{V(z)u^2}{2}}\int_{z}^{\infty} \BO(\frac{\mathcal{P}(\frac{1}{x})e^{\frac{V(x)u^2}{2}}}{x^2 N}) \mathrm{d}x.
\end{equation}
By the lower bound for $\Delta$ in~Eq. (\ref{effect_low_bnd}) included in the proof Proposition~\ref{bound_var}, we know that there exists a constant $F$ independent of $z$ such that $V(z)\le 2\log(1+\frac{F}{z})$, which indicates that $\int_{z}^{\infty} \BO(\frac{\mathcal{P}(\frac{1}{x})e^{\frac{V(x)u^2}{2}}}{x^2 N}) \mathrm{d}x\le  e^{u^2\log(1+\frac{F}{z})}\int_{z}^{\infty}  \BO(\frac{\mathcal{P}(\frac{1}{x})}{x^2N})\mathrm{d}x=\BO(\frac{(1+\frac{F}{z})^{u^2}\mathcal{P}(\frac{1}{z})}{Nz})$. By~(\ref{pde_sol_cha}) and for given $u$ and $z$, we have
\begin{equation}
\label{char_con}
 \Psi(u,z) = e^{-\frac{V(z)u^2}{2}}+\BO(\frac{(1+\frac{F}{z})^{u^2}\mathcal{P}(\frac{1}{z})}{Nz})=e^{-\frac{V(z)u^2}{2}}+\BO(\frac{1}{N})     \xlongrightarrow{N\rightarrow \infty} e^{-\frac{ V(z)u^2}{2}}.
\end{equation}
Let  $\Psi_{c}(u,z)=\E e^{\jmath u (C(\sigma^2)-\overline{C}(\sigma^2))}$. According to the inequality $|e^{\jmath x}-1|=|2\jmath \sin(\frac{x}{2})|\le |x|$, we can obtain that 
\begin{equation}
|\Psi_{c}(u,z)-\Psi(u,z)|=|\E  \Psi_{c}(u,z)(1- e^{\jmath u (\overline{C}(\sigma^2) - \E C(\sigma^2)) }  ) | \le  | u (\overline{C}(\sigma^2) - \E C(\sigma^2))  |=  \BO_{z}(\frac{u}{N}).
\end{equation}
Therefore, by~(\ref{char_con}), we have
\begin{equation}
\label{char_con_new}
 \Psi_{c}(u,z) \xlongrightarrow{N\rightarrow \infty}  \Psi(u,z) \xlongrightarrow{N\rightarrow \infty} e^{-\frac{ V(z)u^2}{2}}.
 \end{equation}
(\ref{char_con_new}) also indicates that $V(z)=-\log(\Delta)-\log(\Delta_{S})$ is the asymptotic variance. 

\subsubsection{Step 4: From the convergence of the characteristic function to CLT.} The proof is mainly motivated by~\cite[Proposition 6]{hachem2008new}. Here we introduce the subscript for  $\overline{C}_{N}(z)$ and $V_{N}(z)$ since they depend on $M$, $N$, and $L$. We first prove the tightness of $\frac{C_{N}(z)-\overline{C}_{N}(z)}{\sqrt{V_{N}(z)}}$. By the upper bound for $V_{N}(z)$ in Proposition~\ref{bound_var}, we can find a small $x>0$ such that,
\begin{equation}
\label{con_clt_inq1}
\frac{1}{x}\int_{-x}^{x}(1-e^{-\frac{u^2V_{N}(z)}{2}}) \mathrm{d} u \le \frac{1}{x}\int_{-x}^{x}(1-e^{-\frac{u^2 M_{z} }{2}})  \mathrm{d} u \le \varepsilon.
\end{equation}
Then, by the convergence of the characteristic function in~(\ref{char_con_new}) and the dominated convergence theorem~\cite{billingsley2008probability}, we can obtain that for a given $\varepsilon$, there holds true
\begin{equation}
\label{con_clt_inq2}
|    \frac{1}{x}\int_{-x}^{x} [1-\Psi_{c,N}(u,z)]\mathrm{d}u -  \frac{1}{x}\int_{-x}^{x} (1-e^{-\frac{u^2V_{N}(z)}{2}})\mathrm{d}u | \le \varepsilon,
\end{equation}
when $N$ is large enough. Inequalities~(\ref{con_clt_inq1}) and~(\ref{con_clt_inq2}) indicate that
\begin{equation}
\label{int_ineq}
 \frac{1}{x}\int_{-x}^{x} [1-\Psi_{c,N}(u,z)]\mathrm{d}u \le 2\varepsilon.
\end{equation}
According to~\cite[Eq. (26.22) in Theorem 26.3]{billingsley2008probability}, the following inequality holds true for a real random variable $A$ with characteristic function $\Psi_{A}(u)$,
\begin{equation}
\label{bil_inq}
\Prob(|A|\ge\frac{2}{a}) \le \frac{1}{a}\int_{-a}^{a}[1-\Psi_{A}(u)]\mathrm{d}u.
\end{equation}
By taking $A=C_{N}(\sigma^2)-\overline{C}_{N}(\sigma^2)$ in~(\ref{bil_inq}) and combining the result with~(\ref{int_ineq}), we have
\begin{equation}
\Prob(| C_{N}(z)-\overline{C}_{N}(z)| \ge\frac{2}{x}) \le 2\varepsilon, 
\end{equation}
for a large $N$, which indicates the tightness of $C_{N}(z)-\overline{C}_{N}(z)$. By Proposition~\ref{bound_var}, we have $V_{N}(z)\ge m_{z}>0$ to conclude that $\frac{C_{N}(z)-\overline{C}_{N}(z)}{\sqrt{V_{N}(z)}}$ is tight. Therefore, we can find a subsequence $p_{N}$ such that $\frac{C_{p_{N}}(z)-\overline{C}_{p_{N}}(z)}{\sqrt{V_{p_{N}}(z)}}$ converges. Since $V_{N}(z)$ belongs to a compact set, we can find a subsequence $q_{N}$ from $p_{N}$ such that $V_{q_{N}}(z) \rightarrow v>0 $. According to~(\ref{char_con_new}), there holds true that $  C_{q_N}(z)-\overline{C}_{q_N}(z) \xlongrightarrow[N\rightarrow \infty]{\mathcal{D}} \mathcal{N}(0,v)$ or equivalently $ \frac{ C_{q_N}(z)-\overline{C}_{q_N}(z) }{\sqrt{V_{q_N}(z)}} \xlongrightarrow[N\rightarrow \infty]{\mathcal{D}} \mathcal{N}(0,1)$. There must be $ \frac{ C_{p_N}(z)-\overline{C}_{p_N}(z) }{\sqrt{V_{p_N}(z)}} \xlongrightarrow[N\rightarrow \infty]{\mathcal{D}} \mathcal{N}(0,1)$ since the limit of $\frac{ C_{p_N}(z)-\overline{C}_{p_N}(z) }{\sqrt{V_{p_N}(z)}}$ is same as the one of the subsequence $\frac{ C_{q_N}(z)-\overline{C}_{q_N}(z) }{\sqrt{V_{q_N}(z)}}$. Therefore, for each subsequence $c_N$ such that $\frac{ C_{c_N}(z)-\overline{C}_{c_N}(z) }{\sqrt{V_{c_N}(z)}}$ converges, there must be
\begin{equation}
\frac{ C_{c_N}(z)-\overline{C}_{c_N}(z) }{\sqrt{V_{c_N}(z)}}  \xlongrightarrow[N\rightarrow \infty]{\mathcal{D}}\mathcal{N}(0,1).
\end{equation}
According to Corollary of~\cite[Theorem 25.10]{billingsley2008probability}, there holds true
\begin{equation}
\frac{ C_{N}(z)-\overline{C}_{N}(z) }{\sqrt{V_{N}(z)}} \xlongrightarrow[N\rightarrow \infty]{\mathcal{D}}\mathcal{N}(0,1).
\end{equation}
By far, we conclude the proof of the CLT.


\subsubsection{Step 5, Convergence rate of the variance}
To determine its convergence rate, we evaluate the variance by the following approach
\begin{equation}
\label{var_int}
\Var(C(\sigma^2))=\int_{\sigma^{2}}^{\infty}-\frac{\mathrm{d}\cov(I(z),I(z)) }{\mathrm{d}z}\mathrm{d}z
=
-2\int_{\sigma^{2}}^{\infty}\E\Tr\BQ\underline{I(z)}\mathrm{d}z.
\end{equation}
The term in the integral can be handled by the same approach used in {step 2)} and {step 3)} to tackle $\E\Tr\BQ\Phi(u,z)$. Specifically, the same approximation of the variance and the convergence rate $\BO(\frac{1}{N})$ can be obtained with the following approximation. By the resolvent identity~(\ref{res_ide}), we have
\begin{equation}
\E\Tr\BQ\underline{I(z)}=-\frac{1}{z}\E\Tr\BQ\BH\BH^{H}\underline{I(z)}+\frac{N}{z}\E\underline{I}(z)=-\frac{1}{z}\E\Tr\BQ\BH\BH^{H}\underline{I(z)}.
\end{equation}
Then, we can utilize the integration by parts formula to evaluate $\E\Tr\BQ\BH\BH^{H}\underline{I(z)}$. Given $\frac{\partial \Phi(u,z)}{\partial z}=\frac{\partial \underline{I}(z)}{\partial z} \Phi(u,z)$, the evaluation of $\E\Tr\BQ\BH\BH^{H}\underline{I(z)}$ is similar to that for the characteristic function. Following the lines from~(\ref{QHHP0}) to~(\ref{for_var_con}), the dominating term in the approximation error can be given by  
\begin{equation}
\begin{aligned}
&\frac{1}{M} |\cov( {\Tr\BZ\BZ^{H}\BQ}, \gamma \underline{I(z)} )
  -\E \underline{\Tr\BZ\BZ^{H}\BQ}[\underline{I(z)}]\E \gamma|
=\frac{1}{M} |\E \underline{\Tr\BZ\BZ^{H}\BQ}[\underline{I(z)}]\underline{\gamma} |
\\
&
 \le\frac{1}{M}  \E^{\frac{1}{2}} |\underline{\Tr\BZ\BZ^{H}\BQ}[\underline{\gamma}]|^2  \E^{\frac{1}{2}}|\underline{I(z)}|^2\overset{(a)}{=}\BO(\frac{1}{N}),
\end{aligned}
\end{equation}
where $\gamma=\Tr\BG_{R,\alpha}\BZ\BF_{S,\alpha}\BY\FT^{\frac{3}{2}}\BG_{T,\alpha}\BY\BZ^{H}\BQ$ and step $(a)$ follows from the following analysis. By Cauchy-Schwarz inequality, we have
\begin{equation}
\begin{aligned}
\label{zzqgama}
\E^{\frac{1}{2}} |\underline{\Tr\BZ\BZ^{H}\BQ}[\underline{\gamma}]|^2 
\le \E^{\frac{1}{4}}|\underline{\Tr\BZ\BZ^{H}\BQ}|^{4}\E^{\frac{1}{4}}|\underline{\gamma}|^{4}.
\end{aligned}
\end{equation}
Given $\Var(x)=\E |x|^2-|\E x|^2$, we have
\begin{equation}
\begin{aligned}
\E |\underline{\Tr\BZ\BZ^{H}\BQ}|^{4}=\Var(|\underline{\Tr\BZ\BZ^{H}\BQ}|^2 )+\Var^2(\Tr\BZ\BZ^{H}\BQ).
\end{aligned}
\end{equation}
By the approach in the proof of Proposition~\ref{var_con}, we can obtain that
\begin{equation}
\begin{aligned}
&\Var(|\underline{\Tr\BZ\BZ^{H}\BQ}|^2 )=2\E(\underline{\Tr\BZ\BZ^{H}\BQ})[\frac{1}{L}(\sum_{i,j}|\frac{\Tr\BZ\BZ^{H}\BQ}{\partial X_{i,j}}  |^2+|\frac{\Tr\BZ\BZ^{H}\BQ}{\partial X_{i,j}^{*}}  |^2)
\\
&
+\frac{1}{M}(\sum_{k,l}|\frac{\Tr\BZ\BZ^{H}\BQ}{\partial Y_{k,l}}|^2+\sum_{k,l}|\frac{\Tr\BZ\BZ^{H}\BQ}{\partial Y_{k,l}^{*}}  |^2)]
\le \frac{\mathcal{P}(\frac{1}{z})}{z^4}.
\end{aligned}
\end{equation}
Therefore, $\E |\underline{\Tr\BZ\BZ^{H}\BQ}|^{4}= \BO(\frac{\mathcal{P}(\frac{1}{z})}{z^4})$. Similarly, we can prove
$\E|\underline{\gamma}|^4= \BO(\frac{\mathcal{P}(\frac{1}{z})}{z^4})$ to show that the RHS of~(\ref{zzqgama}) is $\BO(\frac{\mathcal{P}(\frac{1}{z})}{z^2})$. Since it has been shown that $\Var(C(z))=\BO(1)$, the dominating term in the error is $\BO(\frac{\mathcal{P}(\frac{1}{z})}{Nz^{2}})$. It thus follows from~(\ref{var_int})
\begin{equation}
\begin{aligned}
\Var(C(\sigma^2))=\int_{\sigma^{2}}^{\infty}-\frac{\mathrm{d} V(z)}{\mathrm{d} z}+\varepsilon_{V}(z) \mathrm{d}z
=V(z)+\BO(\frac{1}{N}).
\end{aligned}
\end{equation}

\section{Simulation}
\label{sec_simu}
A widely used correlation model for the linear array with uniformly distributed angle spreads is considered in~\cite{ertel1998overview,yin2013coordinated}. The $(i,j)$-th entry of the correlation matrix for double-scattering MIMO channels can be given by~\cite{hoydis2011asymptotic,kammoun2019asymptotic}
\begin{equation}
[\bold{\Phi}(\mu,\theta,d,n)]_{i,j}=\frac{1}{n}\sum_{m=-\frac{(n-1)}{2}}^{\frac{(n-1)}{2}}\exp\left(-\jmath 2\pi (i-j)d\cos(\frac{\pi}{2}+\frac{m\theta}{n-1}+\mu) \right),
\end{equation}
where $\mu+\frac{m\theta}{n-1}$ denotes the uniform angle spread between the scatterers and the antennas. $d$ represents the antenna spacing and $n$ represents the number of antennas or scatterers. The correlation matrices for the receiver, transmitter, and scatterer are denoted as $\bold{\Phi}_{R}=\bold{\Phi}(\mu,\theta_{R},d_{R},N)$, $\bold{\Phi}_{T}=\bold{\Phi}(\mu,\theta_{T},d_{T},M)$ and $\bold{\Phi}_{S}=\bold{\Phi}(\mu,\theta_{S},d_{S},L)$, respectively, where the parameters are set as $\mu=\frac{\pi}{3}$, $\theta_{R}=\theta_{T}=\frac{\pi}{3}$, $\theta_{S}=\frac{\pi}{6}$,
$d_{R}=d_{T}=0.5$ m, $d_{S}=2$ m. 

\begin{figure*}[!htb] 
    \centering
  \subfloat[\label{qqplot17}]{
       \includegraphics[width=0.4\linewidth]{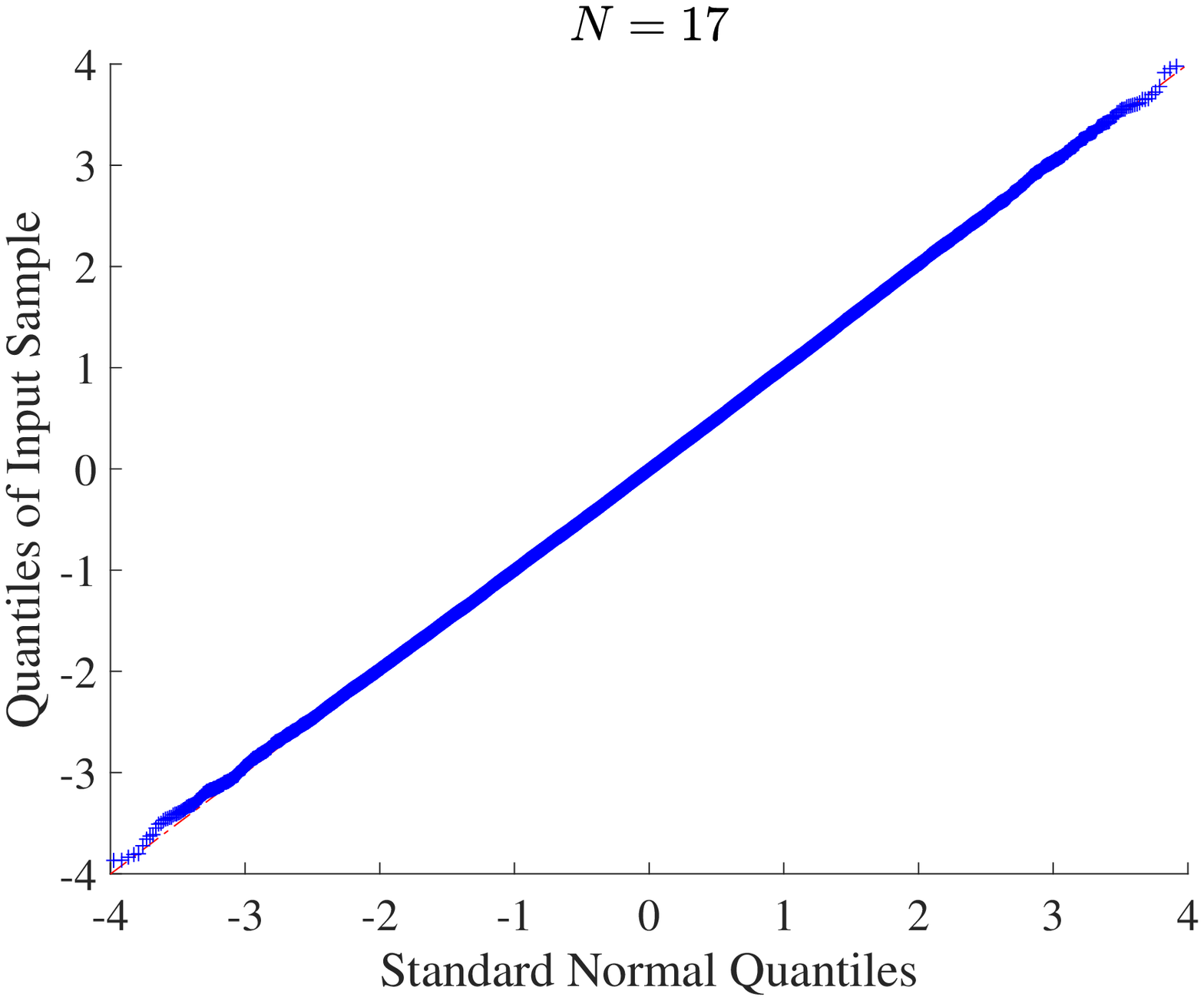}}
  \subfloat[\label{qqplot33}]{
        \includegraphics[width=0.4\linewidth]{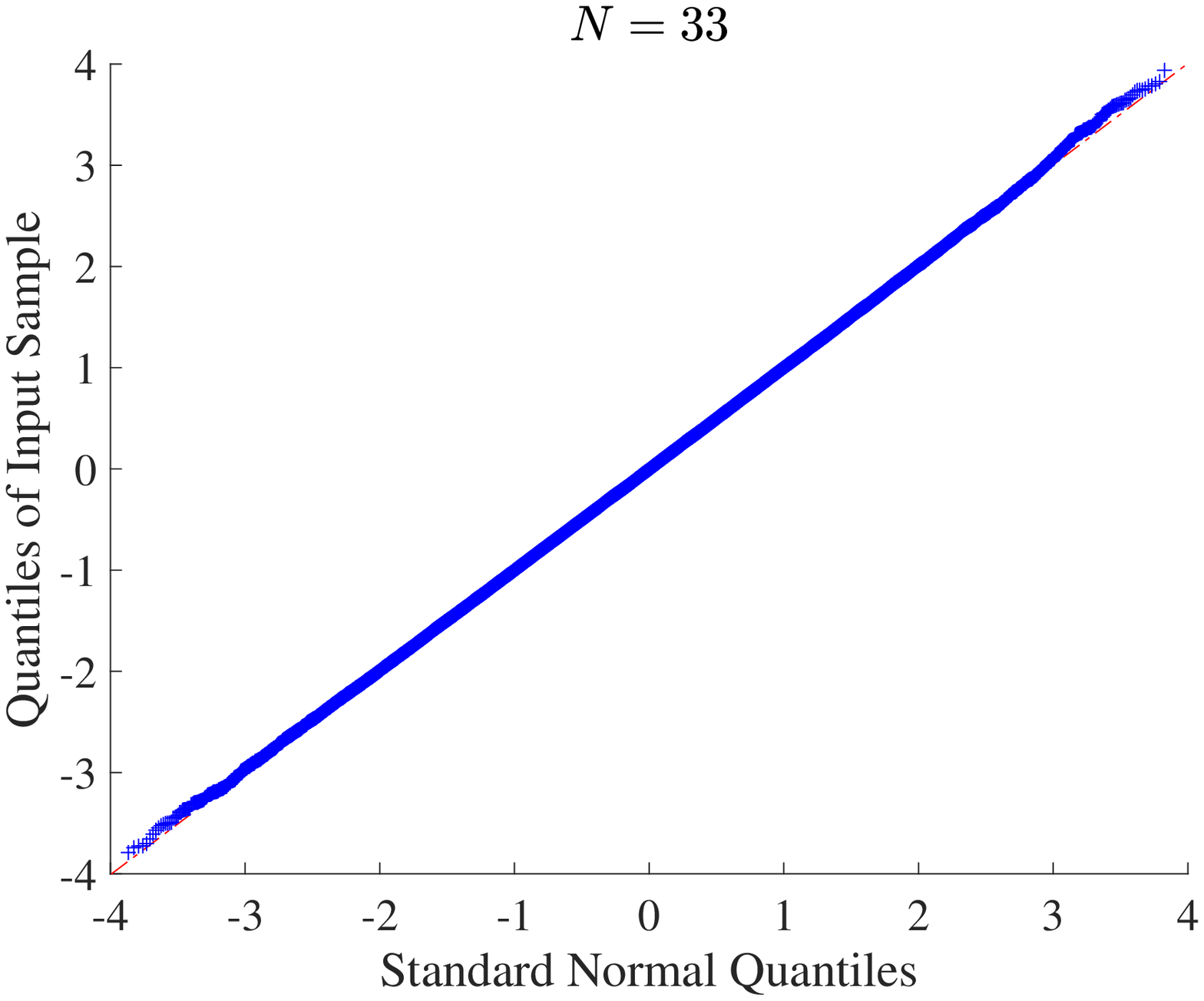}}
        \hfill
  \subfloat[\label{qqplot65}]{
        \includegraphics[width=0.4\linewidth]{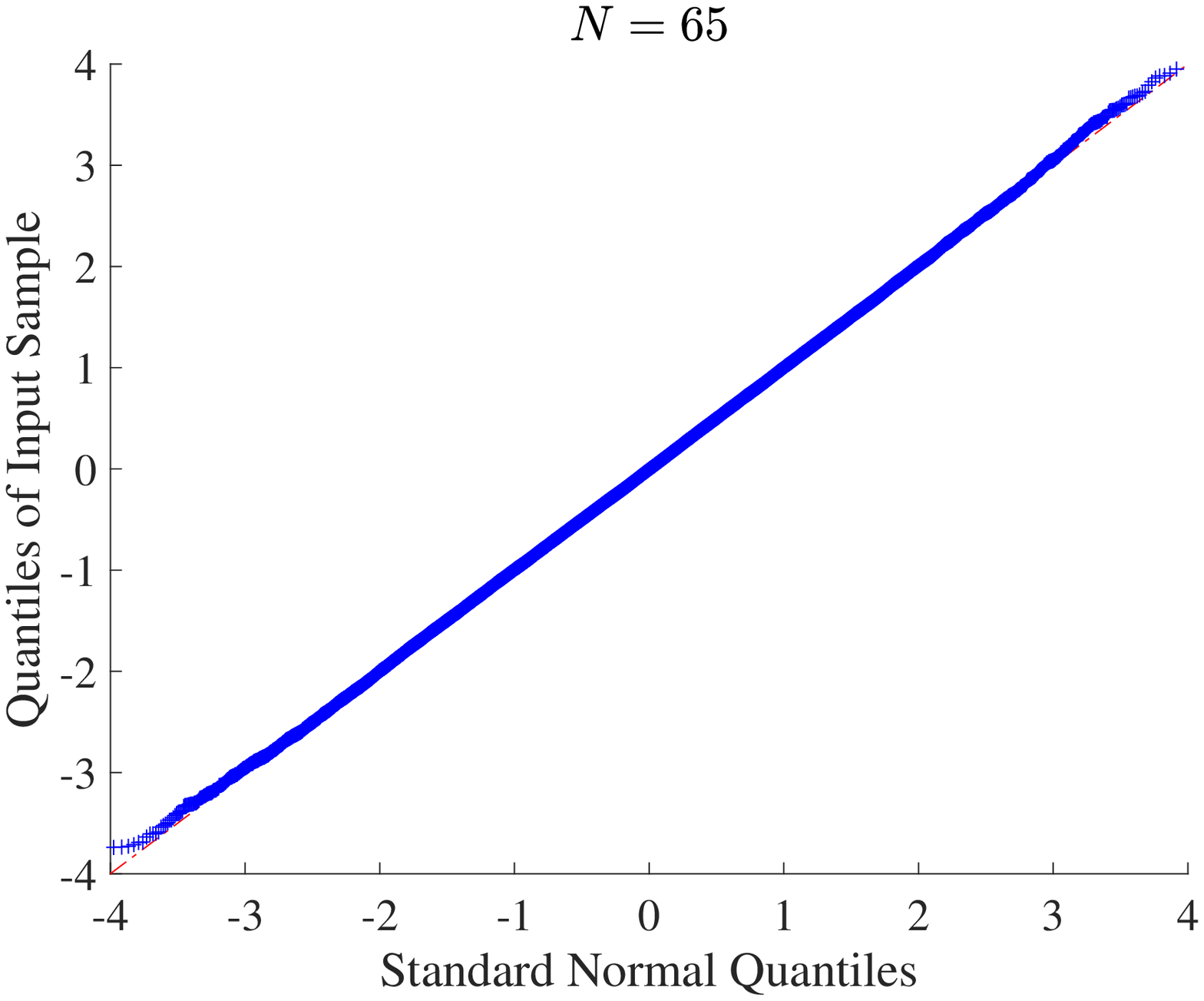}}
  \subfloat[\label{qqplot129}]{
        \includegraphics[width=0.4\linewidth]{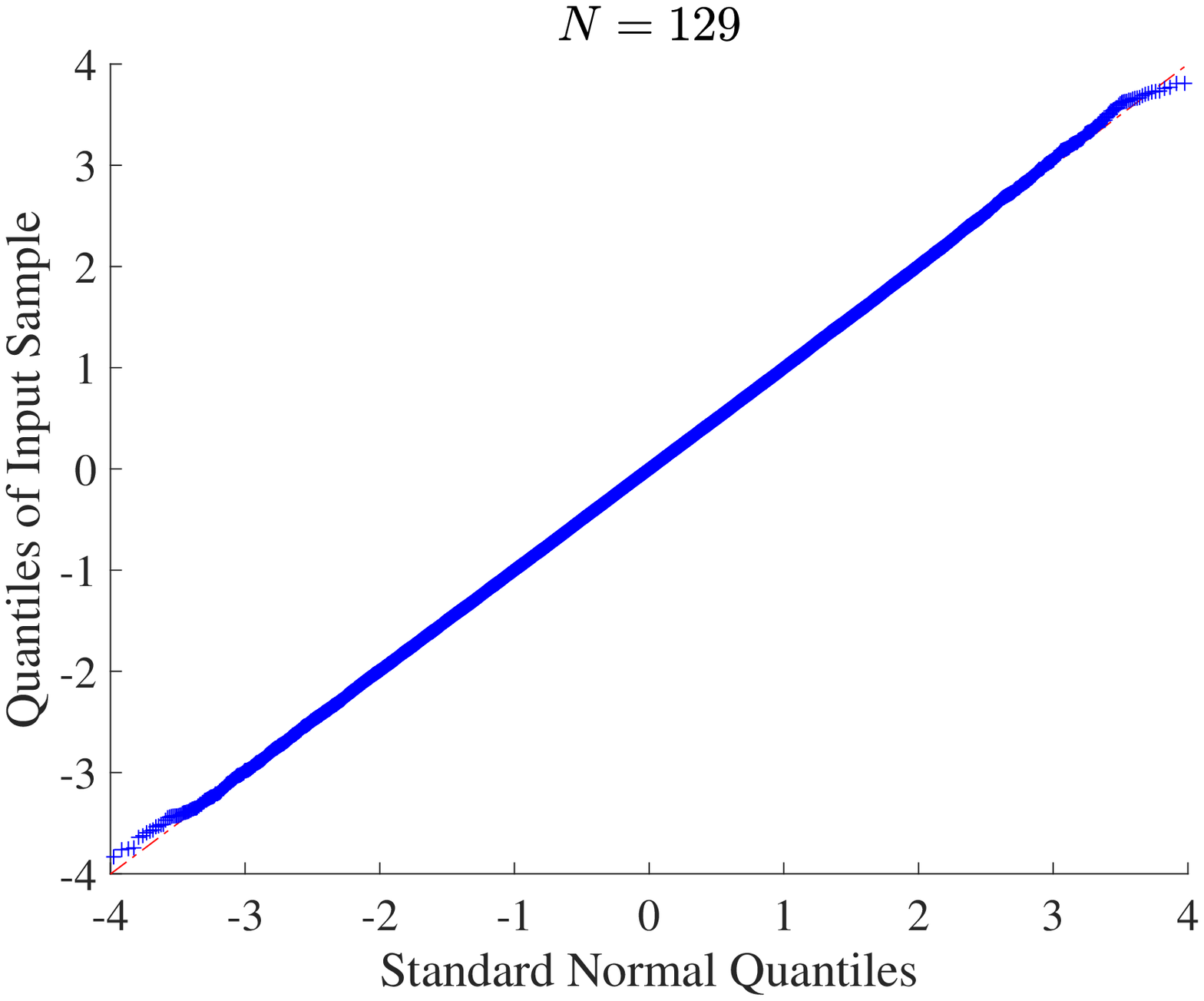}}    
  \caption{Fitness of CLT for the MI.}
 \label{qqplots}
\end{figure*}
\begin{figure}
\centering
       \includegraphics[width=0.6\linewidth]{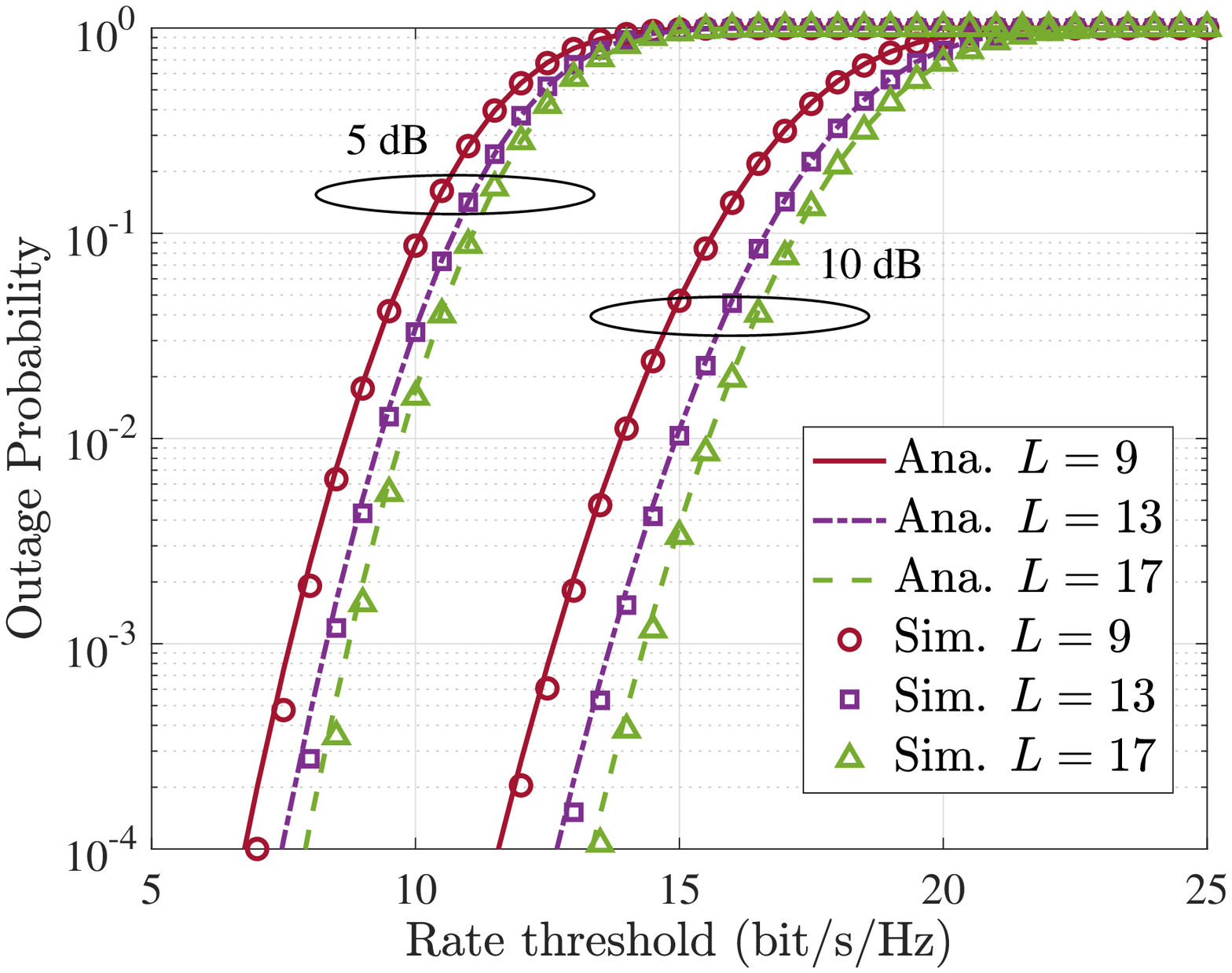}
         \caption{Outage probability approximation.}
         \label{outage_thre}
\end{figure}

In Fig.~\ref{qqplots}, the normal quantile-quantile-plots (QQ-plots) for the standardized $C(\sigma^2)$, i.e., $\frac{C(\sigma^2)-\overline{C}(\sigma^2)}{\sqrt{V(\sigma^2)}}$,  based on $10^{5}$ independent realizations are depicted with different dimensions $N=17, 33, 65, 129$. The results verified the Gaussianity for all the dimensions. To further validate the Gaussianity of the standardized $C(\sigma^2)$, Kolmogorov-Smirnov test (KS test)~\cite{massey1951kolmogorov} and Shapiro-Wilk test~\cite{shapiro1965analysis} are performed with the significant level $\alpha=0.05$. The p-values of the two tests are given in Table~\ref{test_results}. The number of samples for the two tests are $10^4$ and $5\times 10^{3}$, respectively. From Table~\ref{test_results}, we can observe that when $N>9$, the p-values are larger than $\alpha=0.05$, which indicates the Gaussianity of the standardized $C(\sigma^2)$.
 \begin{table*}[!htbp]
\centering
\caption{Tests of Gaussianity.}
\label{test_results}
\begin{tabular}{|c|c|c|c|c|c|c|c|}
\toprule
 $N$ & $9$ & $17$ & $33$ & $65$ & $129$ & $257$  \\
\midrule
 Kolmogorov-Smirnov test & $0.0030$ & $0.0759$ & $0.1009$ & $0.2204$ & $0.4448$ &  $0.6390$\\
 \toprule
Shapiro-Wilk test & $0.0203$ & $0.1789$ & $0.2249$ & $0.3717$ & $0.5798$ & $0.6531$
\\
\bottomrule
\end{tabular}
\end{table*}

In Fig.~\ref{outage_thre}, the outage probability for a given rate is depicted with $N=M=17$, $L=9, 13, 17$, and $\mathrm{SNR}=5,10$ dB. It can be observed that the approximation in~(\ref{appro_outage}) is accurate. The outage probability increases as the number of effective scatterers decreases, which indicates that less scatterers lead to worse reliability performance. 
\begin{figure*}[!htb] 
    \centering
  \subfloat[\label{appro_C1}]{
       \includegraphics[width=0.6\linewidth]{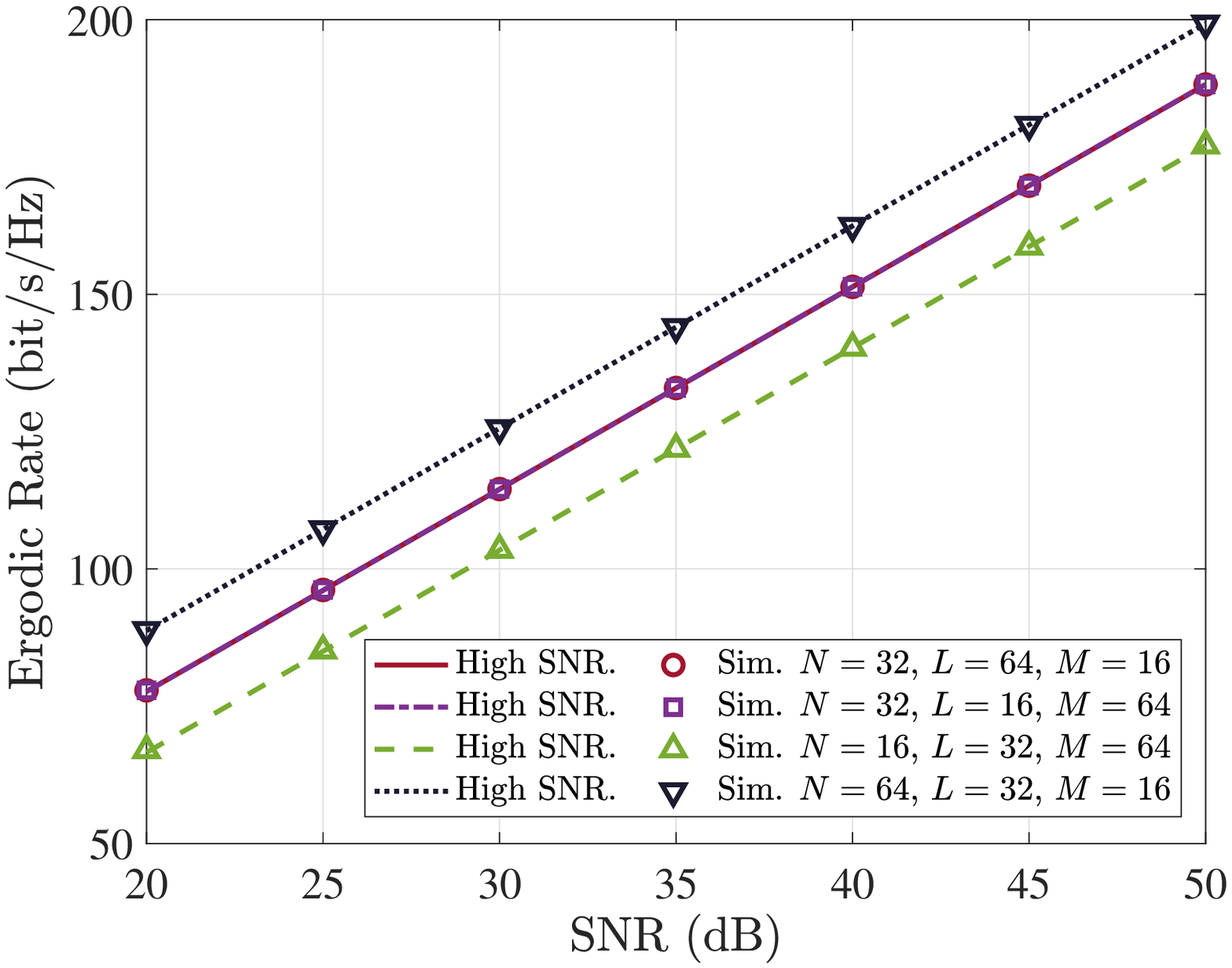}}
       \hfill
  \subfloat[\label{appro_C2}]{
        \includegraphics[width=0.6\linewidth]{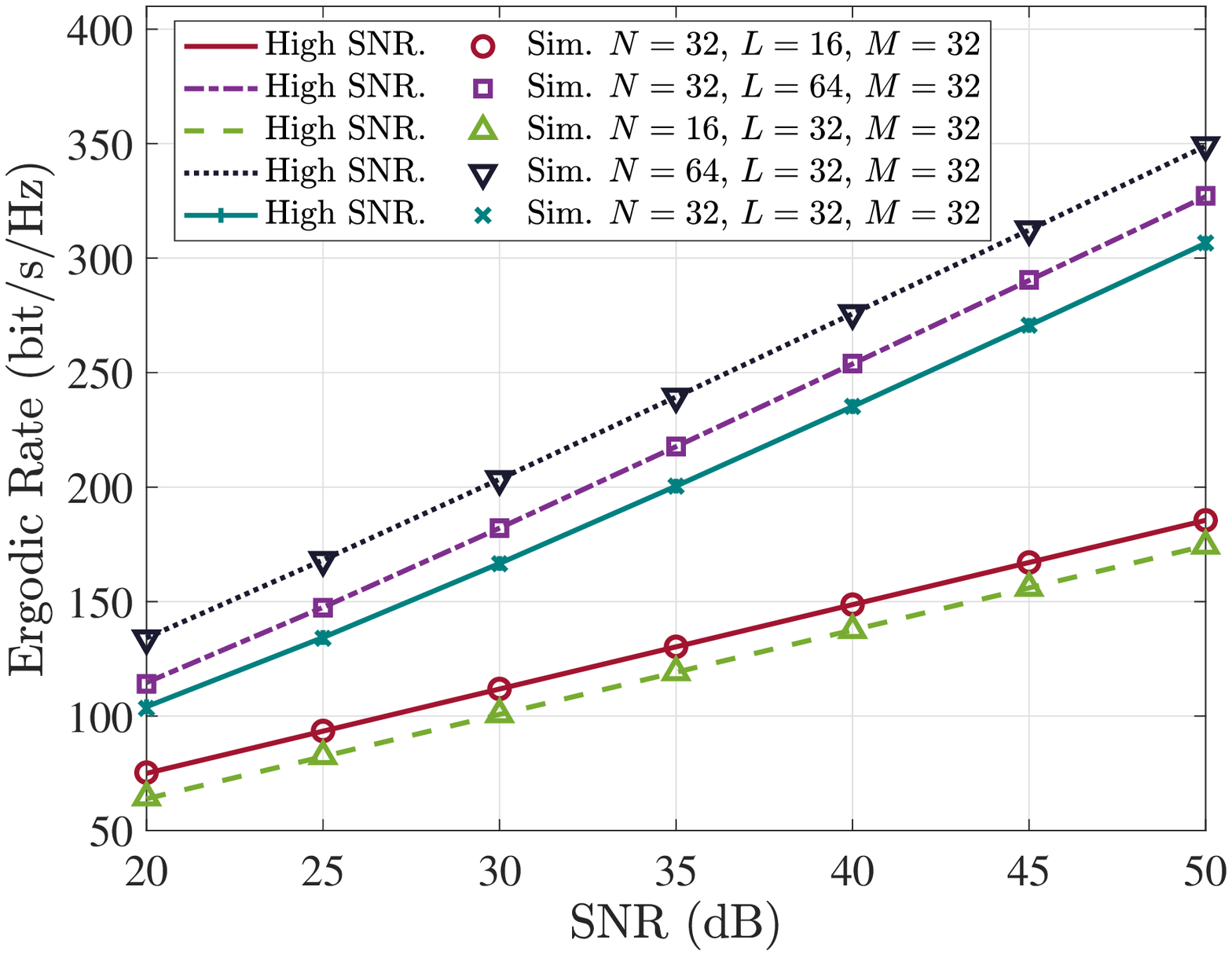}} 
         \caption{The moderate-to-high approximation for the EMI.}
 \label{appro_C}
\end{figure*}

\begin{figure*}[!htb] 
    \centering
  \subfloat[\label{appro_var1}]{
       \includegraphics[width=0.6\linewidth]{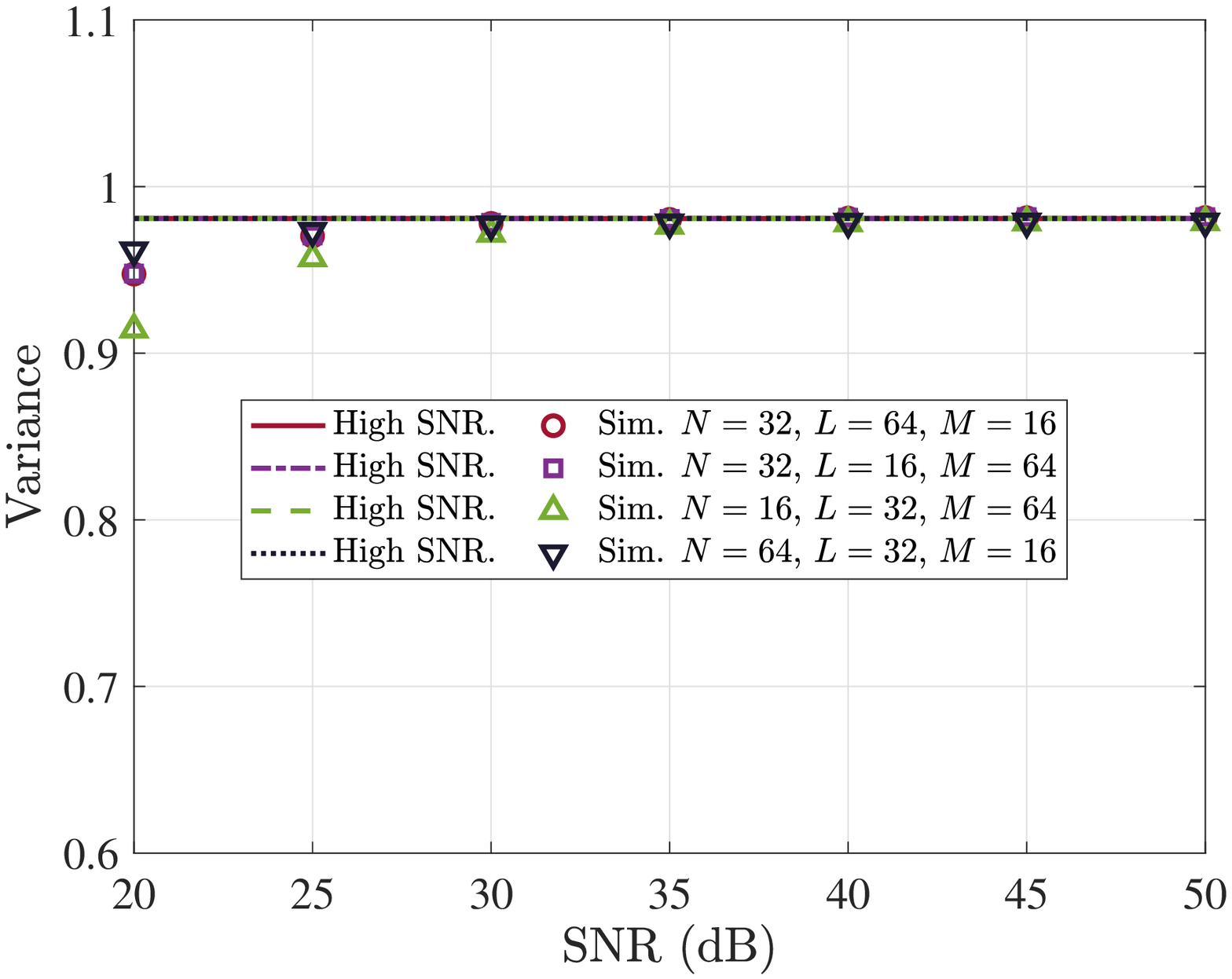}}
       \hfill
  \subfloat[\label{appro_var2}]{
        \includegraphics[width=0.6\linewidth]{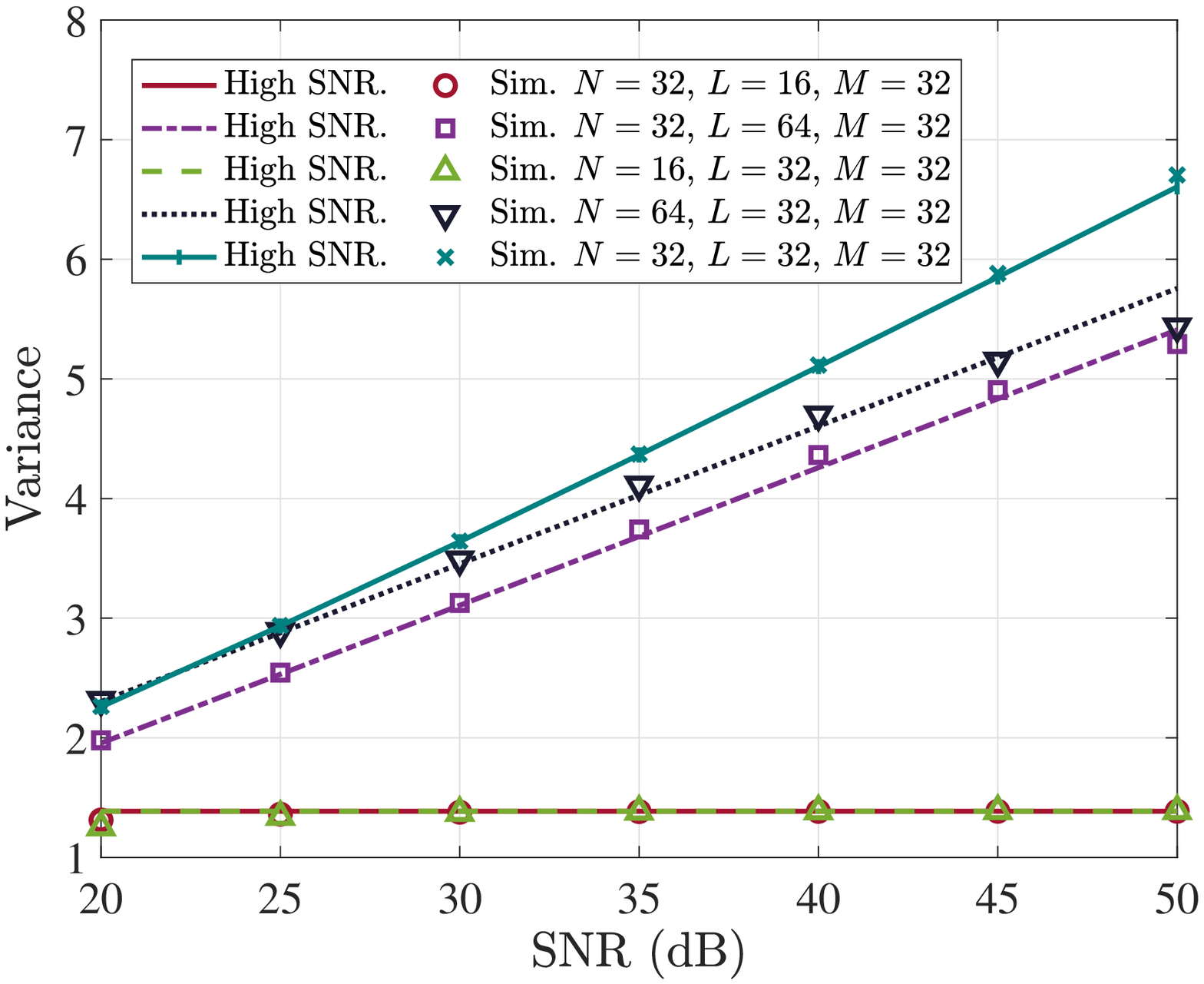}} 
         \caption{The moderate-to-high approximation for the asymptotic variance.}
 \label{appro_var}
\end{figure*}

In the following, we will use $(N,L,M)$ to distinguish the cases with different dimensions for brevity. For example, $(N,L,M)=(32,64,16)$ represents the case with $N=32$, $L=64$, and $M=16$. The moderate-to-high SNR range (typically 10 dB or higher~\cite{besson2000approximate}) considered here is $20$ to $50$ dB. Fig.~\ref{appro_C} illustrates the high-SNR approximation for the EMI shown in~(\ref{appros_Ciid}) versus SNR with various settings of $(N,L,M)$.  The Monte Carlo simulations are generated by $10^{6}$ realizations. The cases with unequal $N,L,M$ and equal $N,L,M$ are shown in Fig.~\ref{appro_C1} and Fig.~\ref{appro_C2}, respectively. Fig.~\ref{appro_C} validates the accuracy of the approximations in~(\ref{appros_Ciid}). It can be observed from Fig.~\ref{appro_C1} that the slopes, determined by $S_{M}\log(\rho)$ according to~(\ref{c_appro_case_1}), are same for the four cases. This indicates that the multiplexing gain is determined by the minimum of $N$, $L$, and $M$, which agrees with the analysis in Remark~\ref{remark_multiplex}. For the rank-deficient case when $L$ is the smallest, $L$ limits the multiplexing gain. In Fig.~\ref{appro_C1}, as predicted by~(\ref{appro_C1}), the approximations and simulation values for case 1 $(32,64,16)$ and case 2 $(32,16,64)$ are overlapped. In Fig.~\ref{appro_C2}, it can be observed that case 2 $(32,64,32)$ and case 4 $(64,32,32)$ have the same slope, which agrees with the result in~(\ref{c_appro_case_2}). Case 1 $(32,16,32)$ and case 3 $(16,32,32)$, corresponding to~(\ref{c_appro_case_1}), have a smaller slope compared with case 2 and case 4.~(\ref{c_appro_case_3}) is validated by case 5 $(32,32,32)$ in Fig.~\ref{appro_C2}. Furthermore, a larger $N$ results in a higher EMI, which validates the impact of the $\log(N)$ term in~(\ref{appros_Ciid}).

Fig.~\ref{appro_var} depicts the high-SNR approximations of the variance in~(\ref{appros_Viid}) versus SNR with the same settings as that for Fig.~\ref{appro_C}. It can be observed from Fig.~\ref{appro_var1} that when $M,N,L$ are unequal, the variance increases slowly with the SNR since the dominating term is $\BO(1)$ in~(\ref{v_appro_case_1}). Case 2 and case 4 in Fig.~\ref{appro_var2} validate~(\ref{v_appro_case_2}). In Fig.~\ref{appro_var2}, the variance of case 5 with $M=N=L$ increases with the highest speed, which agrees with~(\ref{v_appro_case_3}).




\section{Conclusion}
\label{sec_con}
In this paper, we evaluated the asymptotic distribution of the MI over double-scattering channels by large RMT when the number of antennas and the number of scatterers go to infinity with the same pace. By utilizing the Gaussian tools, we derived a closed-form deterministic approximation of the EMI and the variance of the MI with a guaranteed convergence rate $\BO(\frac{1}{N})$. By computing the characteristic function of the MI, we showed that the distribution of the MI converges to a Gaussian distribution with the same convergence rate. Besides the new results in terms of the convergence rates, moderate-to-high SNR approximation also revealed interesting physical insights for double-scattering and IRS-aided MIMO channels. Furthermore, the developed framework can be applied to more involved channel models, such as IRS-aided MIMO channels with line-of-sight (LoS) link or MIMO product channels with an arbitrary number of Gaussian matrices.

\appendices
\section{Proof of Proposition~\ref{var_con}}
\label{pro_var_con}
\begin{proof} Let $f(\BX,\BY)=\Tr\bold{A}\BX \bold{B}\BY\bold{C}\BY^{H}\bold{D}\BX^{H}\bold{E}\BQ$. By Nash-Poincar{\'e} Inequality~(\ref{nash_p}) and the derivative formula~(\ref{t_Q_der}), the variance of $f$ can be bounded by
\begin{align*}
&\Var(\Tr \bold{A}\BX \bold{B}\BY\bold{C}\BY^{H}\bold{D}\BX^{H}\bold{E}\BQ)
{=}\E \frac{1}{L}\sum_{i=1}^{N}\sum_{j=1}^{L}( |\frac{\partial f(\BX,\BY)}{\partial X_{i,j}} |^2  + |\frac{\partial f(\BX,\BY)}{\partial X_{i,j}^{*}} |^2 )
\\
&
+\frac{1}{M}\sum_{k=1}^{L}\sum_{h=1}^{M} (|\frac{\partial f(\BX,\BY)}{\partial Y_{k,h}} |^2  + |\frac{\partial f(\BX,\BY)}{\partial Y_{k,h}^{*}} |^2 )
\\
&
=\E \frac{1}{L}\sum_{i=1}^{N}\sum_{j=1}^{L}
|[\bold{B}\BY\bold{C}\BY^{H}\bold{D}\BX^{H}\bold{E}\BQ\BA]_{j,i}|^2
+\frac{1}{L}\sum_{i=1}^{N}\sum_{j=1}^{L}  |[\bold{E}\BQ\bold{A}\BX \bold{B}\BY\bold{C}\BY^{H}\bold{D}]_{i,j}|^2 \numberthis
\\
&
+\frac{1}{L}\sum_{i=1}^{N}\sum_{j=1}^{L} |\MS\BY\RT\BH^{H}\BQ\bold{A}\BX \bold{B}\BY\bold{C}\BY^{H}\bold{D}\BX^{H}\bold{E}\BQ\LR]_{j,i}  |^2
\\
&
+\frac{1}{L}\sum_{i=1}^{N}\sum_{j=1}^{L}  |[\LR\BQ\bold{A}\BX \bold{B}\BY\bold{C}\BY^{H}\bold{D}\BX^{H}\bold{E} \BQ\BH\RT\BY\MS ]_{i,j}   |^2
\\
&
+\frac{1}{M}\sum_{k=1}^{L}\sum_{h=1}^{M} |[\BC\BY^{H}\bold{D}\BX^{H}\bold{E}\BQ\BA\BX\BB]_{h,k}|^2
+\frac{1}{M}\sum_{k=1}^{L}\sum_{h=1}^{M} |[\bold{D}\BX^{H}\bold{E}\BQ\BA\BX\BB\BY\BC]_{k,h}|^2
\\
&
+\frac{1}{M}\sum_{k=1}^{L}\sum_{h=1}^{M} | [\RT\BH^{H}\BQ\bold{A}\BX \bold{B}\BY\bold{C}\BY^{H}\bold{D}\BX^{H}\bold{E}\BQ\BZ]_{h,k} |^2
+\frac{1}{M}\sum_{k=1}^{L}\sum_{h=1}^{M} | [\BZ^{H}\BQ\bold{A}\BX \bold{B}\BY\bold{C}\BY^{H}\bold{D}\BX^{H}\bold{E}\BQ\BH\RT]_{k,h} |^2
\\
&=W_{1}+W_{2}+W_{3}+W_{4}+W_{5}+W_{6}+W_{7}+W_{8}.
\end{align*}
The first term $W_{1}$ has the following bound,
\begin{equation}
\label{eva_W1}
\begin{aligned}
&
W_{1}=\frac{1}{L}\E \Tr \BB^{H}\BB\BY\bold{C}\BY^{H}\bold{D}\BX^{H}\bold{E}\BQ\BA\BA^{H}\BQ\bold{E}^{H}\BX\bold{D}^{H}\BY\BC^{H}\BY^{H}
\\
&
\le \E \|  \BB^{2}\BY\bold{C}\BY^{H}\bold{D}\BX^{H}\bold{E}^2\BX\bold{D}\BY^{H}\bold{C}\BY  \| \Tr \BQ\bold{A}^2\BQ
\le \E \|  \BB^{2}\BY\bold{C}\BY^{H}\bold{D}\BX^{H}\bold{E}^2\BX\bold{D}\BY^{H}\bold{C}\BY  \|  \frac{N U^2}{L z^2}
\\
&
\le
\frac{N U^{6}}{L z^2}\E \| \BX\bold{E}^2\BX^{H}  \| \E\|\BY\BC\BY^{H} \|^2
\overset{(b)}{\le} \frac{N U^{6}K}{L z^2},
\end{aligned}
\end{equation}
where $(b)$ follows from~(\ref{X_bound}) and~(\ref{Xp_bound}) and $K$ is a constant. The order of $z$ in $W_{i}, i=1,2,...,8$ coincides with the times of $\BQ$ occurring in the term. (\ref{eva_W1}) indicates that $W_{1}$ is a $\BO(\frac{1}{z^2})$ term. We can also obtain that $W_2$, $W_5$, and $W_6$ are $\BO(\frac{1}{z^2})$ terms. Similarly, $W_3$, $W_4$, $W_7$, and $W_8$ are $\BO(\frac{1}{z^4})$ terms. Therefore, we have
\begin{equation}
\Var(\Tr \bold{A}\BX \bold{B}\BY\bold{C}\BY^{H}\bold{D}\BX^{H}\bold{E}\BQ)=\BO(\frac{\mathcal{P}_2(\frac{1}{z})}{z^2}).
\end{equation}
which concludes the proof of~(\ref{V_ABCQ}).~(\ref{VABCQQ}) and~(\ref{VABCQZQ}) can be obtained by similar lines, which are omitted here.
\end{proof}

\section{The Boundness of the Intermediate Approximations}
\label{bound_lemma}
The boundness of the spectral norm for $\BG_{R,\alpha}$, $\BG_{T,\alpha}$,$\BG_{S,\alpha}$, and $\BF_{S,\alpha}$ can be guaranteed by the following lemma.


\begin{lemma} 
\label{bound_alp}
Given that ~\textbf{A.1} to~\textbf{A.3} hold true and $r_{max}$, $s_{max}$, and $t_{max}$ are the spectral norms of $\FR$, $\FS$, $\FT$, respectively, the following bounds hold true when $z>0$,
\begin{equation}
\label{par_bound}
\begin{aligned}
\frac{N\overline{r}}{L(z+ {r_{max}s_{max} t_{max}}  )}   \le \delta \le \frac{N r_{max}}{Lz},& ~~
 \frac{ \omega}{\delta}  \le \frac{L s_{max}}{M},~~
 \overline{\omega}  \le t_{max},
\\
\frac{N\overline{r}}{L(z+ r_{max}s_{max}t_{max})}
 \le \alpha_{\delta} \le  \frac{N r_{max}}{Lz},&~~
e_{\omega}\le \frac{Nr_{max}s_{max}}{zL},~~
\\
\frac{\overline{t}}{1+\frac{Nr_{max}s_{max}t_{max} }{zL}} \le \alpha_{\overline{\omega}} \le t_{max},&~~
\frac{L\overline{s}}{M(\frac{L(z+ \frac{r_{max}^2s_{max}t_{max}}{\overline{r}})}{N\overline{r}}+s_{max}t_{max})} 
 \le \alpha_{\omega}\le \frac{N r_{max} s_{max} }{Mz},
\end{aligned}
\end{equation}
where $\overline{r}=\frac{\Tr\FR}{N}, \overline{s}=\frac{\Tr\FS}{L}, \overline{t}=\frac{\Tr\FT}{M}$.
\end{lemma}
\begin{proof} By the upper bound following from the trace inequality~(\ref{trace_inq}), we have
\begin{equation}
\begin{aligned}
\alpha_{\delta}= \frac{\E\Tr\FR\BQ}{L}\le  \frac{N r_{max}}{Lz}.
\end{aligned}
\end{equation}
Similarly, we can obtain two upper bounds for $\frac{\omega}{\delta}$ and $\overline{\omega}$. The lower bound of $\delta$ can be derived by plugging the upper bounds for $\frac{\omega}{\delta}$ and $\overline{\omega}$ into the expression of $\delta$ in~(\ref{basic_eq1}). Now, we turn to evaluate the lower bound for $\alpha_{\delta}$, which follows from the inequalities below,
\begin{equation}
\begin{aligned}
 \frac{\E\Tr \FR\BQ}{L}=  \sum_{i=1}^{N}\E\frac{r_{i}}{L(z+\lambda_{\BH\BH^{H},i})}
 \overset{(a)}{\ge} \sum_{i=1}^{N}\frac{r_{i}}{L(z+\E \lambda_{\BH\BH^{H},i})} 
  \overset{(b)}{\ge}  \frac{\Tr\FR}{L(z+r_{max}\E\Tr\BH\BH^{H}/\Tr\FR )}
  \ge  \frac{N\overline{r}}{L(z+ {r_{max} s_{max}t_{max}})},
\end{aligned}
\end{equation}
where $(a)$ and $(b)$ follow from the convexity of the function $f(x)=\frac{1}{z+x}$ and Jesen's inequality. By the trace inequality~(\ref{trace_inq}), we have
\begin{equation}
\label{ew_bnd}
\begin{aligned}
e_{\omega}\le \E \frac{\Tr\BZ\BZ^{H} }{zM}=\frac{N r_{max}s_{max}}{zM}.
\end{aligned}
\end{equation}
The bounds for $\alpha_{\omega}$ and  $\alpha_{\overline{\omega}}$ can be obtained by plugging~(\ref{ew_bnd}) into their definitions.
\end{proof}
By Lemma~\ref{bound_alp}, we can show that the matrices $\BG_{R,\alpha}$, $\BG_{T,\alpha}$,$\BG_{S,\alpha}$, and $\BF_{S,\alpha}$ have bounded spectral norm.

\section{Proof of Lemma~\ref{Q_TO_AL}}
\label{pro_Q_TO_AL}
\begin{proof} In the following, we will show that $\E\Tr\BA\BQ$ can be approximated by $\Tr\BA\BG_{R,\alpha}$, which depends on $e_{\omega}$ and $\alpha_{\delta}$, with the approximation error $\BO(\frac{1}{N})$. By the proof, we can also obtain that $e_{\omega}=\frac{\Tr\FS\BG_{S,\alpha}}{M}+\BO(\frac{1}{N^2})$. To this end, we first evaluate $\E\Tr\BH\BH^{H}\BQ$ and then $\E\Tr\BQ$ by the resolvent identity~(\ref{res_ide}). By using the integration by parts formula~(\ref{int_part}) on $Y_{p,j}^{*}$, we have
\begin{equation}
\begin{aligned}
&\E [\BZ^{H} \BQ]_{p,k} [\BZ]_{i,q}[\BY\RT]_{q,j} [\BY^{*}\RT]_{p,j}
=t_{j}\E \frac{\partial [\BZ^{H} \BQ]_{p,k} [\BZ]_{i,q}[\BY]_{q,j}}{\partial Y_{p,j}}
\\
&
= \frac{1}{M} t_{j} \delta(p-q) \E[\BZ^{H} \BQ]_{p,k}[\BZ]_{i,q} 
-\frac{1}{M}t_{j} \E [\BZ^{H} \BQ\BZ]_{p,p}  [\BH^{H}\BQ]_{j,k} [\BZ]_{i,q}[\BY\RT]_{q,j}.
\end{aligned}
\end{equation}
By summing over subscript $p$, we can obtain
\begin{equation}
\label{HHQ_in}
\begin{aligned}
\E [\RT\BY^{H}\BZ^{H} \BQ]_{j,k}  [\BZ]_{i,q}[\BY\RT]_{q,j}
= \frac{1}{M} t_{j} \E[\BZ]_{i,q}[\BZ^{H} \BQ]_{q,k}
-\E\frac{t_{j}\Tr\BZ\BZ^{H}\BQ }{M}  [\BH^{H}\BQ]_{j,k}  [\BZ]_{i,q}[\BY\RT]_{q,j}.
\end{aligned}
\end{equation}
Adding the term $\frac{t_{j}\E\Tr\BZ\BZ^{H}\BQ}{M}\E [\BH^{H}\BQ]_{j,k} [\BZ]_{i,q}[\BY\RT]_{q,j}$ on both sides of~(\ref{HHQ_in}) and dividing both sides by $1+e_{\omega}t_{j}$, the following equation can be obtained by summing over $j$,
\begin{equation}
\begin{aligned}
\label{QYYZZ}
\E [\BY \FT\BY^{H}\BZ^{H} \BQ]_{q,k} [\BZ]_{i,q}  =\alpha_{\overline{\omega}}\E[\BZ^{H} \BQ]_{q,k}[\BZ]_{i,q} 
-\frac{1}{M} \cov( {\Tr\BZ\BZ^{H}\BQ},[\BY\FT^{\frac{3}{2}}(\bold{I}+e_{{\omega}}\FT)^{-1}\BY^{H}\BZ^{H}\BQ ]_{q,k}[\BZ]_{i,q} ).
\end{aligned}
\end{equation}
Similarly, with the integration by parts formula, we can obtain
\begin{equation}
\begin{aligned}
\label{ZZQ}
\E [\BZ^{H} \BQ]_{q,k}[\BZ]_{i,q}=\sum_{m}\E[\LR\BX^{*}\MS]_{m,q} [\BQ]_{m,k}[\BZ]_{i,q} =\frac{s_{q}}{L}\E[\FR\BQ]_{i,k}
-\E\frac{s_{q}\Tr\FR\BQ}{L} [\BY\RT\BH^{H}\BQ]_{q,k}[\BZ]_{i,q}.
\end{aligned}
\end{equation}
By substituting~(\ref{ZZQ}) into~(\ref{QYYZZ}) to replace $\E[\BZ^{H} \BQ]_{q,k}[\BZ]_{i,q}$, and dividing both sides of~(\ref{QYYZZ}) by $1+s_{q}\alpha_{\delta}\alpha_{\overline{\omega}}$, we have
\begin{equation}
\begin{aligned}
\label{HHQ}
&\E[\BY\RT\BH^{H}\BQ]_{q,k}[\BZ]_{i,q}=\frac{s_{q}\alpha_{\overline{\omega}}}{L(1+s_{q}\alpha_{\delta}\alpha_{\overline{\omega}})}\E[\FR\BQ]_{i,k}-\frac{\alpha_{\overline{\omega}}}{L} \cov( {\Tr\FR\BQ},[\FS\BF_{S,\alpha}\BY\RT\BH^{H}\BQ]_{q,k}[\BZ]_{i,q} )
\\
&-\frac{1}{M} \cov( {\Tr\BZ\BZ^{H}\BQ},[\BF_{S,\alpha}\BY\FT\BG_{T,\alpha}\BY\BZ^{H}\BQ ]_{q,k}[\BZ]_{i,q})
=\frac{s_{q}\alpha_{\overline{\omega}}}{L(1+s_{q}\alpha_{\delta}\alpha_{\overline{\omega}})}\E[\FR\BQ]_{i,k}
+\varepsilon_{1,q,i,k}+\varepsilon_{2,q,i,k}.
\end{aligned}
\end{equation}
Therefore, by summing over $q$ and utilizing the resolvent identity~(\ref{res_ide}), we have 
\begin{equation}
\begin{aligned}
\delta(i-k)-\E z[\BQ]_{i,k} =\E [\BH\BH^{H}\BQ]_{i,k} = \frac{M\alpha_{\overline{\omega}}\alpha_{\omega}}{L\alpha_{\delta}}\E[\FR\BQ]_{i,k}
+\varepsilon_{i,k},
\end{aligned}
\end{equation}
where $\varepsilon_{i,k}=\sum_{q}\varepsilon_{1,q,i,k}+\varepsilon_{2,q,i,k}$. Thus, we have
\begin{equation}
\begin{aligned}
\E[\BQ]_{i,k} =(\frac{M\alpha_{\overline{\omega}}\alpha_{\omega}}{L\alpha_{\delta}}r_{i}+z)^{-1}\delta(i-k)+\varepsilon_{I,i,k},
\end{aligned}
\end{equation}
where $\varepsilon_{I,i,k}=-\varepsilon_{i,k}(\frac{M\alpha_{\overline{\omega}}\alpha_{\omega}}{L\alpha_{\delta}}r_{i}+z)^{-1}$. Therefore, for any deterministic matrix $\bold{A}$ with bounded norm, there holds true that
\begin{equation}
\begin{aligned}
\E\Tr\BA\BQ=\Tr\BA\BG_{R,\alpha}
+\varepsilon_{A},
\end{aligned}
\end{equation}
where 
\begin{equation}
\begin{aligned}
&
\varepsilon_{A}=\sum_{i,k}A_{k,i}\varepsilon_{I,i,k}=
\frac{\alpha_{\overline{\omega}}}{L} \cov( {\Tr\FR\BQ},\Tr\BA\BG_{R,\alpha} \BZ\FS\BF_{S,\alpha}\BY\RT\BH^{H}\BQ )
\\
&
+\frac{1}{M} \cov( {\Tr\BZ\BZ^{H}\BQ},\Tr\BA\BG_{R,\alpha}\BZ\BF_{S,\alpha}\BY\FT\BG_{T,\alpha}\BY\BZ^{H}\BQ)
=\varepsilon_{A,1}+\varepsilon_{A,2}.
\end{aligned}
\end{equation}
Next, we will show that $\varepsilon_{A}$ is of order $\BO(\frac{1}{N})$ by the variance control in Proposition~\ref{var_con}.
Noticing that $\| \BG_{R,\alpha} \| \le \frac{1}{z} $, by Cauchy-Schwarz inequality and $\Var(\Tr\FR\BQ)\le \frac{K}{z^4}$, the following bound holds true
\begin{equation}
\label{eps_A1}
\begin{aligned}
|\varepsilon_{A,1}|
\le \frac{\alpha_{\overline{\omega} }}{L}\Var^{\frac{1}{2}}(\Tr\FR\BQ)\Var^{\frac{1}{2}}(\Tr\BA\BG_{R,\alpha}\BZ\FS\BF_{S,\alpha}\BY\RT\BH^{H}\BQ)
\overset{(a)}{=}   \BO(\frac{\mathcal{P}_{1}(\frac{1}{z})}{Nz^4}),
\end{aligned}
\end{equation}
where $\mathcal{P}_{1}$ is a polynomial defined in Proposition~\ref{var_con}. Step $(a)$ follows from the variance control in Proposition~\ref{var_con}.
Similarly, we can obtain $\varepsilon_{A,2}=\BO(\frac{\mathcal{P}_{2}(\frac{1}{z})}{Nz^3})$. Therefore, for a given $z$, we can obtain $\varepsilon_{A}$ is a $\BO(\frac{\mathcal{P}_{2}(\frac{1}{z})}{Nz^3})$ term and
\begin{equation}
\begin{aligned}
\E\Tr\bold{A}\BQ=\Tr\bold{A}\BG_{R,\alpha}+  \BO(\frac{\mathcal{P}_{2}(\frac{1}{z})}{Nz^3}).
\end{aligned}
\end{equation}
From the definition of $\BG_{R,\alpha}$ in~(\ref{int_qua}), we know that $\BG_{R,\alpha}$ depends on $e_{\omega}$ and $\alpha_{\delta}$, which have not been determined yet. We will make a further step to evaluate  $e_{\omega}$. By replacing $\E s_{q}[\BY\RT\BH^{H}\BQ]_{q,i}[\BZ]_{i,q}$ with~(\ref{HHQ}) in~(\ref{ZZQ}), we have
\begin{equation}
\begin{aligned}
{\E \Tr\BZ\BZ^{H}\BQ} =\alpha_{\delta}\Tr\FS-\alpha_{\delta}^2\alpha_{\overline{\omega}}\Tr\FS^2\BF_{S,\alpha} 
+\varepsilon_{Z}
=
{\alpha_{\delta} \Tr\FS\BF_{S,\alpha}}+\varepsilon_{Z}=M\alpha_{\omega}+\varepsilon_{Z},
\end{aligned}
\end{equation}
where
\begin{equation}
\label{handle_epsilz}
\begin{aligned}
\varepsilon_{Z}&= \cov( \frac{\alpha_{\overline{\omega}}\Tr\FR\BQ}{L},\Tr\BZ\FS\BF_{S,\alpha}\BY\RT\BH^{H}\BQ )
+\alpha_{\delta}\cov( \frac{\Tr\BZ\BZ^{H}\BQ}{M},\Tr\BZ\FS\BF_{S,\alpha}\BY\RT(\bold{I}+\alpha_{\overline{\omega}}\FT)^{-1}\BH^{H}\BQ )=Z_{1}+Z_{2}.
\end{aligned}
\end{equation}
Similar to how~(\ref{eps_A1}) was handled, we have
\begin{equation}
\begin{aligned}
|Z_{1}|
\le \frac{\alpha_{\overline{\omega} }}{L}\Var^{\frac{1}{2}}(\Tr\FR\BQ)\Var^{\frac{1}{2}}(\Tr\BZ\FS\bold{F}_{S,\alpha}\BY\RT\BH^{H}\BQ)
= \BO(\frac{\mathcal{P}_{1}(\frac{1}{z})}{Nz^3}),
\end{aligned}
\end{equation}
and $Z_{2}=\BO(\frac{\mathcal{P}_{2}(\frac{1}{z})}{Nz^3})$. It thus follows from the definition of $e_{\omega}$ in~(\ref{int_qua}) that $e_{\omega}=\alpha_{\omega}+\BO(\frac{\mathcal{P}_{2}(\frac{1}{z})}{N^2 z^3})$.
\end{proof}


\section{Proof of Lemma~\ref{AL_TO_QR}}
\label{pro_AL_TO_QR}
Lemma~\ref{Q_TO_AL} provides an approximation for $\E\Tr\BA\BQ$ depending on $e_{\omega}$ and $\alpha_{\delta}$, and an approximation for $e_{\omega}$ depending on $\alpha_{\delta}$. 
Before start the proof of Lemma~\ref{AL_TO_QR}, we make a further step by giving the following lemma, which provides the approximation of $e_{\omega}$ and $\alpha_{\overline{\omega}}$ depending only on $\alpha_{\delta}$.
\begin{lemma}
\label{beta_app}
Let $(\beta_{\omega},\beta_{\overline{\omega}})$ be the solution for the system of the equations
\begin{equation}
\begin{aligned}
 \begin{cases}
&\beta_{\omega}=\frac{1}{M}\Tr\FS(\frac{1}{\alpha_{\delta}}\bold{I}_{L}+\beta_{\overline{\omega}}\FS)^{-1}
\\
&\beta_{\overline{\omega}}=\frac{1}{M}\Tr\FT(\bold{I}_{M}+\beta_{\omega}\FT)^{-1}.
 \end{cases}
\end{aligned}
\end{equation}
Then there holds true that $\alpha_{\omega}=\beta_{\omega}+\BO(\frac{1}{N^2})$, $e_{\omega}=\beta_{\omega}+\BO(\frac{1}{N^2})$, and $\alpha_{\overline{\omega}}=\beta_{\overline{\omega}}+\BO(\frac{1}{N^2})$.
\end{lemma}
\begin{proof} Denoting $\BG_{S,\beta}=\left(\frac{1}{\alpha_{\delta}}\bold{I}_{L}+ \beta_{\overline{\omega}}\FS  \right)^{-1}$ and $\BG_{T,\beta}=\left(\bold{I}_{M}+ \beta_{{\omega}}\FT  \right)^{-1}$, we can obtain 
\begin{equation}
\label{aw_to_bw}
\alpha_{\omega}-\beta_{\omega}=\frac{\Tr\FS\BG_{S,\alpha}}{M}-\frac{\Tr\FS\BG_{S,\beta}}{M}=\gamma_{\alpha,\beta}(\alpha_{\omega}-\beta_{\omega})+\varepsilon,
\end{equation}
where $\gamma_{\alpha,\beta}=\frac{1}{M}\Tr\FS^2\BG_{S,\alpha}\BG_{S,\beta}\frac{1}{M}\Tr\FT^2\BG_{T,\alpha}\BG_{T,\beta}$ and $\varepsilon=\frac{1}{M}\Tr\FS^2\BG_{S,\alpha}\BG_{S,\beta}\frac{1}{M}\Tr\FT^2\BG_{T,\alpha}\BG_{T,\beta}(e_{\omega}-\alpha_{\omega})=\BO(\frac{\mathcal{P}_{2}(\frac{1}{z})}{N^2z^3})$.
By the boundness of $\| \FS\|$, $\|\FT \|$ and the boundness of the intermediate approximations shown by Lemma~\ref{bound_alp} in Appendix~\ref{bound_lemma}, we have 
\begin{equation}
\begin{aligned}
\gamma_{\alpha,\beta}=\frac{1}{M}\Tr\FS^2\BG_{S,\alpha}\BG_{S,\beta}\frac{1}{M}\Tr\FT^2\BG_{T,\alpha}\BG_{T,\beta}<
\frac{L s_{max}^2 t_{max}^2\alpha_{\delta}^2 }{M}
<
\frac{N^2 s_{max}^2t_{max}^2r_{max}^2}{M L z^2}<\frac{1}{2},
\end{aligned}
\end{equation} 
when $z\in (\sqrt{\frac{2N^2 s_{max}^2t_{max}^2r_{max}^2}{M L }}, \infty)$, which indicates that $\alpha_{\omega}-\beta_{\omega}=\frac{\varepsilon}{1-\gamma_{\alpha,\beta}}<2\varepsilon=\BO(\frac{1}{N^2})$. Next we only need to establish the convergence for $z \in  (0, \sqrt{\frac{2N^2 s_{max}^2t_{max}^2r_{max}^2}{M L }})$. To achieve this goal, we will use a standard argument relying on Montel's theorem, which is widely used in RMT (see e.g.,~\cite{hachem2008new,dumont2010capacity,rubio2011spectral}), and establish the convergence in $(0,\infty)$. $\alpha_{\omega}(z)$ and $\beta_{\omega}(z)$ are both functions with respect to $z\in (0, \infty)$ since $\alpha_{\delta}$ is determined by $z$. The sequence of functions $\alpha_{\omega,N}(z)-\beta_{\omega,N}(z)$ can be extended to $z\in\mathbb{C}/\ \mathbb{R}^{-}$ and $|\alpha_{\omega,N}(z)-\beta_{\omega,N}(z)|\le \frac{2N r_{max}s_{max}}{M\mathrm{dist}(z,\mathbb{R}^{-})}$, where $\mathrm{dist}(\cdot,\cdot)$ represents the Euclidean distance. We can conclude that on each compact subset of $\mathbb{C}/\ \mathbb{R}^{-}$, holomorphic functions $\alpha_{\omega,N}(z)-\beta_{\omega,N}(z)$ are uniformly bounded. By Montel's theorem (the normal family theorem), the sequence of functions $\alpha_{\omega,N}(z)-\beta_{\omega,N}(z)$ is compact and there exists a subsequence which converges uniformly on each compact subset to an analytic function, which is $0$ when $z\in(\sqrt{\frac{2N^2 s_{max}^2t_{max}^2r_{max}^2}{M L }}, \infty)$ and thus it will be zero in $\mathbb{C}/\ \mathbb{R}^{-}$. The entire sequence $\alpha_{\omega,N}(z)-\beta_{\omega,N}(z)$ converges to zero on each compact subset of $\mathbb{C}/\ \mathbb{R}^{-}$. Then, for any $\bold{A}$ with bounded spectral norm, there holds true that 
\begin{equation}
\label{AGBS}
\frac{1}{M}\Tr\bold{A}\BG_{S,\alpha}-\frac{1}{M}\Tr\bold{A}\BG_{S,\beta}=\frac{1}{M}\Tr\bold{A}\BG_{S,\alpha}\BG_{S,\beta}\FS(\alpha_{\overline{\omega}}-\beta_{\overline{\omega}})=o(1).
\end{equation}
However, Montel's theorem only indicates the convergence for $z\in(0,\sqrt{\frac{2N^2 s_{max}^2t_{max}^2r_{max}^2}{M L }})$ but does not guarantee the convergence rate $\BO(\frac{1}{N^2})$. The convergence rate will be obtained by the following analysis. By~(\ref{AGBS}), we have 
\begin{equation}
\label{gamma_ab}
\gamma_{\alpha,\beta}=\frac{1}{M}\Tr\FS^2\BG_{S,\beta}^2\frac{1}{M}\Tr\FT^2\BG_{T,\beta}^2+o(1)=\gamma_{\beta}+o(1).
\end{equation}
We also have the following bound for $\gamma_{\beta}$,
\begin{equation}
\label{bnd_deb}
1=\frac{\beta_{\omega}}{\beta_{\overline{\omega}}}\frac{\beta_{\overline{\omega}}}{\beta_{\omega}}
=\frac{\frac{1}{\alpha_{\delta}}\Tr\FS\BG_{S,\beta}^2+\beta_{\overline{\omega}}\Tr\FS^2\BG_{S,\beta}^2}{M\beta_{\overline{\omega}}}\frac{\Tr\FT\BG_{T,\beta}^2+\beta_{\omega}\Tr\FT^2\BG_{T,\beta}^2}{M\beta_{\omega}}
\ge \gamma_{\beta}+\frac{\Tr\FS\BG_{S,\beta}^2\Tr\FT\BG_{T,\beta}^2}{M^2\alpha_{\delta}\beta_{\omega}\beta_{\overline{\omega}}}.
\end{equation}
Meanwhile, by Cauchy-Schwarz inequality, we can obtain
\begin{equation}
\frac{\Tr\FS\BG_{S,\beta}^2}{M} \ge \frac{M\beta_{\omega}^2}{\Tr\FS},~~\frac{\Tr\FT\BG_{T,\beta}^2}{M} \ge \frac{M\beta_{\overline{\omega}}^2}{\Tr\FT},
\end{equation}
\begin{equation}
\beta_{\omega}\ge \frac{(\Tr\FS)^2}{M\Tr\FS(\frac{1}{\alpha_{\delta}}\bold{I}_{L}+\beta_{\overline{\omega}}\FS)},~~
\beta_{\overline{\omega}}\ge \frac{(\Tr\FT)^2}{M\Tr\FT(\bold{I}_{M}+\beta_{{\omega}}\FT)}.
\end{equation}
Therefore, by Lemma~\ref{bound_alp} and assumptions~\textbf{A.1} to~\textbf{A.3}, there holds true
\begin{equation}
\frac{1}{\Delta_{S,\beta}}\le \frac{LMs_{max}t_{max}}{\Tr\FS\Tr\FT} (1+ \frac{N}{Lz}r_{max}s_{max}t_{max})^2<\infty,
\end{equation}
where $ \Delta_{S,\beta}= 1-\frac{1}{M}\Tr\FS^2\BG_{S,\beta}^2\frac{1}{M}\Tr\FT^2\BG_{T,\beta}^2=1-\gamma_{\beta}$. 
By~(\ref{AGBS}),~(\ref{bnd_deb}) and~(\ref{gamma_ab}), we know that $\Delta_{S,\alpha}$ converges to $\Delta_{S,\beta}$ uniformly on a compact set, so there exists $N_{0}$ such that when $N>N_{0}$,
\begin{equation}
\frac{1}{\Delta_{S,\alpha}}=\frac{1}{1-\gamma_{\alpha,\beta}} \le  \frac{LMs_{max}t_{max}}{\Tr\FS\Tr\FT} (1+ \frac{N}{Lz}r_{max}s_{max}t_{max})^2=\frac{(1+k_{S,1}z)^2}{k_{S,0}z^2}.
\end{equation}
Therefore, when $N$ is large enough, we can further obtain
\begin{equation}
|\alpha_{\omega}-\beta_{\omega}|=\frac{\BO(\frac{\mathcal{P}_{2}(\frac{1}{z})}{z^3 N^2})}{\Delta_{S,\beta}}\le  \frac{(\frac{1}{z}+k_{S,1})^2\mathcal{P}_{2}(\frac{1}{z}) }{N^2 k_{S,0}z^3}
=\frac{\mathcal{P}_{4}(\frac{1}{z}) }{N^2 z^3}.
\end{equation}
This concludes $\alpha_{\omega}=\beta_{\omega}+\BO(\frac{\mathcal{P}(\frac{1}{z})}{z^2 N^2})$, $e_{\omega}=\alpha_{\omega}+\BO(\frac{\mathcal{P}(\frac{1}{z})}{z^2 N^2})=\beta_{\omega}+\BO(\frac{\mathcal{P}(\frac{1}{z})}{z^2 N^2})$, and $\alpha_{\overline{\omega}}=\beta_{\overline{\omega}}+\BO(\frac{\mathcal{P}(\frac{1}{z})}{z^2 N^2})$.

In fact, we can also bound $1-\gamma_{\alpha,\beta}$ as
\begin{equation}
\begin{aligned}
&1-\frac{{\alpha}_{\omega}}{M\beta_{\overline{\omega}}}\Tr\FT^2\BG_{T,\alpha}\BG_{T,\beta}+\frac{{\alpha}_{\omega}}{M\beta_{\overline{\omega}}}\Tr\FT^2\BG_{T,\alpha}\BG_{T,\beta}-   \frac{1}{M}\Tr\FS^2\BG_{S,\alpha}\BG_{S,\beta}\frac{1}{M}\Tr\FT^2\BG_{T,\alpha}\BG_{T,\beta}
\ge \frac{\Tr\FT\BG_{T,\alpha}\BG_{T,\beta}}{M\beta_{\overline{\omega}}}
\\
&\overset{(a)}{\ge}  \frac{1}{(1+   \frac{N}{Mz} r_{max} s_{max}  t_{max} )},
\end{aligned}
\end{equation}
where $(a)$ follows from~(\ref{par_bound}) as $\BG_{T,\alpha}\ge \frac{1}{(1+   \frac{N}{Mz} r_{max} s_{max}  t_{max} )}\BI_{M}  $.
By~(\ref{aw_to_bw}), we can obtain
\begin{equation}
\label{accura_error1}
\begin{aligned}
\alpha_{\omega}-\beta_{\omega}&=\BO(\frac{\mathcal{P}_{4}(\frac{1}{z})}{N^2z^2}),
\\
\alpha_{\overline{\omega}}-\beta_{\overline{\omega}}&=\BO(\frac{\mathcal{P}_{4}(\frac{1}{z})}{N^2z^2}).
\end{aligned}
\end{equation}
\end{proof}

Now, we prove Lemma~\ref{AL_TO_QR}.

\begin{proof} 
We first show $\alpha_{\delta}=\delta+\BO(\frac{1}{N^2})$. By Lemma~\ref{beta_app}, we know that $e_{\omega}$, $\alpha_{\omega}$ and $\alpha_{\overline{\omega}}$ can be approximated based on only $\alpha_{\delta}$. Meanwhile, by the steps in the proof of Lemma~\ref{beta_app}, we have
\begin{equation}
\begin{aligned}
\beta_{\omega}-\omega=\frac{\Tr\FS\BG_{S,\beta}\BG_{S}(\alpha_{\delta}-\delta)}{M\delta\alpha_{\delta}\Delta_{S,\beta}}, 
~~
\beta_{\overline{\omega}}-\overline{\omega}=-\frac{\Tr\FS\BG_{S,\beta}\BG_{S}\Tr\FT^2\BG_{T,\beta}\BG_{T} (\alpha_{\delta}-\delta)}{M^2\delta\alpha_{\delta}\Delta_{S,\beta}}.
\end{aligned}
\end{equation}
According to~(\ref{accura_error1}) in the proof of Lemma~\ref{beta_app}, we can replace $\alpha_{\omega}$ and $\alpha_{\overline{\omega}}$ with $\beta_{\omega}$ and $\beta_{\overline{\omega}}$, respectively, which only introduces the error $\BO(\frac{\mathcal{P}_{5}(\frac{1}{z})}{N^2z^4})$. Therefore, we have
\begin{equation}
\begin{aligned}
&\alpha_{\delta}-\delta=\frac{M \Tr\FR^2\BG_{R,\alpha}\BG_{R}}{L^2}\frac{\omega\overline{\omega}}{\delta\alpha_{\delta}}(\alpha_{\delta}-\delta)-\frac{M \Tr\FR^2\BG_{R,\alpha}\BG_{R}}{L^2}\frac{\overline{\omega}}{\alpha_{\delta}}(\alpha_{\omega}-\omega)
\\
&
 -\frac{M}{L^2}\Tr\FR^2\BG_{R,\alpha}\BG_{R}\frac{\alpha_{\omega}}{\alpha_{\delta}}(\alpha_{\overline{\omega}}-\overline{\omega})
 +\BO(\frac{\mathcal{P}_{2}(\frac{1}{z})}{N^2z^3})
=\frac{M}{L^2}\Tr\FR^2\BG_{R,\alpha}\BG_{R}\frac{\omega\overline{\omega}}{\delta\alpha_{\delta}} (\alpha_{\delta}-\delta)
\\
&
-\frac{M}{L^2}\Tr\FR^2\BG_{R,\alpha}\BG_{R}\frac{1}{M}\Tr\FS\BG_{S}\BG_{S,\beta} \frac{1}{M}\Tr\FT\BG_{T}\BG_{T,\beta}\frac{(\alpha_{\delta}-\delta)}{\alpha_{\delta}^2\delta\Delta_{S,\beta}}+\BO(\frac{\mathcal{P}_{1}(\frac{1}{z})}{z}  \frac{\mathcal{P}_{4}(\frac{1}{z})}{N^2z^2} +\frac{\mathcal{P}_{2}(\frac{1}{z})}{N^2z^3} ),
\end{aligned}
\end{equation}
from which we can obtain the expression of $\alpha_{\delta}-\delta$. Given $z_{0}=\max(\frac{r_{max}^2 s_{max} t_{max}}{\overline{r}},2\sqrt{\frac{r_{max}^2 s_{max} t_{max}}{\overline{r}}})$ and $z\in (z_{0},\infty)$, we have
\begin{equation}
\begin{aligned}
\frac{M \Tr\FR^2\BG_{R,\alpha}\BG_{R}}{L^2}\frac{\omega\overline{\omega}}{\delta\alpha_{\delta}} 
\overset{(a)}{\le}  \frac{N r_{max}^2s_{max}t_{max}}{L z^2} \frac{L(z+\frac{r_{max}^2 s_{max} t_{max}}{\overline{r}})}{N\overline{r}} <  \frac{1}{2},
\end{aligned}
\end{equation}
where $(a)$ follows from~(\ref{par_bound}) in Lemma~\ref{bound_alp}. Then we can obtain $\Delta_{\alpha}>\frac{1}{2} $ so that
\begin{equation}
|\alpha_{\delta}-\delta|=\frac{\BO(\frac{\mathcal{P}_{5}(\frac{1}{z})}{N^2z^3})}{\Delta_{\alpha}}=\BO(\frac{2\mathcal{P}_{5}(\frac{1}{z})}{N^2z^3}),
\end{equation}
where $\Delta_{\alpha}=1-\frac{M}{L^2}\Tr\FR^2\BG_{R,\alpha}\BG_{R}\frac{\omega\overline{\omega}}{\delta\alpha_{\delta}}+\frac{M}{L^2\alpha_{\delta}^2\delta\Delta_{S,\beta}}\Tr\FR^2\BG_{R,\alpha}\BG_{R}\frac{1}{M}\Tr\FS\BG_{S}\BG_{S,\beta} \frac{1}{M}\Tr\FT\BG_{T}\BG_{T,\beta}$. Next, we will show the convergence for $z \in (0,z_{0})$ by similar lines in the proof of Lemma~\ref{beta_app}. We have $|\alpha_{\delta,N}(z)-\delta_{N}(z)|=\frac{2}{\mathrm{dist}(z,\mathbb{R}^{-})}$ so that $\alpha_{\delta,N}(z)-\delta_{N}(z)$ is a normal family. By Montel's theorem, we have $\alpha_{\delta,N}(z)-\delta_{N}(z) $ converges to zero so that 
\begin{equation}
\frac{1}{M}\Tr\bold{A}\BG_{R,\alpha}-\frac{1}{M}\Tr\bold{A}\BG_{R}=o(1)
\end{equation}
 for $z \in \mathbb{C}/\ \mathbb{R}^{-}$, where $\bold{A}$ has a bounded norm. Let $\bold{A}=\FR^2\BG_{R}$, we have $\frac{1}{L}\Tr\FR^2\BG_{R,\alpha}\BG_{R}=\frac{1}{L}\Tr\FR^2\BG_{R}^2 +o(1)=\nu_{R}+o(1)$, where $\nu_{R}$ is given in Table~\ref{var_list}. Similarly, for $\beta_{\omega}$ and $\beta_{\overline{\omega}}$, we have
 \begin{equation}
\begin{aligned}
\beta_{\omega}-\omega=\frac{\Tr\FS\BG_{S,\beta}\BG_{S}(\alpha_{\delta}-\delta)}{M\alpha_{\delta}\delta\Delta_{S,\beta}}=o(1),
 \end{aligned}
\end{equation}
 and $\beta_{\overline{\omega}}-\overline{\omega}=o(1)$, such that $\frac{\Tr\bold{B}\BG_{S,\beta}}{M}=\frac{\Tr\bold{B}\BG_{S}}{M}+o(1)$ and $\frac{\Tr\bold{C}\BG_{T,\beta}}{M}=\frac{\Tr\bold{C}\BG_{T}}{M}+o(1)$, where $\bold{B}$ and $\bold{C}$ are deterministic matrices with bounded spectral norm. By letting $\bold{B}=\FS\BG_{S}$ and $\bold{C}=\FT\BG_{T}$, we have $\frac{1}{M}\Tr\FS\BG_{S}\BG_{S,\beta}=\nu_{S,I}+o(1)$, $\frac{1}{M}\Tr\FT\BG_{T}\BG_{T,\beta}=\nu_{T,I}+o(1)$ and $\Delta_{S,\beta}=\Delta_{S}+o(1)$. Therefore, we can obtain
  \begin{equation}
\begin{aligned}
\Delta_{\alpha}=1-\frac{M \omega\overline{\omega}\nu_{R}}{L\delta^2}+\frac{M \nu_{R}\nu_{S,I}\nu_{T,I}}{L\delta^3\Delta_{S}}+o(1)=\Delta+o(1).
  \end{aligned}
\end{equation}
Next, we will show that $\inf \Delta >0$, which follows from Lemma~\ref{bound_alp} and assumptions~\textbf{A.1} to~\textbf{A.3},
\begin{equation}
\begin{aligned}
\Delta \ge 1-\frac{M\omega\overline{\omega}\nu_{R}}{L\delta^2}=\frac{z\nu_{R,I}}{\delta}>\frac{N\underline{r}z}{L(z+r_{max}s_{max}t_{max})^2}\frac{Lz}{Nr_{max}}>0.
\end{aligned}
\end{equation}
Therefore, there exists $N_{1}$ such that when $N>N_{1}$, 
\begin{equation}
\Delta_{\alpha}>\frac{N\underline{r}z}{L(z+r_{max}s_{max}t_{max})^2}\frac{Lz}{Nr_{max}},
\end{equation}
so that
\begin{equation}
|\alpha_{\delta}-\delta|=\frac{|\varepsilon_{r}|}{\Delta_{\alpha}} \le \frac{l_{1}(1+\frac{l_{1}}{z})^2\mathcal{P}(\frac{1}{z})}{N^2 l_{0}z^2}=\BO(\frac{\mathcal{P}(\frac{1}{z})}{N^2z^2}),
\end{equation}
where $l_{0}$ and $l_{1}$ are independent of $z$ and $N$.
We have $\alpha_{\delta}-\delta=\BO(\frac{\mathcal{P}(\frac{1}{z})}{z^2 N^2})$, which can be further used to obtain $\alpha_{\omega}-\omega=\BO(\frac{\mathcal{P}(\frac{1}{z})}{z^2 N^2})$ and $\alpha_{\overline{\omega}}-\overline{\omega}=\BO(\frac{\mathcal{P}(\frac{1}{z})}{z^2 N^2})$. It thus follows that 
\begin{equation}
\frac{1}{M}\E\Tr\bold{A}\BQ-\frac{1}{M}\Tr\bold{A}\BG_{R}=\frac{1}{M}\Tr\bold{A}\BG_{R,\alpha}\BG_{R}\FR(\frac{M \alpha_{\omega}\alpha_{\overline{\omega}}}{L\alpha_{\delta}}-\frac{M \omega\overline{\omega}}{L\delta})  + \BO(\frac{\mathcal{P}(\frac{1}{z})}{N^2z^2})
=\BO(\frac{\mathcal{P}(\frac{1}{z})}{N^2z^2})
\end{equation}
since $\frac{M \alpha_{\omega}\alpha_{\overline{\omega}}}{L\alpha_{\delta}}-\frac{M \omega\overline{\omega}}{L\delta}=\BO(\frac{\mathcal{P}(\frac{1}{z})}{N^2z^2})$. In fact, $\Delta_{\alpha}$ can be bounded by
\begin{equation}
\Delta_{\alpha} \ge \frac{z\Tr\FR\BG_{R,\alpha}\BG_{R}}{L\alpha_{\delta}} \ge \frac{ N\overline{r}}{zL(1+ \frac{1}{z}r_{max}s_{max}t_{max} )^2}\frac{Lz}{Nr_{max}},
\end{equation}
so that the approximation errors can also be bounded by $|\alpha_{\delta}-\delta|=\frac{\BO(\frac{\mathcal{P}_{5}(\frac{1}{z})}{N^2z^3})}{\Delta_{\alpha}}=\BO(\frac{\mathcal{P}_{8}(\frac{1}{z})}{N^2z^2})$ and 
\begin{equation}
\label{trace_QG_err}
|\E\Tr\BQ-\Tr\BG_{R}|=\frac{\BO(\frac{\mathcal{P}_{5}(\frac{1}{z})}{Nz^3})}{\Delta_{\alpha}}=\BO(\frac{\mathcal{P}_{8}(\frac{1}{z})}{Nz^2}).
\end{equation}

\end{proof}
\section{Proof of Proposition~\ref{bound_var}}
\label{proof_bound_var}
\begin{proof} By letting $z=\sigma^2$, $V(z)$ can be rewritten as 
\begin{equation}
V(z)=-\log(\Delta_v(z)),
\end{equation}
where $\Delta_v(z)=\Delta\Delta_{S}=\frac{z\nu_{R,I}(1-\nu_S\nu_T)}{\delta}+\frac{M\nu_R\nu_{S,I}\nu_{T,I}}{L\delta^3}$. Now we only need to prove that $\Delta_v(z)$ is bounded away from zero and is strictly lower than $1$. First, we have


\begin{equation}
\label{deltas_relax}
1=\frac{\omega}{\overline{\omega}}\frac{\overline{\omega}}{\omega}
=\frac{\frac{\nu_{S,I}}{\delta}+\overline{\omega}\nu_{S}}{\overline{\omega}}\frac{\nu_{T,I}+\omega\nu_{T}}{\omega}
\ge \nu_{S}\nu_{T}+\frac{\nu_{S,I}\nu_{T,I}}{\delta \omega\overline{\omega}}.
\end{equation}
By Cauchy-Schwarz inequality, we can obtain
\begin{equation}
\nu_{R,I}\ge \frac{L\delta^2}{N\overline{r}},~~\nu_{S,I} \ge \frac{M\omega^2}{L\overline{s}},~~\nu_{T,I} \ge \frac{\overline{\omega}^2}{\overline{t}},
\end{equation}
and the lower bounds for $\omega$ and $\overline{\omega}$
\begin{equation}
\label{ome_barome}
\omega\ge \frac{(\Tr\FS)^2}{M\Tr\FS(\frac{1}{\delta}\bold{I}_{L}+\overline{\omega}\FS)},~~
\overline{\omega} \ge \frac{(\Tr\FT)^2}{M\Tr\FT(\bold{I}_{M}+\omega\FT)}.
\end{equation}
By Lemma~\ref{bound_alp} and~(\ref{deltas_relax}) to~(\ref{ome_barome}), $\Delta_v(z)$ can be lower bounded by
\begin{equation}
\label{effect_low_bnd}
\begin{aligned}
&\Delta_v(z) \ge \frac{z\nu_{R,I}(1-\nu_{S}\nu_{T}) }{\delta} \ge 
\frac{zL}{N\overline{r}}\frac{M\omega\overline{\omega}}{L\overline{s}\overline{t}}
\ge    \frac{zL}{N\overline{r}(\frac{L(z+ {r_{max}s_{max} t_{max}}  )}{N\overline{r}}+s_{max}t_{max})(1+ \frac{N r_{max}s_{max}t_{max}}{Mz} )}
\\
& \ge \frac{1}{(1+\frac{r_{max}s_{max}t_{max}}{z}+\frac{N \overline{r} s_{max}t_{max}}{Lz})(1+ \frac{N r_{max}s_{max}t_{max}}{Mz} )}.
\end{aligned}
\end{equation}
By assumptions~\textbf{A.2} and~\textbf{A.3} ($\inf\limits_{N}\overline{r}>0$, $\inf\limits_{N}\overline{s}>0$, and $\inf\limits_{N}\overline{t}>0$), we can obtain that there exists $m_{z}$ such that
\begin{equation}
\inf_{N}\Delta_v(z) \ge m_{z} >0. 
\end{equation}
Then, we will show that $\sup\limits_{N}\Delta_v(z) \le M_{z} <1$. By the upper bounds in Lemma~\ref{bound_alp}, we can obtain
\begin{equation}
\label{tmp_relax}
\begin{aligned}
&\frac{z\nu_{R,I}\nu_{S}\nu_{T}}{\delta} \ge 
\frac{z\delta\nu_{R,I}\frac{L\overline{s}^2\overline{t}^2 }{M}}{(1+\frac{N r_{max}s_{max}t_{max}}{zL})^2(1+\frac{N r_{max}s_{max}t_{max}}{zM})^2 }
\\
&\ge\frac{\frac{N^2}{L^2}\frac{L\overline{r}^2\overline{s}^2\overline{t}^2 }{Mz^2}}{(1+\frac{N r_{max}s_{max}t_{max}}{zL})^2(1+\frac{N r_{max}s_{max}t_{max}}{zM})^2(1+ \frac{r_{max}s_{max}t_{max}}{z})^3}:= L_{z}.
\end{aligned}
\end{equation}
By~\textbf{A.1},~\textbf{A.2} and~\textbf{A.3}, it is easy to verify $0<L_{z}<1$. By~(\ref{tmp_relax}), $\Delta_v(z) $ can be upper bounded as
\begin{equation}
\Delta_{v}(z) = -\frac{z\nu_{R,I}\nu_{S}\nu_{T}}{\delta} +\frac{z\nu_{R,I}}{\delta} +\frac{M\nu_R\nu_{S,I}\nu_{T,I}}{L\delta^3}=
-\frac{z\nu_{R,I}\nu_{S}\nu_{T}}{\delta} +1-\frac{M\omega\overline{\omega}\nu_{R}}{L\delta^2}+\frac{M\nu_R\nu_{S,I}\nu_{T,I}}{L\delta^3}
\overset{(a)}{\le} 1-\frac{z\nu_{R,I}\nu_{S}\nu_{T}}{\delta}
\le 1-L_{z} <1.
\end{equation}
where the inequality $(a)$ follows from $\frac{\nu_{S,I}}{\delta}\le \omega$ and $\nu_{T,I}\le \overline{\omega}$. By assumptions~\textbf{A.1} and~\textbf{A.2}, we can conclude that there exists a $M_{z}>0$ such that
\begin{equation}
\sup_{N} \Delta_{v}(z)\le M_z <1.
\end{equation}
Therefore, there exist two positive numbers $b_{\sigma^2}$ and $B_{\sigma^2}$ satisfying
\begin{equation}
0<b_{\sigma^2} \le \inf_{N} V(\sigma^2) \le  \sup_{N} V(\sigma^2) \le B_{\sigma^2} \le \infty.
\end{equation}
\end{proof}

\section{Proof of Lemma~\ref{ka_mu}}
\label{pro_ka_mu}
\begin{proof}
By using the integration by parts formula~(\ref{int_part}) on $Y_{m,i}^{*}$ and following the lines from~(\ref{QYYZZ}), we can obtain
\begin{equation}
\label{QABCH}
\begin{aligned}
&
\E [\RT\bold{Y}^{H}\MS\bold{X}^{H}\LR\BQ\bold{A}\bold{X}\bold{B}]_{i,q}[\bold{Y}\bold{C}]_{q,i}
=\E\sum_{m} [\RT\bold{Y}^{H}]_{i,m}[\MS\bold{X}^{H}\LR\BQ\bold{A}\bold{X}\bold{B}]_{m,q}[\bold{Y}\bold{C}]_{q,i}
\\
&
=\E\frac{- t_{i}}{M} \Tr\BZ\BZ^{H}\BQ [ \BH^{H}\BQ\bold{A}\bold{X}\bold{B}]_{i,q}[\bold{Y}\bold{C}]_{q,i}
+\frac{1}{M}[\FT^{\frac{1}{2}}\bold{C} ]_{i,i}[\MS\bold{X}^{H}\LR\BQ\bold{A}\bold{X}\bold{B}]_{q,q}
\\
&=\E\frac{1}{M}[\RT\bold{C}\BG_{T,\alpha}]_{i,i}[\MS\BX^{H}\LR\BQ\bold{A}\bold{X}\bold{B}]_{q,q}
-\frac{t_{i}}{M}\cov(\Tr\BZ\BZ^{H}\BQ,[ \BH^{H}\BQ\bold{A}\bold{X}\bold{B}]_{i,q}[\bold{Y}\bold{C}\BG_{T,\alpha}]_{q,i}).
\end{aligned}
\end{equation}
By using the integration by parts formula~(\ref{int_part}) on $X_{m,i}^{*}$, we have
\begin{equation}
\begin{aligned}
\label{QABX}
&\E[\MS\bold{X}^{H}\LR\BQ\bold{A}\bold{X}\bold{B}]_{i,i}
=\E\sum_{m,q}[\MS\bold{X}^{H}\LR]_{i,m}[\BQ]_{m,q}[\bold{A}\bold{X}\bold{B}]_{q,i}
\\
&
=\E \sum_{m,q}\frac{1}{L}[\bold{A}]_{q,m}[\MS\bold{B}]_{i,i}[\LR\BQ]_{m,q}
-\frac{1}{L} [\BQ\FR]_{m,m} [\FS\BY\RT\BH^{H}\BQ]_{i,q}[\bold{A}\bold{X}\bold{B}]_{q,i}
\\
&=\E \frac{1}{L}\Tr\bold{A}\LR\BQ  [\MS\bold{B}]_{i,i}-\frac{1}{L}\Tr\FR\BQ [\FS\BY\RT\BH^{H}\BQ\bold{A}\bold{X}\bold{B}]_{i,i}.
\end{aligned}
\end{equation}
If we let $\bold{C}=\RT$ and plug~(\ref{QABX}) into~(\ref{QABCH}), we can obtain
\begin{equation}
\begin{aligned}
\E[\bold{Y}\FT\bold{Y}^{H}\MS\bold{X}^{H}\LR\BQ\bold{A}\bold{X}\bold{B}]_{q,q}
=\E\frac{\alpha_{\overline{\omega}}}{L}\Tr \bold{A}\LR\bold{Q}[\MS\bold{B}\BF_{S,\alpha}]_{q,q}+\varepsilon_{T,q},
\end{aligned}
\end{equation}
where $\varepsilon_{T,q}=-\frac{\alpha_{\omega}}{L}\cov(\Tr\FR\BQ,[\BF_{S,\alpha}\FS\BY\RT\BH^{H}\BQ\bold{A}\bold{X}\bold{B}]_{q,q})-\frac{1}{M}\cov(\Tr\BZ\BZ^{H}\BQ,[ \bold{Y}\bold{C}\BG_{T,\alpha}\FT\BH^{H}\BQ\bold{A}\bold{X}\bold{B}\BF_{S,\alpha}]_{q,q})$.

By summing over $q$ and utilizing the variance control in Lemma~\ref{var_con}, we have
\begin{equation}
\begin{aligned}
\label{QAB}
\mu(\BA,\BB)
&
=\frac{1}{L}\E\Tr\bold{A}\LR\BQ\frac{1}{M}\Tr\MS\bold{B}\BF_{S,\alpha}+\BO_z(\frac{1}{N^2})
\\
&
\overset{(a)}{=}\frac{1}{L}\E\Tr\bold{A}\LR\BG_{R}\frac{1}{M}\Tr\MS\bold{B}\BF_{S}+\BO_z(\frac{1}{N^2}),
\end{aligned}
\end{equation}
where $(a)$ follows from Theorem~\ref{fir_the}. This proves~(\ref{exp_mu}).
By plugging~(\ref{QAB}) into~(\ref{QABCH}), we have
\begin{equation}
\label{QABCH_exp}
\begin{aligned}
\kappa(\BA,\BB,\BC)
=\frac{1}{M}\Tr\RT\bold{C}\bold{G}_{T}\frac{1}{L}\Tr\bold{A}\LR\BG_{R}\frac{1}{M}\Tr\MS\bold{B}\BF_{S}
+\BO_z(\frac{1}{N^2}),
\end{aligned}
\end{equation}
which proves~(\ref{exp_ka}).

\end{proof}


\section{Proof of Lemma~\ref{sec_app}}
\label{pro_qua}
\begin{proof} In this section, we use $\BO(\frac{1}{N})$ to represent $\BO(\frac{\mathcal{P}(\frac{1}{z})}{z^2 N^2})$, which can be verified by similar lines in Section~\ref{pro_ka_mu}. Furthermore, we consider the more general evaluations $\Upsilon(\BA,\BB,\BC,\BM)=\frac{1}{M}\E\Tr\BA\BX\BB\BY\BC\BH^{H}\BQ\BM\BQ$ and $\zeta(\BA,\BB,\BM)=\frac{1}{M}\E\Tr\BA\BX\BB\BZ^{H}\BQ\BM\BQ$. In particular, we have $\Upsilon(\BA,\BB,\BC,\FR)=\Upsilon(\BA,\BB,\BC)$ and $\zeta(\BA,\BB,\FR)=\zeta(\BA,\BB)$.
\subsection{$\Gamma(\BA,\BB,\BC)$ and $\chi(\BA,\BB)$}
\label{gamma_chi}
By using the integration by parts formula over $Y_{m,i}^{*}$, we have
\begin{equation}
\label{QZZQABCH}
\begin{aligned}
&
\E [\BH^{H}\BQ\BZ\BZ^{H} \BQ\bold{A}\bold{X}\bold{B}]_{i,q}[\bold{Y}\bold{C}]_{q,i}
=\E\sum_{m} [\RT\bold{Y}^{H}]_{i,m}[\BZ^{H}\BQ\BZ\BZ^{H}\BQ\bold{A}\bold{X}\bold{B}]_{m,q}[\bold{Y}\bold{C}]_{q,i}
\\
&
=\E \{-\frac{ t_{i}}{M} \Tr\BQ\BZ\BZ^{H} [ \BH^{H}\BQ\BZ\BZ^{H} \BQ\bold{A}\bold{X}\bold{B}]_{i,q}[\bold{Y}\bold{C}]_{q,i}
-\frac{ t_{i}}{M} \Tr\BQ\BZ\BZ^{H} \BQ\BZ\BZ^{H}  [ \BH^{H}\BQ\bold{A}\bold{X}\bold{B}]_{i,q}[\bold{Y}\bold{C}]_{q,i}
\\
&+\frac{1}{M}[\FT^{\frac{1}{2}}\bold{C} ]_{i,i}[\MS\bold{X}^{H}\LR\BQ\BZ\BZ^{H} \BQ\bold{A}\bold{X}\bold{B}]_{q,q}\}
\\
&\overset{(a)}{=}\E\{\frac{1}{M}[\RT\bold{C}\bold{G}_{T,\alpha}]_{i,i}[\MS\BX^{H}\LR\BQ\BZ\BZ^{H} \BQ\bold{A}\bold{X}\bold{B}]_{q,q}
\\
&
-\frac{1}{M} \Tr\BQ\BZ\BZ^{H} \BQ\BZ\BZ^{H}  [\FT\bold{G}_{T,\alpha} \BH^{H}\BQ\bold{A}\bold{X}\bold{B}]_{i,q}[\bold{Y}\bold{C}]_{q,i}\}
+\varepsilon_{H,q,i},
\end{aligned}
\end{equation}
where $\varepsilon_{H,q,i}=-\cov(\frac{\Tr\BZ\BZ^{H}\BQ}{M}, [ \BH^{H}\BQ\bold{A}\bold{X}\bold{B}]_{i,q}[\bold{Y}\bold{C}\BG_{T,\alpha}\FT]_{q,i})$ and $(a)$ can be obtained according to lines from~(\ref{QHHP0}) to~(\ref{QHHP1}). Summing over $i$, we can obtain
\begin{equation}
\label{QZZQABCH}
\begin{aligned}
& \E[\bold{Y}\bold{C}\BH^{H}\BQ\BZ\BZ^{H} \BQ\bold{A}\bold{X}\bold{B}]_{q,q}
=\E\{\frac{\Tr\BC\RT\BG_{T,\alpha}}{M}[\BZ^{H}\BQ\BZ\BZ^{H} \BQ\bold{A}\bold{X}\bold{B}]_{q,q}
\\
&
-\frac{1}{M} \Tr\BQ\BZ\BZ^{H} \BQ\BZ\BZ^{H}  [\bold{Y}\bold{C}\FT\bold{G}_{T,\alpha} \BH^{H}\BQ\bold{A}\bold{X}\bold{B}]_{q,q}  \} +\varepsilon_{H,q},
\end{aligned}
\end{equation}
where $\varepsilon_{H,q}=\sum_{i}\varepsilon_{H,q,i}=-\cov(\frac{\Tr\BZ\BZ^{H}\BQ}{M},[\bold{Y}\bold{C}\FT\BG_{T,\alpha} \BH^{H}\BQ\bold{A}\bold{X}\bold{B}]_{q,q})$. By summing over $q$ and utilizing Lemma~\ref{ka_mu}, we can obtain
\begin{equation}
\Gamma(\BA,\BB,\BC)=\frac{\Tr\BC\RT\BG_{T}}{M}\chi(\BA,\BB)-\frac{\Tr\BA\LR\BG_{R}}{L}\frac{\Tr\BB\MS\BF_{S}}{M}\frac{\Tr\BC\FT^{\frac{3}{2}}\BG^2_{T}}{M}\chi+\varepsilon_{\Gamma},
\end{equation}
where $\varepsilon_{\Gamma}$ comes from the substitution of $\BG_{T,\alpha}$ by $\BG_{T}$ and the covariance term $\varepsilon_{H,q}$ and can be handled similarly as~(\ref{handle_epsilz}) to obtain $\varepsilon_{\Gamma}=\BO(\frac{\mathcal{P}(\frac{1}{z})}{z^2})$. This proves~(\ref{GAABC}).

Similarly, by utilizing the integration by parts formula on $X_{i,q}^{*}$, we can obtain
\begin{equation}
\label{QZZQAB}
\begin{aligned}
&
\E[\BZ^{H}\BQ\BZ\BZ^{H} \BQ\bold{A}\bold{X}\bold{B}]_{q,q}
=\E\sum_{i}[\LR \BQ\BZ\BZ^{H}\BQ\bold{A}\bold{X}\bold{B}]_{i,q}[\MS\BX^{H} ]_{q,i}
\\
&
=\E\{-\frac{1}{L}\Tr\FR\BQ [\FS \BY\RT\BH^{H} \BQ\BZ\BZ^{H} \BQ\bold{A}\bold{X}\bold{B}]_{q,q}
-\frac{1}{L}\Tr\FR\BQ\BZ\BZ^{H} \BQ [\FS \BY\RT\BH^{H} \BQ\bold{A}\bold{X}\bold{B}]_{q,q}
\\
&
+\frac{1}{L}\Tr\FR\BQ [\FS\BZ^{H}\BQ\bold{A}\BX\bold{B}]_{q,q}
+\frac{1}{L} \Tr\bold{A}\LR\BQ\BZ\BZ^{H}\BQ [\MS\bold{B}]_{q,q}
\}
\\
&\overset{(b)}{=}\E\{ \alpha_{\delta}\chi[\bold{Y}\bold{C}\FT\bold{G}_{T,\alpha} \BH^{H}\BQ\bold{A}\bold{X}\bold{B}\FS\BF_{S,\alpha}]_{q,q}-\frac{M\zeta }{L} [\FS\BF_{S,\alpha} \BY\RT\BH^{H} \BQ\bold{A}\bold{X}\bold{B}]_{q,q}
\\
&
+\alpha_{\delta}[\BZ^{H}\BQ\bold{A}\BX\bold{B}\FS\BF_{S,\alpha}]_{q,q}+\frac{M}{L}\zeta(\LR,\RT,\BA\LR) [\MS\bold{B}\BF_{S,\alpha}]_{q,q}\}+\varepsilon_{S,q},
\end{aligned}
\end{equation}
where 
\begin{equation}
\label{QZZAB}
\begin{aligned}
&\varepsilon_{S,q}=-\cov(\frac{\Tr\FR\BQ}{L},[\FS\BF_{S,\alpha} \BY\RT\BH^{H} \BQ\BZ\BZ^{H} \BQ\bold{A}\bold{X}\bold{B}]_{q,q})
-\cov(\frac{\Tr\BQ\FR\BQ\BZ\BZ^{H}}{L},[\FS\BF_{S,\alpha} \BY\RT\BH^{H} \BQ\bold{A}\bold{X}\bold{B}]_{q,q})
\\
&
+\cov(\frac{1}{M} \Tr\BQ\BZ\BZ^{H} \BQ\BZ\BZ^{H},  [\bold{Y}\bold{C}\FT\bold{G}_{T,\alpha} \BH^{H}\BQ\bold{A}\bold{X}\bold{B}\FS\BF_{S,\alpha}]_{q,q})+\cov(\frac{\Tr\FR\BQ}{M},[\BZ^{H}\BQ\bold{A}\BX\bold{B}\FS\BF_{S,\alpha}]_{q,q})
\\
&
+\cov(\frac{\Tr\BQ\BZ\BZ^{H}\BQ\BA\LR}{L},[\MS\bold{B}\BF_{S,\alpha}]_{q,q})
+\cov(\frac{\Tr\BZ\BZ^{H}\BQ}{M},[\bold{Y}\bold{C}\FT\BG_{T,\alpha} \BH^{H}\BQ\bold{A}\bold{X}\bold{B}\FS\BF_{S,\alpha}]_{q,q}).
\end{aligned}
\end{equation}
Step $(b)$ in~(\ref{QZZQAB}) can be obtained by: 1. letting $\BC$ be $\RT$ and $\BB$ be $\BB\FS$ in~(\ref{QZZQABCH}). 2. Replacing $\E[\FS \BY\RT\BH^{H} \BQ\BZ\BZ^{H} \BQ\bold{A}\bold{X}\bold{B}]_{q,q}$ in~(\ref{QZZQAB}) and multiplying both sides by $[\BG_{S,\alpha}]_{q,q}$. According to~(\ref{QZZAB}), and by replacing $\BG_{T,\alpha}$ and $\BF_{S,\alpha}$ with $\BG_{T}$ and $\BF_{S}$ respectively, $\chi(\BA,\BB)$ can be represented by
\begin{equation}
\label{chiAB}
\begin{aligned}
&\chi(\BA,\BB)=\frac{\chi\nu_{T}\delta \Tr\BA\LR\BG_{R}}{L}\frac{\Tr\FS^{\frac{3}{2}}\BB\BF_{S}^2}{M}
+\frac{\delta\Tr\BA\LR\BG_{R}}{L}\frac{\Tr\BB\FS^{\frac{3}{2}}\BF_{S}^2}{M}
\\
&
+\frac{M\zeta(\LR,\MS,\BA\LR)}{L}\frac{\Tr\MS\bold{B}\BF_{S}}{M}
-\frac{M\zeta\overline{\omega}}{L}\frac{\Tr\BA\LR\BG_{R}}{L}\frac{\Tr\BB\FS^{\frac{3}{2}}\BF_{S}^2}{M}
+\varepsilon_{Z,Z}+\varepsilon_{S},
\end{aligned}
\end{equation}
where $\varepsilon_{S}=\frac{1}{M}\sum_{q}\varepsilon_{S,q}=\BO(\frac{1}{N^2})$, which follows from the variance control. $\varepsilon_{Z,Z}$ originates from the computation of $\mu(\BA,\BB,\BC)$ and the result of substituting $\BG_{R,\alpha}$, $\BG_{S,\alpha}$, $\BG_{T,\alpha}$ by $\BG_{R}$, $\BG_{S}$, $\BG_{T}$, respectively, which is also of the order $\BO(\frac{1}{N^2})$ by Theorem~\ref{QA_conv}. Specially, by letting $\BA=\LR$ and $\BB=\MS$ in 
(\ref{chiAB}), $\chi$ can be further expressed as
\begin{equation}
\begin{aligned}
\label{chi}
\chi=\frac{1}{1-\nu_{S}\nu_{T}}(\nu_{S}+\frac{M\nu_{S,I}\zeta }{L\delta^2} )+\BO_{z}(\frac{1}{N^2}).
\end{aligned}
\end{equation}
\subsection{$\Upsilon(\BA,\BB,\BC,\BM)$ and $\zeta(\BA,\BB,\BM)$}
By the analysis in Appendix~\ref{gamma_chi}, we also have the following relation
\begin{equation}
\label{QMQABC}
\begin{aligned}
&
\E[\bold{Y}\bold{C}\RT\bold{Y}^{H}\MS\bold{X}^{H}\LR\BQ\BM\BQ\bold{A}\bold{X}\bold{B}]_{q,q}
=\E\{
-\frac{ 1}{M} \Tr\BQ\BM\BQ\BZ\BZ^{H}  [\bold{Y}\bold{C}\FT\BG_{T,\alpha} \BH^{H}\BQ\bold{A}\bold{X}\bold{B}]_{q,q}
\\
&+\frac{1}{M}\Tr\bold{C}\RT\BG_{T,\alpha}[\BZ^{H}\BQ\BM\BQ\bold{A}\bold{X}\bold{B}]_{q,q}\}
+\varepsilon_{RH,q},
\end{aligned}
\end{equation}
where $\varepsilon_{RH,q}=-\cov(\frac{\Tr\BZ\BZ^{H}\BQ}{M},[\bold{Y}\bold{C}\FT\BG_{T,\alpha} \BH^{H}\BQ\BM\BQ\bold{A}\bold{X}\bold{B}]_{q,q})$. By summing over $q$ and Lemma~\ref{ka_mu}, we can obtain
\begin{equation}
\Upsilon(\BA,\BB,\BC,\BM)=\frac{\Tr\BC\RT\BG_{T}}{M}\zeta(\BA,\BB,\BM)-\frac{\Tr\BA\LR\BG_{R}}{L}\frac{\Tr\BB\MS\BF_{S}}{M}\frac{\Tr\BC\FT^{\frac{3}{2}}\BG^2_{T}}{M}\zeta(\LR,\MS,\BM)+\varepsilon_{\Upsilon},
\end{equation}
where $\varepsilon_{\Upsilon}=\BO(\frac{1}{N^2})$, which comes from the substitution of $\BG_{T,\alpha}$ by $\BG_{T}$ and the covariance terms $\varepsilon_{RH,q}$. This concludes the proof of~(\ref{UPSABC}) by letting $\BM=\FR$.

Similarly, we also have
\begin{equation}
\begin{aligned}
\label{QQABX}
&\E[\BZ^{H}\BQ\BM\BQ\bold{A}\bold{X}\bold{B}]_{i,i}
=\E\{\frac{1}{L}\Tr\bold{A}\LR\BQ\BM\BQ  [\MS\bold{B}]_{q,q}-\frac{1}{L}\Tr\FR\BQ [\FS\BY\RT\BH^{H}\BQ\BM\BQ\bold{A}\bold{X}\bold{B}]_{q,q}
\\
&-\frac{1}{L}\Tr\BQ\bold{M}\BQ\FR [\FS\BY\RT\BH^{H}\BQ\bold{A}\bold{X}\bold{B}]_{q,q}\}
\\
&
\overset{(c)}{=}\E\{\frac{1}{L}\Tr\bold{A}\LR\BQ\BM\BQ[\MS\BB\BF_{S,\alpha}]_{q,q}+\frac{1}{L}\Tr\FR\BQ\frac{1}{M}\Tr\BQ\BM\BQ\BZ\BZ^{H}[\bold{Y}\FT^{\frac{3}{2}}\BG_{T,\alpha} \BH^{H}\BQ\bold{A}\bold{X}\bold{B}\FS\BF_{S,\alpha}]_{q,q}
\\
&-\frac{1}{L}\Tr\BQ\BM\BQ\FR [\FS\BY\RT\BH^{H}\BQ\bold{A}\bold{X}\bold{B}\BF_{S,\alpha}]_{q,q}\}+\varepsilon_{S,q},
\end{aligned}
\end{equation}
where 
\begin{equation}
\begin{aligned}
&\varepsilon_{ZR,q}=-\cov(\frac{1}{L}\Tr\FR\BQ,[\FS\BY\RT\BH^{H}\BQ\BM\BQ\bold{A}\bold{X}\bold{B}\BF_{S,\alpha}]_{q,q})
\\
&+\cov(\frac{\Tr\BZ\BZ^{H}\BQ}{M},[\bold{Y}\bold{C}\FT\BG_{T,\alpha} \BH^{H}\BQ\BM\BQ\bold{A}\bold{X}\bold{B}\FS\BF_{S,\alpha}]_{q,q}).
\end{aligned}
\end{equation}
Step (c) in~(\ref{QQABX}) can be obtained by replacing $\E[\FS\BY\RT\BH^{H}\BQ\bold{A}\bold{X}\bold{B}]_{q,q}$ in~(\ref{QQABX}) with~(\ref{QMQABC}). Therefore, we have
\begin{equation}
\label{ZEBAM}
\begin{aligned}
&\zeta(\BA,\BB,\BM)=\frac{\Tr\MS\BB\BF_{S}}{ M}\frac{\E\Tr\bold{A}\LR\BQ\BM\BQ}{L}
+ \frac{\delta\zeta(\LR,\MS,\BM)\nu_{T}\Tr\BA\LR\BG_{R}}{L}\frac{\Tr\BB\FS^{\frac{3}{2}}\BF_{S}^2}{ M}
\\
&-\frac{\E\Tr\BQ\BM\BQ\FR}{L}\frac{\Tr\BA\LR\BG_{R}}{L}\frac{\overline{\omega}\Tr\BB\FS^{\frac{3}{2}}\BF_{S}^2}{ M}
+\BO_{z}(\frac{1}{N^2}).
\end{aligned}
\end{equation}
Letting $\BA=\LR$ and $\BB=\MS$ in~(\ref{ZEBAM}), we can obtain
\begin{equation}
\label{QZRSM}
\zeta(\LR,\MS,\BM)=\frac{\nu_{S,I}\E\Tr\BQ\BM\BQ\FR}{L\delta^2(1-\nu_{S}\nu_{T})}+\BO_{z}(\frac{1}{N^2}).
\end{equation}
By plugging~(\ref{QZRSM}) into~(\ref{ZEBAM}) and substituting $\BM$ by $\FR$, we can rewrite~(\ref{ZEBAM}) as
\begin{equation}
\label{ZEBAM1}
\begin{aligned}
&\zeta(\BA,\BB)=\frac{\Tr\MS\BB\BG_{S}}{\delta M}\frac{\E\Tr\bold{A}\LR\BQ\FR\BQ}{L}
- \frac{\nu_{T,I}\Tr\BA\LR\BG_{R}}{L\delta^2\Delta_{S}}\frac{\Tr\BB\FS^{\frac{3}{2}}\BG_{S}^2}{M}\frac{\E\Tr\BQ\FR\BQ\FR}{L}
+\BO_{z}(\frac{1}{N^2}),
\end{aligned}
\end{equation}
It follows from~(\ref{QZRSM}) and~(\ref{ZEBAM1}) that, if we can obtain the evaluation of $\frac{1}{L}\E\Tr\BQ\FR\BQ\BM$, we will be able to finish all the evaluations in Lemma~\ref{sec_app}. Next, we will evaluate $\frac{1}{L}\E\Tr\BQ\FR\BQ\BM$.

\subsection{The Evaluation of $\frac{\E\Tr\BQ\FR\BQ\BM}{L}$}
By the identity $\BA-\BB=\BB(\BB^{-1}-\BA^{-1})\BA$, we have
\begin{equation}
\label{QMQR_ex}
\begin{aligned}
&\frac{1}{L}\E\Tr\BM\BQ\FR\BQ=\frac{1}{L}\E\Tr\BM\BQ\FR(\BQ-\BG_{R})+\frac{1}{L}\E\Tr\BM\BQ\FR\BG_{R}
=
\frac{1}{L}\E\Tr\BM\BQ\FR\BQ_{R}+
\frac{M\omega\overline{\omega}}{L^2\delta}\E\Tr\BM\BQ\FR\BG_{R}\FR\BQ
\\
&
-\frac{1}{L}\E\Tr\BM\BQ\FR\BG_{R}\BH\BH^{H}\BQ
=\frac{\Tr\FR\BG_{R}\bold{M}\BG_{R}}{L}+R_{1}+R_{2}+\BO_{z}(\frac{1}{N^2}).
\end{aligned}
\end{equation}
By~(\ref{QMQABC}), we have
\begin{equation}
\label{X2exp}
\begin{aligned}
&R_{2}=-\frac{M}{L}\Upsilon(\BG_{R}\FR^{\frac{3}{2}},\MS,\RT,\BM)
=-\frac{M \omega\overline{\omega}\E\Tr\BQ\BM\BQ\FR^2\BG_{R}}{L^2\delta}
-\frac{M\overline{\omega}\nu_{T}\nu_{S}\nu_{S,I}\E\Tr\BQ\BM\BQ\FR }{L\delta^2\Delta_{S} }
\\
&+\frac{M\overline{\omega}\nu_{S}\E\Tr\BQ\BM\BQ\FR}{L^2\delta}
+\frac{M\nu_{S}\nu_{S,I}\nu_{T}\omega\E\Tr\BQ\BM\BQ\FR}{L\delta^2\Delta_{S}}
+\BO_{z}(\frac{1}{N^2})
\\
&=-\frac{M \omega\overline{\omega}\E\Tr\BQ\BM\BQ\FR^2\BG_{R}}{L^2\delta}
+(1-\Delta)\frac{\E\Tr\BQ\BM\BQ\FR}{L}+\BO_{z}(\frac{1}{N^2})
=S_{1}+S_{2}+\BO_{z}(\frac{1}{N^2}).
\end{aligned}
\end{equation}
Plugging~(\ref{X2exp}) into~(\ref{QMQR_ex}) and noticing $R_{1}+S_{1}=0$, we can obtain
\begin{equation}
\label{QMQR}
\frac{\E\Tr\BQ\BM\BQ\FR}{L}=\frac{\Tr\FR\BG_{R}\bold{M}\BG_{R}}{L\Delta}+\BO_{z}(\frac{1}{N^2}).
\end{equation}

\subsection{Return to evaluate to $\zeta(\BA,\BB)$ and $\chi(\BA,\BB)$}
According to~(\ref{chiAB}),~(\ref{chi}),~(\ref{QZRSM}), and~(\ref{QMQR}), we can obtain the deterministic approximation for $\chi(\BA,\BB)$ as
\begin{equation}
\begin{aligned}
&\chi(\BA,\BB)=
\frac{\Tr\BA\LR\BG_{R}}{L}\frac{\Tr\FS^{\frac{3}{2}}\BB\BG_{S}^2}{M\delta}[\frac{\nu_{T}\nu_{S}}{\Delta_{S}}
+\frac{M\nu_{T}\nu_{S,I}^2\nu_{R}}{L\delta^4\Delta_{S}^2\Delta}+1-\frac{M\overline{\omega}\nu_{S,I}\nu_{R}}{L\delta^3\Delta_{S}\Delta}]
+\frac{M\nu_{S,I}\Tr\BG_{R}^2\FR^{\frac{3}{2}}\BA}{L^2\delta^3\Delta_{S}\Delta}\frac{\Tr\MS\bold{B}\BG_{S}}{M}
\\
&
=\frac{\Tr\BA\LR\BG_{R}}{L}\frac{\Tr\FS^{\frac{3}{2}}\BB\BG_{S}^2}{M\delta}[-1+\frac{1}{\Delta_{S}}
+\frac{M\overline{\omega}\nu_{S,I}\nu_{R}}{L\delta^3\Delta_{S}\Delta}-\frac{M\nu_{S,I}\nu_{T,I}\nu_{R}}{L\delta^3\Delta_{S}^2\Delta} +1-\frac{M\overline{\omega}\nu_{S,I}\nu_{R}}{L\delta^3\Delta_{S}\Delta}]
\\
&
+\frac{M\nu_{S,I}\Tr\BG_{R}^2\FR^{\frac{3}{2}}\BA}{L^2\delta^3\Delta_{S}\Delta}\frac{\Tr\MS\bold{B}\BG_{S}}{M}+\BO_{z}(\frac{1}{N^2})
\\
&=\frac{\Tr\BA\LR\BG_{R}}{L}\frac{\Tr\FS^{\frac{3}{2}}\BB\BG_{S}^2}{M\delta}\frac{1-\frac{M \omega\overline{\omega}\nu_{R}}{L\delta^2}}{\Delta_{S}\Delta}
+\frac{M\nu_{S,I}\Tr\BG_{R}^2\FR^{\frac{3}{2}}\BA}{L^2\delta^3\Delta_{S}\Delta}\frac{\Tr\MS\bold{B}\BG_{S}}{M}
+\BO_{z}(\frac{1}{N^2}),
\end{aligned}
\end{equation}
which proves~(\ref{CHAB}). Furthermore, it follows from~(\ref{ZEBAM}),~(\ref{QZRSM}) and~(\ref{QMQR}) that
\begin{equation}
\label{AMBQQ}
\begin{aligned}
&\zeta(\BA,\BB)=\frac{\Tr\BA\FR^{\frac{3}{2}}\BG_{R}^2}{L\Delta} \frac{\Tr\MS\BB\BG_{S}}{\delta M}
-\frac{\nu_{R}\nu_{T,I}\Tr\BA\LR\BG_{R}}{L\delta^2\Delta_{S}\Delta}\frac{\Tr\BB\FS^{\frac{3}{2}}\BG_{S}^2}{M}
+\BO_{z}(\frac{1}{N^2}),
\end{aligned}
\end{equation}
which concludes the proof of~(\ref{ZEAB}). $\zeta$ can be obtained by letting $\BM=\FR$ in~(\ref{QZRSM}) and~(\ref{QMQR}) and $\chi$ can be obtained according to~(\ref{chi}). This concludes the proof of~(\ref{zeta_chi_sim}). 
\end{proof}

\section{Proof of Lemma~\ref{lem_ZZQP}}
\label{pro_lem_ZZQP}
\begin{proof}
By substituting the product of $s_{q}$ and~(\ref{HHQP}) to replace the last term in~(\ref{QZZP}), subtracting~(\ref{ZZQ}),
 and summing over $q$ and $i$, we can obtain
\begin{equation}
\label{zzq_1}
\begin{aligned}
&\E \underline{\Tr\BZ\BZ^{H} \BQ} \Phi(u,z)=\E \Tr\BZ\BZ^{H} \BQ\Phi(u,z)-\E\Tr\BZ\BZ^{H}\BQ \E\Phi(u,z)
\\
&
\overset{(a)}{=}\E[\frac{1}{L}\Tr\FS\Tr\FR\BQ -\frac{\underline{\Tr\FR\BQ}}{L}\Tr \BZ\FS\BY\RT\BH^{H}\BQ-\frac{(\E\Tr\FR\BQ)}{L}\Tr \BZ\FS\BY\RT\BH^{H}\BQ  
\\
&+\frac{\jmath u \Tr \BZ \FS\BY\RT \BH^{H} \BQ\FR\BQ}{L}
  ]  \Phi(u,z)-\Tr\FS\BF_{S} \E\Tr\FR\BQ \E \Phi(u,z)+\BO_{z}(\frac{1}{N})
\\
&\overset{(b)}{=}
\E[\frac{1}{L}\Tr\FS\Tr\FR\BQ -\frac{M \delta\overline{\omega}\frac{\Tr\FS^2\BF_{S}}{M} }{L}\underline{\Tr\FR\BQ}-\delta(\frac{\overline{\omega}}{L}\Tr\FS^2\BF_{S}\Tr\FR\BQ -\frac{\delta\overline{\omega}^2}{L}\Tr\FS^3\BF_{S}\underline{\Tr\FR\BQ}
\\
&
+\frac{\jmath u \overline{\omega}}{L} \Tr \BZ\FS^2\BF_{S}\BY\RT\BH^{H}\BQ\FR\BQ
-\frac{\delta}{M^2} \Tr\FS^2\BF_{S}^2\Tr \FT^2\BG_{T}^2\underline{\Tr\BZ\BZ^{H}\BQ}
\\
&
+\frac{\jmath u  }{M} \Tr\BZ\FS\BF_{S}\BY\FT^{\frac{3}{2}}\BG_{T}\BH^{H}\BQ\BZ\BZ^{H} \BQ
   )
   +\frac{\jmath u}{L}\Tr\BZ \FS \BY\RT \BH^{H} \BQ\FR\BQ
  ]\Phi(t,z)
   \\
   &
   -\frac{1}{L}\Tr\FS\BF_{S}  \E\Tr\FR\BQ\E \Phi(u,z)+\BO_{z}(\frac{1}{N})
   \\
   &
= \nu_{S}\nu_{T} \E\underline{\Tr\BZ\BZ^{H}\BQ} \Phi(u,z)
+(\frac{M \omega }{L\delta}-\frac{M \nu_{S}}{L\delta} )\E\underline{\Tr\FR\BQ}\Phi(t,z)
+\frac{\jmath u   M}{L}\Upsilon(\LR,\FS^{\frac{3}{2}}\BF_{S},\RT)\E \Phi(u,z)
\\
&-\jmath u \delta  \Gamma(\LR,\FS^{\frac{3}{2}}\BF_{S},\FT^{\frac{3}{2}}\BG_{T} )\E\Phi(u,z)+\BO_{z}(\frac{1}{N}),
\end{aligned}
\end{equation}   
where the $\BO_{z}(\frac{1}{N})$ on the RHS of step $(a)$ follows from the approximation of $\E\Tr\BZ\BZ^{H}\BQ$ in the last term of the first line. Step $(b)$ follows from Lemma~\ref{ka_mu} by evaluating $\frac{1}{L}\E\Tr \BZ\FS\BY\RT\BH^{H}\BQ$. The evaluation of $\E \underline{\Tr\BZ\BZ^{H}\BQ} \Phi(u,z)$ can be obtained by moving $ \nu_{S}\nu_{T} \E\underline{\Tr\BZ\BZ^{H}\BQ} \Phi(u,z)$ to the LHS of~(\ref{zzq_1}).
\end{proof}

\section{Proof of Proposition~\ref{pro_rnk}}
\label{proof_pro}

\begin{proof} There are eight cases concerned. Case $1$ to Case $3$ include scenarios when $M$, $N$, $L$ are pairwise unequal. Case $4$ to Case $6$ represent the cases when only two of $M$, $N$, and $L$ are equal. Case $7$ considers $M=N=L$. We further summarize the $7$ cases above to the $3$ cases, labeled by $a$,$b$, and $c$ in~(\ref{appros_Ciid}) and~(\ref{appros_Viid}). In the following proof, $\overline{C}(\sigma^2)$ and $V(\sigma^2)$ refer to the i.i.d case.

\textit{Case 1: $\eta>1$ and $\kappa<1$}. ($M$ is the smallest.)

We can first obtain the high-SNR approximation for $\omega$. In this case, by observing the dominating terms in~(\ref{new_cubic_eq}), we know $(\eta-1)(\kappa-1)\rho\omega^2$, $(\eta\kappa-2\eta+1)\rho\omega$, and $-\eta\rho$ are negative terms since the coefficients of these terms are negative. Meanwhile, these terms should be compensated by positive terms $\omega^3$, $\omega^2$, and $\omega$ so that the equality holds. First, $\omega$ is not $\BO(1)$, otherwise $\BO(\rho)$ terms will not be compensated. If $\omega$ has a higher order than $\rho$, i.e., $\omega=\Theta(\rho^{1+\varepsilon})$, the LHS of~(\ref{new_cubic_eq}) will grow to infinity since $\omega^3=\Theta(\rho^{3+3\varepsilon})$, $\varepsilon>0$, which can not be compensated by the negative terms. Therefore, $\omega^3$ is the highest-order positive term and should be compensated by the highest-order negative, i.e., $(\eta-1)(\kappa-1)\rho\omega^2$. Therefore, we have $\omega=\Theta(\rho)$ and $\omega=-(\eta-1)(\kappa-1)\rho+\BO(1)$. $\delta=(\eta\kappa-\kappa)\rho+\BO(1)$ can be obtained by the approximation of $\omega$ in~(\ref{delta_iid_exp}). This approach is also applicable for other cases and we will omit the detailed analysis. $\overline{C}(\rho^{-1})$ can be approximated by 
\begin{equation}
\label{C_neq_m}
\begin{aligned}
\overline{C}(\rho^{-1})&=M[\frac{ \log(\rho)}{\kappa}-( \eta-\frac{1}{\kappa})\log(1-\frac{1}{\eta})+(1-\frac{1}{\kappa})\log((\eta-1)(1-\kappa)\rho)-2+\frac{\log(\eta)}{\kappa}]+\BO(\rho^{-1})
\\
&=M[\log(\rho)+\log(\eta)+(1-\eta)\log(1-\frac{1}{\eta})+(1-\frac{1}{\kappa})\log(1-\kappa)-2 ]+\BO(\rho^{-1})
\\
&=M\log(\frac{\rho N}{e^2M})-(N-M)\log(1-\eta)-(L-M)\log(1-\kappa) +\BO(\rho^{-1}).
\end{aligned}
\end{equation}
The approximation for $V(\rho^{-1})$  can be obtained by plugging the approximations of $\omega$ and $\delta$ into~(\ref{V_iid_exp}) and is given by
\begin{equation}
V(\rho^{-1})=-\log((1-\kappa)(1-\frac{1}{\eta}))+\BO(\rho^{-1}).
\end{equation}
When $\eta<1$, i.e., $M>N$, $\omega$ should be $\BO(1)$, due to the constraint $(\eta-1)\omega+\eta>0$. The dominating term is $[(\eta-1)(\kappa-1)\omega^2+(\eta\kappa-2\eta+1)\omega-\eta ]\rho$ and $\omega$ can be obtained by letting the coefficient of $\rho$ be zero so that $\omega\in \{\frac{1}{\kappa-1},-\frac{\eta}{\eta-1}\}$.

\textit{Case 2: $\eta<1$ and $\eta\kappa<1$}. ($N$ is the smallest.)

In this case, $\omega=\frac{1}{\kappa-1}+o(1)$ is not feasible since $(\eta-1)\omega+\eta=\frac{\eta\kappa-1}{\kappa-1}<0$. Therefore, we have the approximations $\omega=-\frac{\eta}{\eta-1}-\frac{\eta\rho^{-1}}{(1-\eta)^3(1-\eta\kappa)}+\BO(\rho^{-2})$ and $\delta=(1-\eta)^{-1}(\frac{1}{\eta\kappa}-1)^{-1}+\BO(\rho^{-1})$. Then, we can obtain
\begin{equation}
\label{C_neq_n}
\begin{aligned}
\overline{C}(\rho^{-1})&=M[\eta\log(\rho)+(\eta-\frac{1}{\kappa})\log((1-\eta)(1-\eta\kappa))-\log(1-\eta)-\frac{\log(\eta)-\log(1-\eta)}{\kappa}-2\eta+\frac{\log(\eta)}{\kappa} ]+\BO(\rho^{-1})
\\
&=M[\eta\log(\rho)+(\eta-1)\log(1-\eta)+(\eta-\frac{1}{\kappa})\log(1-\eta\kappa)  -2\eta ]+\BO(\rho^{-1}),
\\
&=N\log(\frac{\rho}{e^2  }) -(M-N)\log(1-\eta)-(L-N)\log(1-\eta\kappa)+\BO(\rho^{-1}),
\end{aligned}
\end{equation}
and 
\begin{equation}
V(\rho^{-1})=-\log((1-\eta)(1-\eta\kappa))+\BO(\rho^{-1}).
\end{equation}

\textit{Case 3: ($\eta>1$ and $\kappa>1$) or ($\eta<1$ and $\eta\kappa>1$)}. ($L$ is the smallest.)

We first consider $\eta>1$ and $\kappa>1$, i.e., $N> M$ and $M>L$. In this case, we have the approximations $\omega=\frac{1}{\kappa-1}+\BO(\rho^{-1})$ and $\delta=(\eta\kappa-1)\rho+\BO(\rho^{-1})$. The approximations of $\overline{C}(\rho^{-1})$ and $V(\rho^{-1})$ are given  by
\begin{equation}
\label{C_neq_l}
\begin{aligned}
\overline{C}(\rho^{-1})&=M[\frac{\log(\rho)}{\kappa} -  (\eta-\frac{1}{\kappa})\log(1-\frac{1}{\eta\kappa})+\log(\kappa)-\log(\kappa-1)+\frac{\log(\kappa-1)}{\kappa}-\frac{2}{\kappa}+\frac{\log(\eta)}{\kappa}  ] +\BO(\rho^{-1}) 
\\
&=M[\frac{\log(\rho)}{\kappa}-  (\eta-\frac{1}{\kappa})\log(1-\frac{1}{\eta\kappa})+\frac{\log(\eta\kappa)}{\kappa}-(1-\frac{1}{\kappa})\log(1-\frac{1}{\kappa})-\frac{2}{\kappa} ]+\BO(\rho^{-1})
\\
&
=L\log(\frac{\rho N}{e^2 L})-(N-L)\log(1-\frac{1}{\eta\kappa})-(M-L)\log(1-\frac{1}{\kappa}) +\BO(\rho^{-1}),.
\end{aligned}
\end{equation}
and 
\begin{equation}
\begin{aligned}
V(\rho^{-1})=-\log((1-\frac{1}{\kappa})(1-\frac{1}{\eta\kappa}))+\BO(\rho^{-1}).
\end{aligned}
\end{equation}
Now we consider $\eta<1$ and $\eta\kappa>1$, i.e., $M>N$ and $N>L$. By the analysis before~\textit{Case 2}, we have $\omega=\BO(1)$. If $\omega=-\frac{\eta}{\eta-1}$, then $\delta=(1-\eta)^{-1}(\frac{1}{\eta\kappa}-1)^{-1}+o(1)<0$ for large $\rho$, which is impossible. Therefore, $\omega=\frac{1}{\kappa-1}+\BO(\rho^{-1})$ and the result for this case coincides with that with $\eta>1$ and $\kappa>1$.

\textit{Case 4: $\eta=1$ and $\kappa \neq 1$.} In this case,~(\ref{new_cubic_eq}) becomes $\omega^3+2\omega^2+(1+\rho\kappa-\rho)\omega-\rho=0$. The approximations of $\omega$ and $\delta$ are given by
\begin{equation}
\omega=
\begin{cases}
\sqrt{(1-\kappa)\rho }+[\frac{1}{2(1-\kappa)}-1]+\BO(\rho^{-\frac{1}{2}}),~~\kappa<1,
\\
\frac{1}{\kappa-1}+\BO(\rho^{-1}),~~\kappa>1,
\end{cases}
\end{equation}
\begin{equation}
\delta=
\begin{cases}
\frac{\kappa\rho^{\frac{1}{2}}}{\sqrt{1-\kappa}} -\frac{\kappa}{2(1-\kappa)^{2}} +\BO(\rho^{-\frac{1}{2}}) ,~~\kappa<1,
\\
(\kappa-1)\rho +\BO(1),~~\kappa>1.
\end{cases}
\end{equation}
The approximations for $\overline{C}(\rho^{-1})$ are given by
\begin{equation}
\label{eta_equal}
\begin{aligned}
\overline{C}(\rho^{-1}) &=M[\frac{\log(\rho)}{\kappa} +\frac{2-\frac{1}{\kappa}}{2} \log((1-\kappa)\rho)
+\frac{(2-\frac{1}{\kappa})\rho^{-\frac{1}{2}}}{2(1-\kappa)^{\frac{3}{2}}}
-\frac{\log((1-\kappa)\rho)}{2\kappa}-\frac{(2\kappa-1)\rho^{-\frac{1}{2}}}{2\kappa(1-\kappa)^{\frac{3}{2}}}
-2 + \frac{2\rho^{-\frac{1}{2}}}{(1-\kappa)^{\frac{1}{2}}} ]+\BO(\rho^{-1})
\\
&=M\log(\frac{\rho}{e^{2}}) -(L-M)\log(1-\kappa)
+\frac{2M\rho^{-\frac{1}{2}}}{(1-\kappa)^{\frac{1}{2}}}+\BO(\rho^{-1}),
\end{aligned}
\end{equation}
and
\begin{equation}
\label{equal_C_l}
\begin{aligned}
\overline{C}(\rho^{-1})&=M[\frac{\log(\rho)}{\kappa} -(1-\frac{1}{\kappa})\log(1-\frac{1}{\kappa})
+\log(\frac{\kappa}{\kappa-1})+\frac{\log(\kappa-1)}{\kappa}-\frac{2}{\kappa}]+\BO(\rho^{-1})
\\
&=L\log(\frac{\rho\kappa}{e^2})-2(M-L)\log(1-\frac{1}{\kappa}) ]+\BO(\rho^{-1}),
\end{aligned}
\end{equation}
for $\kappa<1$ and $\kappa>1$, respectively. The approximation for the variance with $\kappa<1$ and $\kappa>1$ can be obtained by
\begin{equation}
\begin{aligned}
V(\rho^{-1})&=\frac{\log(\rho)}{2}-\log(\frac{2\kappa}{\sqrt{1-\kappa}})-\frac{(\kappa-1)\rho^{-\frac{1}{2}}}{(1-\kappa)^{\frac{3}{2}}}
+\log(\frac{\kappa}{1-\kappa})-\frac{\kappa\rho^{-\frac{1}{2}}}{(1-\kappa)^{\frac{3}{2}}}
+\BO(\rho^{-1})
\\
&
=\frac{1}{2}\log(\frac{\rho}{4(1-\kappa)})+\frac{(1-2\kappa)\rho^{-\frac{1}{2}}}{(1-\kappa)^{\frac{3}{2}}}+\BO(\rho^{-1}),
\end{aligned}
\end{equation}
and 
\begin{equation}
\begin{aligned}
V(\rho^{-1})&=-2\log(1-\frac{1}{\kappa})+\BO(\rho^{-1}).
\end{aligned}
\end{equation}

\textit{Case 5: $\kappa=1$ and $\eta \neq 1$.}

When $\kappa=1$,~(\ref{new_cubic_eq}) becomes $\omega^3+2\omega^2+(1-\rho\eta+\rho)\omega-\eta\rho=0$. In this case, the approximations of $\omega$ and $\delta$ are given by
\begin{equation}
\omega=
\begin{cases}
\sqrt{(\eta-1)\rho }+[\frac{1}{2(\eta-1)}-\frac{1}{2}]+\BO(\rho^{-\frac{1}{2}}),~~\eta>1,
\\
\frac{\eta}{1-\eta}-\frac{\eta\rho^{-1}}{(1-\eta)^4} +\BO(\rho^{-2}),~~\eta<1,
\end{cases}
\end{equation}
\begin{equation}
\delta=
\begin{cases}
(\eta-1)\rho+\frac{\rho^{\frac{1}{2}}}{\sqrt{\eta-1}}+\BO(1) ,~~\eta>1,
\\
\frac{\eta}{(1-\eta)^2} +\BO(\rho^{-1}),~~\eta<1.
\end{cases}
\end{equation}
The approximations for $\overline{C}(\rho^{-1})$ are given by
\begin{equation}
\label{kappa_equal}
\begin{aligned}
\overline{C}(\rho^{-1})&=M[{\log(\rho)} -  (\eta-1)\log(1-\frac{1}{\eta})- (\eta-1)\frac{\rho^{-\frac{1}{2}}}{(\eta-1)^{\frac{3}{2}}}
+ \frac{\rho^{-\frac{1}{2}}}{\sqrt{\eta-1}}+  \frac{2\rho^{-\frac{1}{2}}}{\sqrt{\eta-1}}-2+\log(\eta)+\BO(\rho^{-1})
 ]
\\
&=M\log(\frac{\rho\eta}{e^2}) -  (N-M)\log(1-\frac{1}{\eta})+  \frac{2M(\eta\rho)^{-\frac{1}{2}}}{\sqrt{1-\frac{1}{\eta}}}+\BO(\rho^{-1}) ],
\end{aligned}
\end{equation}
and
\begin{equation}
\label{equal_C_n}
\begin{aligned}
\overline{C}(\rho^{-1})&=M[{\log(\rho)} -  (\eta-1)\log( \frac{\rho^{-1}}{(1-\eta)^2}  )
 -\log(1-\eta)-\log(\frac{\eta}{1-\eta})- 2\eta +\log(\eta)
 +\BO(\rho^{-1})]
\\
&=N\log(\frac{\rho}{e^2}) -2(N-M)\log(1-\eta)  +\BO(\rho^{-1}),
\end{aligned}
\end{equation}
for $\eta>1$ and $\eta<1$, respectively. The approximation for the variance with $\eta>1$ and $\eta<1$ can be obtained by
\begin{equation}
\begin{aligned}
V(\rho^{-1})&=-\log(2(\eta-1))-\frac{3\rho^{-\frac{1}{2}}}{2(\eta-1)^{\frac{3}{2}}}+\frac{  \rho^{-\frac{1}{2}} }{2(\eta-1)^{\frac{1}{2}}}
+\frac{1}{2}\log(\eta-1)+\frac{\rho^{-\frac{1}{2}}}{2(\eta-1)^{\frac{3}{2}}}+\frac{\rho^{-\frac{1}{2}}}{2(\eta-1)^{\frac{1}{2}}}+\frac{\log(\rho)}{2}
+\log(\eta)
+\BO(\rho^{-1})
\\
&
=\frac{1}{2}\log(\frac{\rho\eta^2}{4(\eta-1)})+\frac{(\eta-2)\rho^{-\frac{1}{2}}}{(\eta-1)^{\frac{3}{2}}}+\BO(\rho^{-1})
\\
&
=\frac{1}{2}\log(\frac{\rho\eta}{4(1-\frac{1}{\eta})})+\frac{(1-\frac{2}{\eta})(\eta\rho)^{-\frac{1}{2}}}{(1-\frac{1}{\eta})^{\frac{3}{2}}}+\BO(\rho^{-1}),
\end{aligned}
\end{equation}
and 
\begin{equation}
V(\rho^{-1})=-2\log(1-\eta  )+\BO(\rho^{-1}),
\end{equation}
respectively.

\textit{Case 6: $\eta\kappa=1$ and $\eta \neq 1$.}

When $\eta\kappa=1$,~(\ref{new_cubic_eq}) becomes $\omega^3+[2+(2-\eta-\frac{1}{\eta})\rho]\omega^2+(1-2\rho\eta+2\rho)\omega-\rho\eta=0$. In this case, the approximations of $\omega$ and $\delta$ are given by
\begin{equation}
\omega=
\begin{cases}
(\eta+\frac{1}{\eta}-2)\rho
+\BO(1),~~\eta>1,
\\
\frac{\eta}{1-\eta}-\frac{\eta\rho^{-\frac{1}{2}}}{(1-\eta)^{\frac{5}{2}}}+\frac{\eta(2\eta+1)\rho^{-1}}{2(1-\eta)^{4}}  +  \BO(\rho^{-1}),~~\eta<1,
\end{cases}
\end{equation}
\begin{equation}
\delta=
\begin{cases}
(1-\frac{1}{\eta})\rho
+\BO(1),~~\eta>1,
\\
\frac{\rho^{\frac{1}{2}}}{(1-\eta)^{\frac{1}{2}}} 
-\frac{1}{2(1-\eta)^2}
 +  \BO(\rho^{-\frac{1}{2}}),~~\eta<1.
\end{cases}
\end{equation}
The approximations for $\overline{C}(\rho^{-1})$ are given by
\begin{equation}
\label{equal_C_m}
\begin{aligned}
\overline{C}(\rho^{-1})&=M[\eta\log(\rho)+ (1-\eta)\log(\rho)+(1-\eta)\log(\eta+\frac{1}{\eta}-2)-2
+\eta\log(\eta)
 +\BO(\rho^{-1})]
\\
&=M[ \log(\frac{\rho}{e^2}) +(1-\eta)\log(\eta+\frac{1}{\eta}-2) +\eta\log(\eta)  ]+\BO(\rho^{-1})
\\
&=M\log(\frac{\rho\eta}{e^2})-2(M-N)\log(1-\frac{1}{\eta})+\BO(\rho^{-1}),
\end{aligned}
\end{equation}
and 
\begin{equation}
\label{etakappa_equal}
\begin{aligned}
&\overline{C}(\rho^{-1})=M[\eta\log(\rho)- \log(1-\eta)-\frac{\eta\rho^{-\frac{1}{2}}}{(1-\eta)^{\frac{3}{2}}}
-\eta\log(\frac{\eta}{1-\eta})
+\frac{\eta\rho^{-\frac{1}{2}}}{(1-\eta)^{\frac{3}{2}}}
-2\eta+\frac{2\eta\rho^{-\frac{1}{2}}}{\sqrt{1-\eta}}
+\eta\log(\eta) +\BO(\rho^{-1})]
\\
&= N\log(\frac{\rho}{e^2}) -(N-M)\log(1-\eta)+\frac{2N\rho^{-\frac{1}{2}}}{\sqrt{1-\eta}}  +\BO(\rho^{-1}) ],
\end{aligned}
\end{equation}
for $\eta>1$ and $\eta<1$, respectively. In this case, the variances for $\eta>1$ and $\eta<1$ can be evaluated by
\begin{equation}
\begin{aligned}
V(\rho^{-1})&=-\log(1-\frac{1}{\eta})+\log(\frac{\eta}{\eta-1})
+\BO(\rho^{-1})=-2\log(1-\frac{1}{\eta})+\BO(\rho^{-1}).
\end{aligned}
\end{equation}
and 
\begin{equation}
\begin{aligned}
V(\rho^{-1})&=-\log(2(1-\eta))-\frac{(4\eta-1)\rho^{-\frac{1}{2}}}{2(1-\eta)^{\frac{3}{2}}}
+\frac{1}{2}\log(1-\eta)
+\frac{\log(\rho)}{2}
+\frac{\rho^{-\frac{1}{2}}}{2(1-\eta)^{\frac{3}{2}}}
 +\BO(\rho^{-1})
 \\
 &=\frac{1}{2}\log(\frac{\rho}{4(1-\eta)})+\frac{(1-2\eta)\rho^{-\frac{1}{2}}}{(1-\eta)^{\frac{3}{2}}} +\BO(\rho^{-1}).
 \end{aligned}
\end{equation}

\textit{Case 7: $\eta\kappa=1$ and $\eta=1$.} In this case,~(\ref{new_cubic_eq}) becomes $\omega^3+2\omega^2+\omega-\rho=0$ and the approximations of $\omega$ and $\delta$ are given by
\begin{equation}
\omega=
\rho^{\frac{1}{3}}-\frac{2}{3}+\BO(\rho^{-\frac{1}{3}}),
\end{equation}
\begin{equation}
\delta=\rho^{\frac{2}{3}}-\frac{\rho^{\frac{1}{3}}}{3}+\BO(1).
\end{equation}
The approximations for $\overline{C}(\rho^{-1})$ and $V(\rho^{-1})$ are given by
\begin{equation}
\label{equal_ccc}
\overline{C}(\rho^{-1})=M\log(\frac{\rho}{e^2})+3M\rho^{-\frac{1}{3}}+\BO(\rho^{-\frac{2}{3}}),
\end{equation}
and
\begin{equation}
V(\rho^{-1})=\frac{2\log(\rho)}{3}-\log(3)+\frac{4\rho^{-\frac{1}{3}}}{3}+\BO(\rho^{-\frac{2}{3}}).
\end{equation}

Given the ordered version $(S_{N},S_{L},S_{M})$ of $N,L,M$, we can summarize the results on $\overline{C}(\rho^{-1})$ in~(\ref{C_neq_m}),~(\ref{C_neq_n}),~(\ref{C_neq_l}),~(\ref{equal_C_n}), ~(\ref{equal_C_l}), and~(\ref{equal_C_m}) for the case $S_L \neq S_M$ as
\begin{equation}
\label{C_sum}
\overline{C}(\rho^{-1})=S_{M}\log(\frac{\rho N}{e^2 S_{M} })-(S_{N}-S_{M})\log(1-\frac{S_M}{S_N})-(S_L-S_M)\log(1-\frac{S_{M}}{S_{L}}) +\BO(\rho^{-1}),
\end{equation}
which concludes~(\ref{c_appro_case_1}). When $S_{L}=S_{M}$ and $S_{N}\neq S_{M}$,~(\ref{eta_equal}),~(\ref{kappa_equal}) and~(\ref{etakappa_equal}) can be summarized as 
\begin{equation}
\overline{C}(\rho^{-1})=S_M\log(\frac{\rho N }{S_{M} e^{2}}) -(S_N-S_M)\log(1-\frac{S_{M}}{S_{N}})
+\frac{2S_{M}  (\frac{N \rho}{S_M})^{-\frac{1}{2}}}{(1-\frac{S_{M}}{S_{N}})^{\frac{1}{2}}}+\BO(\rho^{-1}),
\end{equation}
which concludes~(\ref{c_appro_case_2}). The result for the case with $S_{L}=S_{M}=S_{N}$ is given in~(\ref{equal_ccc}), which concludes~(\ref{c_appro_case_3}). The approximations for the variances can also be summarized similarly to conclude~(\ref{v_appro_case_1}) to~(\ref{v_appro_case_3}).
\end{proof}

\ifCLASSOPTIONcaptionsoff
  \newpage
\fi



\bibliographystyle{IEEEtran}
\bibliography{IEEEabrv,ref}

\begin{thebibliography}{10}
\providecommand{\url}[1]{#1}
\csname url@samestyle\endcsname
\providecommand{\newblock}{\relax}
\providecommand{\bibinfo}[2]{#2}
\providecommand{\BIBentrySTDinterwordspacing}{\spaceskip=0pt\relax}
\providecommand{\BIBentryALTinterwordstretchfactor}{4}
\providecommand{\BIBentryALTinterwordspacing}{\spaceskip=\fontdimen2\font plus
\BIBentryALTinterwordstretchfactor\fontdimen3\font minus
  \fontdimen4\font\relax}
\providecommand{\BIBforeignlanguage}[2]{{%
\expandafter\ifx\csname l@#1\endcsname\relax
\typeout{** WARNING: IEEEtran.bst: No hyphenation pattern has been}%
\typeout{** loaded for the language `#1'. Using the pattern for}%
\typeout{** the default language instead.}%
\else
\language=\csname l@#1\endcsname
\fi
#2}}
\providecommand{\BIBdecl}{\relax}
\BIBdecl

\bibitem{gesbert2002outdoor}
D.~Gesbert, H.~Bolcskei, D.~A. Gore, and A.~J. Paulraj, ``Outdoor {MIMO}
  wireless channels: Models and performance prediction,'' \emph{{IEEE} Trans.
  Commun.}, vol.~50, no.~12, pp. 1926--1934, Dec. 2002.

\bibitem{almers2006keyhole}
P.~Almers, F.~Tufvesson, and A.~F. Molisch, ``Keyhole effect in {MIMO} wireless
  channels: {Measurements} and theory,'' \emph{{IEEE} Trans. Wireless Commun.},
  vol.~5, no.~12, pp. 3596--3604, Dec. 2006.

\bibitem{basar2021indoor}
E.~Basar, I.~Yildirim, and F.~Kilinc, ``Indoor and outdoor physical channel
  modeling and efficient positioning for reconfigurable intelligent surfaces in
  mmwave bands,'' \emph{{IEEE} Trans. Commun.}, vol.~69, no.~12, pp.
  8600--8611, Dec. 2021.

\bibitem{shin2008mimo}
H.~Shin and M.~Z. Win, ``{MIMO} diversity in the presence of double
  scattering,'' \emph{{IEEE} Trans. Inf. Theory}, vol.~54, no.~7, pp.
  2976--2996, Jul. 2008.

\bibitem{shin2003capacity}
H.~Shin and J.~H. Lee, ``Capacity of multiple-antenna fading channels: Spatial
  fading correlation, double scattering, and keyhole,'' \emph{{IEEE} Trans.
  Inf. Theory}, vol.~49, no.~10, pp. 2636--2647, Oct. 2003.

\bibitem{shi2021outage}
Z.~Shi, H.~Wang, Y.~Fu, G.~Yang, S.~Ma, and F.~Gao, ``Outage analysis of
  reconfigurable intelligent surface aided {MIMO} communications with
  statistical {CSI},'' \emph{{IEEE} Trans. Wireless Commun.}, vol.~21, no.~2,
  pp. 823--839, Feb. 2022.

\bibitem{hachem2008new}
W.~Hachem, O.~Khorunzhiy, P.~Loubaton, J.~Najim, and L.~Pastur, ``A new
  approach for mutual information analysis of large dimensional multi-antenna
  channels,'' \emph{{IEEE} Trans. Inf. Theory}, vol.~54, no.~9, pp. 3987--4004,
  Sep. 2008.

\bibitem{hachem2008clt}
W.~Hachem, P.~Loubaton, and J.~Najim, ``A {CLT} for information-theoretic
  statistics of gram random matrices with a given variance profile,'' \emph{The
  Annals of Applied Probability}, vol.~18, no.~6, pp. 2071--2130, 2008.

\bibitem{hachem2012clt}
W.~Hachem, M.~Kharouf, J.~Najim, and J.~W. Silverstein, ``A {CLT} for
  information-theoretic statistics of non-centered {G}ram random matrices,''
  \emph{Random Matrices: Theory. Appl.}, vol.~1, no.~2, p. 1150010, Dec. 2012.

\bibitem{bao2015asymptotic}
Z.~Bao, G.~Pan, and W.~Zhou, ``Asymptotic mutual information statistics of
  {MIMO} channels and {CLT} of sample covariance matrices,'' \emph{{IEEE}
  Trans. Inf. Theory}, vol.~61, no.~6, pp. 3413--3426, Jun. 2015.

\bibitem{hu2019central}
J.~Hu, W.~Li, and W.~Zhou, ``Central limit theorem for mutual information of
  large {MIMO} systems with elliptically correlated channels,'' \emph{{IEEE}
  Trans. Inf. Theory}, vol.~65, no.~11, pp. 7168--7180, Nov. 2019.

\bibitem{hoydis2011asymptotic}
J.~Hoydis, R.~Couillet, and M.~Debbah, ``Asymptotic analysis of
  double-scattering channels,'' in \emph{Proc. Conf. Rec. 45th Asilomar Conf.
  Signals, Syst. Comput. (ASILOMAR)}, Pacific Grove, CA, USA, Nov. 2011, pp.
  1935--1939.

\bibitem{muller2002random}
R.~R. Muller, ``A random matrix model of communication via antenna arrays,''
  \emph{{IEEE} Trans. Inf. Theory}, vol.~48, no.~9, pp. 2495--2506, Sep. 2002.

\bibitem{hoydis2011iterative}
J.~Hoydis, R.~Couillet, and M.~Debbah, ``Iterative deterministic equivalents
  for the performance analysis of communication systems,'' \emph{arXiv preprint
  arXiv:1112.4167}, Dec. 2011.

\bibitem{zhang2021large}
J.~Zhang, J.~Liu, S.~Ma, C.-K. Wen, and S.~Jin, ``Large system achievable rate
  analysis of {RIS}-assisted {MIMO} wireless communication with statistical
  {CSIT},'' \emph{{IEEE} Trans. Wireless Commun.}, vol.~20, no.~9, pp.
  5572--5585, Sept. 2021.

\bibitem{moustakas2021capacity}
A.~L. Moustakas, G.~C. Alexandropoulos, and M.~Debbah, ``Capacity optimization
  using reconfigurable intelligent surfaces: {A} large system approach,'' in
  \emph{Proc. {IEEE} Global Commun. Conf. (GLOBECOM)}, Madrid, Spain, Dec,
  2020, pp. 01--06.

\bibitem{ye2022sum}
J.~Ye, Q.-U.-A. Nadeem, A.~Kammoun, and M.-S. Alouini, ``Sum-rate analysis of a
  multi-cell multi-user {MISO} system under double scattering channels,''
  \emph{{IEEE} Trans. Commun.}, vol.~70, no.~1, pp. 332--349, Jan. 2022.

\bibitem{ye2020asymptotic}
------, ``Asymptotic analysis of {MRT} over double scattering channels with
  {MMSE} estimation,'' \emph{{IEEE} Trans. Wireless Commun.}, vol.~19, no.~12,
  pp. 7851--7863, Dec. 2020.

\bibitem{kammoun2019asymptotic}
A.~Kammoun, M.~Debbah, M.-S. Alouini \emph{et~al.}, ``Asymptotic analysis of
  {RZF} over double scattering channels with {MMSE} estimation,'' \emph{{IEEE}
  Trans. Wireless Commun.}, vol.~18, no.~5, pp. 2509--2526, May. 2019.

\bibitem{zheng2016asymptotic}
Z.~Zheng, L.~Wei, R.~Speicher, R.~R. M{\"u}ller, J.~H{\"a}m{\"a}l{\"a}inen, and
  J.~Corander, ``Asymptotic analysis of {Rayleigh} product channels: {A} free
  probability approach,'' \emph{{IEEE} Trans. Inf. Theory}, vol.~63, no.~3, pp.
  1731--1745, Mar. 2016.

\bibitem{zhang2022outage}
X.~Zhang, X.~Yu, and S.~Song, ``Outage probability and finite-{SNR} {DMT}
  analysis for {IRS}-aided {MIMO} systems: How large {IRSs} need to be?''
  \emph{{IEEE} J. Sel. Topics Signal Process.}, vol.~16, no.~5, pp. 1070--1085,
  Aug. 2022.

\bibitem{dumont2010capacity}
J.~Dumont, W.~Hachem, S.~Lasaulce, P.~Loubaton, and J.~Najim, ``On the capacity
  achieving covariance matrix for rician {MIMO} channels: an asymptotic
  approach,'' \emph{{IEEE} Trans. Inf. Theory}, vol.~56, no.~3, pp. 1048--1069,
  Mar. 2010.

\bibitem{dupuy2011capacity}
F.~Dupuy and P.~Loubaton, ``On the capacity achieving covariance matrix for
  frequency selective {MIMO} channels using the asymptotic approach,''
  \emph{{IEEE} Trans. Inf. Theory}, vol.~57, no.~9, pp. 5737--5753, Sep. 2011.

\bibitem{zhang2013capacity}
J.~Zhang, C.-K. Wen, S.~Jin, X.~Gao, and K.-K. Wong, ``On capacity of
  large-scale {MIMO} multiple access channels with distributed sets of
  correlated antennas,'' \emph{{IEEE} J. Sel. Areas Commun.}, vol.~31, no.~2,
  pp. 133--148, Jan. 2013.

\bibitem{asgharimoghaddam2020capacity}
H.~Asgharimoghaddam, J.~Kaleva, and A.~T{\"o}lli, ``Capacity approaching low
  density spreading in uplink {NOMA} via asymptotic analysis,'' \emph{{IEEE}
  Trans. Commun.}, vol.~69, no.~3, pp. 1635--1649, Mar. 2020.

\bibitem{artigue2011precoder}
C.~Artigue and P.~Loubaton, ``On the precoder design of flat fading {MIMO}
  systems equipped with {MMSE} receivers: {A} large-system approach,''
  \emph{{IEEE} Trans. Inf. Theory}, vol.~57, no.~7, pp. 4138--4155, Jul. 2011.

\bibitem{hoydis2015second}
J.~Hoydis, R.~Couillet, and P.~Piantanida, ``The second-order coding rate of
  the {MIMO} quasi-static rayleigh fading channel,'' \emph{{IEEE} Trans. Inf.
  Theory}, vol.~61, no.~12, pp. 6591--6622, Dec. 2015.

\bibitem{zhang2022second}
X.~Zhang and S.~Song, ``Second-order coding rate of quasi-static
  rayleigh-product {MIMO} channels,'' \emph{arXiv preprint arXiv:2210.08832},
  2022.

\bibitem{lytova2009central}
A.~Lytova and L.~Pastur, ``Central limit theorem for linear eigenvalue
  statistics of random matrices with independent entries,'' \emph{Ann.
  Probab.}, vol.~37, no.~5, pp. 1778--1840, 2009.

\bibitem{gotze2017distribution}
F.~G{\"o}tze, A.~Naumov, and A.~Tikhomirov, ``Distribution of linear statistics
  of singular values of the product of random matrices,'' \emph{Bernoulli},
  vol.~23, no.~4B, pp. 3067--3113, 2017.

\bibitem{telatar_capacity}
E.~Telatar, ``Capacity of multi-antenna gaussian channels,'' \emph{Eur. Trans.
  Telecommun.}, vol.~10, no.~6, pp. 585--595, 1999.

\bibitem{levin2006multi}
G.~Levin and S.~Loyka, ``Multi-keyhole mimo channels: Asymptotic analysis of
  outage capacity,'' in \emph{Proc. {IEEE} Int. Symp. Inf. Theory.
  (ISIT)}.\hskip 1em plus 0.5em minus 0.4em\relax Seattle, WA, USA: IEEE, Jul.
  2006, pp. 1305--1309.

\bibitem{zhi2022two}
K.~Zhi, C.~Pan, H.~Ren, K.~Wang, M.~Elkashlan, M.~Di~Renzo, R.~Schober, H.~V.
  Poor, J.~Wang, and L.~Hanzo, ``Two-timescale design for reconfigurable
  intelligent surface-aided massive {MIMO} systems with imperfect {CSI},''
  \emph{{IEEE} Trans. Inf. Theory}, To appear. 2022.

\bibitem{zhao2020intelligent}
M.-M. Zhao, Q.~Wu, M.-J. Zhao, and R.~Zhang, ``Intelligent reflecting surface
  enhanced wireless networks: Two-timescale beamforming optimization,''
  \emph{{IEEE} Trans. Wireless Commun.}, vol.~20, no.~1, pp. 2--17, Jan. 2021.

\bibitem{zhang2021bias}
X.~Zhang and S.~Song, ``Bias for the trace of the resolvent and its application
  on non-gaussian and non-centered {MIMO} channels,'' \emph{{IEEE} Trans. Inf.
  Theory}, vol.~68, no.~5, pp. 2857--2876, May. 2022.

\bibitem{couillet2011deterministic}
R.~Couillet, M.~Debbah, and J.~W. Silverstein, ``A deterministic equivalent for
  the analysis of correlated {MIMO} multiple access channels,'' \emph{{IEEE}
  Trans. Inf. Theory}, vol.~57, no.~6, pp. 3493--3514, Jun. 2011.

\bibitem{zx2021large}
X.~Zhang, X.~Yu, S.~Song, and K.~B. Letaief, ``{IRS}-aided {MIMO} systems over
  double-scattering channels: {Impact} of channel rank deficiency,'' in
  \emph{Proc. {IEEE} Wireless Commun. Netw. Conf. (WCNC)}, Austin, TX, USA,
  April. 2022, pp. 2076--2081.

\bibitem{zhang2021outageext}
X.~Zhang, X.~Yu, and S.~Song, ``Outage probability and finite-{SNR} {DMT}
  analysis for {IRS}-aided {MIMO} systems: {H}ow large {IRSs} need to be?''
  \emph{arXiv preprint arXiv:2111.15123}, May. 2022.

\bibitem{zheng2003diversity}
L.~Zheng and D.~N.~C. Tse, ``Diversity and multiplexing: {A} fundamental
  tradeoff in multiple-antenna channels,'' \emph{{IEEE} Trans. Inf. Theory},
  vol.~49, no.~5, pp. 1073--1096, May. 2003.

\bibitem{loyka2010finite}
S.~Loyka and G.~Levin, ``Finite-{SNR} diversity-multiplexing tradeoff via
  asymptotic analysis of large {MIMO} systems,'' \emph{{IEEE} Trans. Inf.
  Theory}, vol.~56, no.~10, pp. 4781--4792, Oct. 2010.

\bibitem{kamath2005asymptotic}
M.~A. Kamath and B.~L. Hughes, ``The asymptotic capacity of multiple-antenna
  {Rayleigh}-fading channels,'' \emph{{IEEE} Trans. Inf. Theory}, vol.~51,
  no.~12, pp. 4325--4333, Dec. 2005.

\bibitem{yang2011diversity}
S.~Yang and J.-C. Belfiore, ``Diversity-multiplexing tradeoff of double
  scattering {MIMO} channels,'' \emph{{IEEE} Trans. Inf. Theory}, vol.~57,
  no.~4, pp. 2027--2034, Apr. 2011.

\bibitem{pastur2005simple}
L.~A. Pastur, ``A simple approach to the global regime of {Gaussian} ensembles
  of random matrices,'' \emph{Ukrainian Mathematical Journal}, vol.~57, no.~6,
  pp. 936--966, 2005.

\bibitem{rubio2012clt}
F.~Rubio, X.~Mestre, and W.~Hachem, ``A {CLT} on the {SNR} of diagonally loaded
  {MVDR} filters,'' \emph{{IEEE} Trans. Signal Process.}, vol.~60, no.~8, pp.
  4178--4195, Aug. 2012.

\bibitem{petersen2008matrix}
K.~B. Petersen, M.~S. Pedersen \emph{et~al.}, ``The matrix cookbook,''
  \emph{Technical University of Denmark}, vol.~7, no.~15, p. 510, 2008.

\bibitem{kammoun2019asymptoticit}
A.~Kammoun, L.~Sanguinetti, M.~Debbah, and M.-S. Alouini, ``Asymptotic analysis
  of {RZF} in large-scale {MU-MIMO} systems over {Rician} channels,''
  \emph{{IEEE} Trans. Inf. Theory}, vol.~65, no.~11, pp. 7268--7286, Nov. 2019.

\bibitem{kammoun2009central}
A.~Kammoun, M.~Kharouf, W.~Hachem, and J.~Najim, ``A central limit theorem for
  the sinr at the lmmse estimator output for large-dimensional signals,''
  \emph{{IEEE} Trans. Inf. Theory}, vol.~55, no.~11, pp. 5048--5063, Nov. 2009.

\bibitem{billingsley2008probability}
P.~Billingsley, \emph{Probability and measure}.\hskip 1em plus 0.5em minus
  0.4em\relax John Wiley \& Sons, 2008.

\bibitem{ertel1998overview}
R.~B. Ertel, P.~Cardieri, K.~W. Sowerby, T.~S. Rappaport, and J.~H. Reed,
  ``Overview of spatial channel models for antenna array communication
  systems,'' \emph{IEEE personal communications}, vol.~5, no.~1, pp. 10--22,
  1998.

\bibitem{yin2013coordinated}
H.~Yin, D.~Gesbert, M.~Filippou, and Y.~Liu, ``A coordinated approach to
  channel estimation in large-scale multiple-antenna systems,'' \emph{{IEEE} J.
  Sel. Areas Commun.}, vol.~31, no.~2, pp. 264--273, 2013.

\bibitem{massey1951kolmogorov}
F.~J. Massey~Jr, ``The kolmogorov-smirnov test for goodness of fit,''
  \emph{Journal of the American statistical Association}, vol.~46, no. 253, pp.
  68--78, 1951.

\bibitem{shapiro1965analysis}
S.~S. Shapiro and M.~B. Wilk, ``An analysis of variance test for normality
  (complete samples),'' \emph{Biometrika}, vol.~52, no. 3/4, pp. 591--611,
  1965.

\bibitem{besson2000approximate}
O.~Besson, F.~Vincent, P.~Stoica, and A.~B. Gershman, ``Approximate maximum
  likelihood estimators for array processing in multiplicative noise
  environments,'' \emph{{IEEE} Trans. Signal Process.}, vol.~48, no.~9, pp.
  2506--2518, Sep. 2000.

\bibitem{rubio2011spectral}
F.~Rubio and X.~Mestre, ``Spectral convergence for a general class of random
  matrices,'' \emph{Statistics \& probability letters}, vol.~81, no.~5, pp.
  592--602, 2011.

\end{thebibliography}
\end{document}